\documentclass[a4paper,11pt,oneside,reqno,final]{amsart}
\usepackage{graphicx}
\usepackage[latin1]{inputenc}
\usepackage{geometry}
\usepackage{latexsym}
\usepackage{amstext} \usepackage{amsmath} \usepackage{amssymb}
\usepackage{psfrag}
\usepackage{graphics}
\usepackage{color}

\numberwithin{equation}{section}

 2
 3
 4
 2

\newtheorem{theorem}{Theorem}
\newtheorem{remark}{Remark}
\newtheorem{lemma}{Lemma}
\newtheorem{prop}{Proposition}
\begin{document}

\title[Asymptotics of a cubic sine kernel determinant]{Asymptotics of a cubic sine kernel determinant}
\author{Thomas Bothner}
\address{Department of Mathematical Sciences,
Indiana University-Purdue University Indianapolis,
402 N. Blackford St., Indianapolis, IN 46202, U.S.A.}
\email{tbothner@iupui.edu}

\author{Alexander Its}
\address{Department of Mathematical Sciences,
Indiana University-Purdue University Indianapolis,
402 N. Blackford St., Indianapolis, IN 46202, U.S.A.}
\email{itsa@math.iupui.edu}

\thanks{This work was supported in part by the National Science Foundation (NSF) Grant DMS-1001777.}
\dedicatory{To the memory of Vladimir Savelievich Buslaev}
\date{\today}

\begin{abstract}
We study the one parameter family of Fredholm determinants $\det(I-\gamma K_{\textnormal{csin}}),\gamma\in\mathbb{R}$ of an integrable Fredholm operator $K_{\textnormal{csin}}$ acting on the interval $(-s,s)$ whose kernel is a cubic generalization of the sine kernel which appears in random matrix theory. This Fredholm determinant appears  in the description of the Fermi distribution of semiclassical non-equilibrium Fermi states in condensed matter physics as well as in random matrix theory. Using the Riemann-Hilbert method, we calculate the large $s$-asymptotics of $\det\left(I-\gamma K_{\textnormal{csin}}\right)$ for all values of the real parameter $\gamma$.

\end{abstract}

\maketitle

\section{Introduction}\label{sec1}
Let us consider the real vector space $\mathcal{M}(n)\simeq\mathbb{R}^{n^2}$ of $n\times n$ Hermitian matrices $M=(M_{ij})=\overline{M}^t$ equipped with the probability distribution
\begin{equation}\label{UEmodel}
	P^{(n,N)}(M)dM = ce^{-N\textnormal{tr}V(M)}dM,\hspace{1cm} c\int\limits_{\mathcal{M}(n)}e^{-N\textnormal{tr}V(M)}dM=1.
\end{equation}
Here $dM$ denotes the Haar measure on $\mathcal{M}(n)$, $N$ is a fixed integer and the potential $V:\mathbb{R}\rightarrow\mathbb{R}$ is assumed to be real analytic satisfying the growth condition
\begin{equation}\label{potregul}
	\frac{V(x)}{\ln(x^2+1)}\rightarrow\infty\hspace{0.5cm}\ \textnormal{as}\ |x|\rightarrow\infty.
\end{equation}
The model $\big(\mathcal{M}(n),P^{(n,N)}dM\big)$ is commonly referred to as unitary matrix model and it is well known (cf. \cite{D,M}) that basic statistical quantities such as
\begin{equation*}
	E_{n,N}(s) = \textnormal{Prob}\big(M\in\mathcal{M}(n)\ \textnormal{has no eigenvalues in the interval}\ (-s,s),\ s>0\big)
\end{equation*}
can be expressed as Fredholm determinants, indeed
\begin{eqnarray*}
	E_{n,N}(s) &=& \sum_{j=0}^n\frac{(-1)^j}{j!}\int\limits_{-s}^s\cdots\int\limits_{-s}^s\det\big(K_{n,N}(x_k,x_l)_{k,l=1}^j\big)dx_1\cdots dx_j\\
	&\equiv& \det(I-K_{n,N})
\end{eqnarray*}
where $K_{n,N}$ is the finite rank operator with kernel
\begin{equation}\label{correlkernel}
	K_{n,N}(x,y)=e^{-\frac{N}{2}V(x)}e^{-\frac{N}{2}V(y)}\sum_{i=0}^{n-1}p_i(x)p_i(y),\hspace{0.5cm}\int\limits_{\mathbb{R}}p_i(x)p_j(x)e^{-NV(x)}dx=\delta_{ij}
\end{equation}
acting on $L^2\big((-s,s),d\lambda\big)$.\smallskip

The given assumption \eqref{potregul} enables one to study scaling limits $n,N\rightarrow\infty$ (compare \cite{D},\cite {DKM}). First the mean eigenvalue density $\frac{1}{n}K_{n,N}(x,x)$ has a limit, that is for all bounded continuous functions $g:\mathbb{R}\rightarrow\mathbb{R}$
\begin{equation}\label{meandens}
	\lim_{\substack{n,N\rightarrow\infty\\
	\frac{n}{N}\rightarrow 1}}\frac{1}{n}\int g(x)K_{n,N}(x,x)dx = \int\limits_{\Sigma_V}g(x)\rho_V(x)dx\equiv\int\limits_{\Sigma_V}g(x)d\mu_V(x)
\end{equation}
and the support $\Sigma_V$ of the equilibrium measure $\mu_V$ consists of a finite union of intervals. Here the density $\rho_V(x)$ is determined by the potential $V(x)$. At the same time, the local statistics of eigenvalues in the large $n,N$ limit satisfies the so-called {\it universality property}, i.e. it is determined only by the local characteristics of the eigenvalue density $\rho_V$  (see \cite{PS, BI1, CK, DKMVZ}): For instance let us choose a regular point $x^{\ast}\in\Sigma_V$, i.e. $\rho_V(x^{\ast})>0$. Then the \emph{bulk universality} states that
\begin{equation}\label{bulkregular}
	\lim_{n\rightarrow\infty}\frac{1}{n\rho_V(x^{\ast})}K_{n,n}\bigg(x^{\ast}+\frac{\lambda}{n\rho_V(x^{\ast})},x^{\ast}+\frac{\mu}{n\rho_V(x^{\ast})}\bigg) = K_{\sin}(\lambda,\mu)\equiv \frac{\sin\pi(\lambda-\mu)}{\pi(\lambda-\mu)}
\end{equation}
uniformly on compact subsets of $\mathbb{R}$, or in other words for a regular point $x^{\ast}$
\begin{equation*}
	\lim_{n\rightarrow\infty}\textnormal{Prob}\bigg(M\in\mathcal{M}(n)\ \textnormal{has no eigenvalues}\in \Big(x^{\ast}-\frac{s}{n\rho_V(x^{\ast})},x^{\ast}+\frac{s}{n\rho_V(x^{\ast})}\Big)\bigg)
\end{equation*}
\begin{equation}\label{dyson1}
	=\det(I-K_{\sin})
\end{equation}
where $K_{\sin}$ is the trace class operator on $L^2\big((-s,s);d\lambda\big)$ with kernel $K_{\sin}(\lambda,\mu)$ given in \eqref{bulkregular}. The latter equation shows that in double scaling limits the basic statistical properties of hermitian random matrices are still expressible in terms of Fredholm determinants. On the other hand, the Fredholm determinant in the right hand side of \eqref{dyson1} admits the following asymptotic expansion \cite{dyson} 
\begin{equation}\label{dyson2}
	\ln\det(I-K_{\sin}) = -\frac{(\pi s)^2}{2}-\frac{1}{4}\ln(\pi s)+\frac{1}{12}\ln 2+3\zeta'(-1)+O\big(s^{-1}\big),\hspace{0.5cm}s\rightarrow\infty,
\end{equation}
where $\zeta(z)$ is the Riemann zeta-function. This formula yields one of the most important results in random matrix theory, i.e. an explicit evaluation of {\it the large gap probability}, and its rigorous derivation was obtained in a series of papers \cite{DIKZ,E,K}.\smallskip

The focus of the current paper lies on the computation of similar expansions for a certain family of generalized sine kernels. Let $K_{\textnormal{csin}}$ denote the trace class operator on $L^2\big((-s,s);d\lambda\big)$ with kernel
\begin{equation}\label{lambda3kernel}
	K_{\textnormal{csin}}(\lambda,\mu) = \frac{\sin\big(\frac{4}{3}(\lambda^3-\mu^3)+x(\lambda-\mu)\big)}{\pi(\lambda-\mu)},\ \ x\in\mathbb{R}
\end{equation}
and $\det(I- K_{\textnormal{csin}})$ the corresponding Fredholm determinant. This Fredholm determinant appeared in \cite{BoI1} in the asymptotical analysis of the Fredholm determinant
\begin{equation*}
	\det\left(I-K_{\textnormal{PII}}\right)
\end{equation*}
which corresponds to the first bulk critical universality class in the unitary matrix model: Let us assume $\rho_V(x)$ vanishes quadratically at an interior point $x^{\ast}\in\Sigma_V$. Opposed to \eqref{bulkregular}, the scaling limit is more complicated \cite{BI2,CK}. Let $\rho_V(x^{\ast})=\rho'_V(x^{\ast})=0,\rho''_V(x^{\ast})>0$ and $n,N\rightarrow\infty$ such that
\begin{equation*}
	\lim_{n,N\rightarrow\infty}n^{2/3}\bigg(\frac{n}{N}-1\bigg)=C
\end{equation*}
exists with $C\in\mathbb{R}$. Then the \emph{critical bulk universality} guarantees existence of positive constants $c$ and $c_1$ such that
\begin{equation}\label{bulkcritical}
	\lim_{n,N\rightarrow\infty}\frac{1}{cn^{1/3}}K_{n,N}\bigg(x^{\ast}+\frac{\lambda}{cn^{1/3}},x^{\ast}+\frac{\mu}{cn^{1/3}}\bigg)=K_{\textnormal{PII}}(\lambda,\mu;x)
\end{equation}
uniformly on compact subsets of $\mathbb{R}$ where the variable $x$ is the scaling parameter defined by the relation
\begin{equation*}
	\lim_{n,N\rightarrow\infty}n^{2/3}\bigg(\frac{n}{N}-1\bigg)=xc_1.
\end{equation*}
Here the limiting kernel $K_{\textnormal{PII}}(\lambda,\mu;x)$ is constructed out of the $\Psi$-function associated with a special solution of the second Painlev\'e equation, the Hastings-McLeod solution $u=u(x)$, see \cite{BoI1} for a precise definition of the Painlev\'e II kernel. The main result in \cite{BoI1} is the following analogue of the Dyson formula \eqref{dyson2} for the Painlev\'e II determinant. As $s\rightarrow\infty$
\begin{equation*}
	\ln \det(I-K_{\textnormal{PII}})=-\frac{2}{3}s^6-s^4x-\frac{1}{2}(sx)^2 - \frac{3}{4}\ln s + \int\limits_x^{\infty}(y-x)u^2(y)dy
\end{equation*}
\begin{equation}\label{theo1err}
	 -\frac{1}{6}\ln 2+3\zeta'(-1)+O\big(s^{-1}\big),
\end{equation}
and the error term in \eqref{theo1err} is uniform on any compact subset of the set 
\begin{equation}\label{exceptset0}
	\left\{x\in\mathbb{R}:\ -\infty<x<\infty\right\}.
\end{equation}
The proof of the latter expansion in \cite{BoI1} is based on a Riemann-Hilbert approach and an approximation argument which allows to derive the constant - with respect to $s$ - term
\begin{equation*}
	c_0\equiv c_0(x) = -\ln F_{TW}(x)-\frac{1}{6}\ln 2+3\zeta'(-1)
\end{equation*}
where $F_{TW}(x)$ is the celebrated Tracy-Widom distribution function,
\begin{equation*}
	F_{TW}(x) = e^{-\int_x^{\infty}(y-x)u^2(y)dy}.
\end{equation*}
In more detail, the $x$ -dependence of the constant term follows from the analysis of the
relevant differential identities. These identities though can not produce the numerical part of the
constant. This part follows from the fact that  in the large positive $x$-limit
\begin{equation*}
	K_{\textnormal{PII}}(\lambda,\mu;x)=K_{\textnormal{csin}}(\lambda,\mu;x)\left(1+O\Big(x^{-1/4}e^{-\frac{2}{3}x^{3/2}}\Big)\right),\hspace{0.5cm}x\rightarrow+\infty,\ \lambda,\mu\in(-s,s).
\end{equation*}
The last estimate, in conjunction with the Riemann-Hilbert analysis for the family $\check{K}_{\textnormal{csin}}$ with kernel
\begin{equation*}
	\check{K}_{\textnormal{csin}}(\lambda,\mu;s) = \frac{\sin\left(\frac{4}{3}t(\lambda^3-\mu^3)+x(\lambda-\mu)\right)}{\pi(\lambda-\mu)},\hspace{0.5cm} t\in[0,1]
\end{equation*}
leads to the constant term in \eqref{theo1err}. As a by-product of this analysis, \cite{BoI1} states, besides \eqref{theo1err}, the asymptotic relation,
\begin{equation}\label{theo2err}
\ln \det(I-K_{\textnormal{csin}})=-\frac{2}{3}s^6-s^4x-\frac{1}{2}(sx)^2 - \frac{3}{4}\ln s 
 -\frac{1}{6}\ln 2+3\zeta'(-1)+O\big(s^{-1}\big),
\end{equation}
valid as $s\rightarrow\infty$ and the error term in \eqref{theo2err} is uniform on any compact subset of the set \eqref{exceptset0}.\smallskip

In order to describe other spectral properties of large Hermitian matrices, one needs to study the Fredholm determinant
\begin{equation}\label{parafamily}
	\det\left(I-\gamma K_{\textnormal{PII}}\right)
\end{equation}
for the values of $\gamma$ which are different from $\gamma=1$. Such one-parameter families of determinants already appear in connection with the sine-kernel determinant, for instance in the famous Montgomery-Odlyzko conjecture \cite{Mo,O} concerning the zeros of the Riemann zeta-function, in the description of the emptiness formation probability and other correlation functions in one-dimensional impenetrable Bose gas \cite{IIK,IIKV1,IIKV2} as well as in a number of other important mathematical and theoretical physics applications.\smallskip

The analytical challenge of the determinants \eqref{parafamily} is once again the large $s$ asymptotics. In the case of the sine-kernel determinats, the result is well known (see  \cite{suleiman, MT1, MT2, BT, BB} and \cite{DIKZ} for more on the history of the question)
\begin{enumerate}
	\item[1.] As $s\rightarrow\infty$,
	\begin{equation*}
		\ln\det\left(I-\gamma K_{\sin}\right)=-4\kappa\pi s + 2\kappa^2\ln\pi s+\chi_{\sin} +O\big(s^{-1}\big)
	\end{equation*}
	uniformly on and compact subset of the set $\left\{\gamma\in\mathbb{R}:\ -\infty<\gamma<1\right\}$, where
	\begin{equation*}
		\kappa\equiv \kappa(\gamma) = -\frac{1}{2\pi}\ln(1-\gamma).
	\end{equation*}
	The constant $\chi_{\sin}\equiv \chi_{\sin}(\gamma)$ is given by the equation,
	\begin{equation}\label{chidef}
	\chi_{\sin} = 2\kappa^2 + 4\kappa^2\ln 2 - 4\int_{0}^{\gamma}\kappa(t)\frac{d}{dt}\arg\Gamma(i\kappa(t))dt,
	\end{equation}
	which was obtained by A. Budylin and V. Buslaev in \cite{BB} as a corollary of their main result in \cite{BB} -
	the asymptotics of the resolvent of the kernel $\gamma K_{\sin}(\lambda, \mu)$. Formula (\ref{chidef}) also 
	follows from the general theorem of E. Basor and H. Widom concerning the determinants of Toeplitz integral operators
	with piecewise continuous symbols \cite{BW}  . 
	\item[2.] For $\gamma$ chosen from any compact subset of the set $\left\{\gamma\in\mathbb{R}:\ 1<\gamma<\infty\right\}$, the Fredholm determinant $\det\left(I-\gamma K_{\sin}\right)$ has infinitely many zeros $\left\{s_n\right\}$ which accumulate at infinity (see \cite{suleiman, MT1, MT2}).
\end{enumerate}
The main results of the present paper are the following analogues for the cubic sine-kernel \eqref{lambda3kernel}, which together with \eqref{theo2err} state the large $s$ behavior of $\det\left(I-\gamma K_{\textnormal{csin}}\right)$ for all values of the parameter $\gamma$.
\begin{theorem}\label{theo1}
Let $K_{\textnormal{csin}}$ denote the trace class operator on $L^2\big((-s,s);d\lambda\big)$ with kernel as in \eqref{lambda3kernel}. As $s\rightarrow\infty$ 
\begin{equation}\label{theo1result}
	\ln\det(I-\gamma K_{\textnormal{csin}}) = -\kappa\bigg(\frac{16}{3}s^3+4xs\bigg)+6\kappa^2\ln s-\int\limits_x^{\infty}(y-x)u^2(y)dy +\chi +O\big(s^{-1}\big)
\end{equation}
uniformly on any compact subset of the set
\begin{equation}\label{excset1}
	\{(\gamma,x)\in\mathbb{R}^2:\ -\infty<\gamma<1,\ -\infty<x<\infty\},
\end{equation}
where
\begin{equation*}
		\kappa \equiv -i\nu(\gamma)= -\frac{1}{2\pi}\ln(1-\gamma),
\end{equation*}
\begin{equation*}
	\chi =2\kappa^2 + 8\kappa^2\ln 2 - 4\int_{0}^{\gamma}\kappa(t)\frac{d}{dt}\arg\Gamma(i\kappa(t))dt,
\end{equation*}
with the Euler gamma-function $\Gamma(z)$ and $u=u(x,\gamma)$ denotes the real-valued Ablowitz-Segur solution of the second Painlev\'e equation $u_{xx}=xu+2u^3$ corresponding to the monodromy surface
\begin{equation*}
		 \mathbb{M}=\{ (s_1,\ldots,s_6)|\ s_1=-i\gamma,s_2=0,s_3=\bar{s}_1,s_{n+3}=-s_n\}.
\end{equation*}
\end{theorem}
The Ablowitz-Segur solution $u=u(x,\gamma)$ \cite{SA} used in the statement of Theorem \ref{theo1} is given as the unique solution to the boundary value problem
\begin{equation}\label{Absegureq}
	u_{xx}=xu+2u^3,\hspace{0.5cm}u(x)\sim\gamma\textnormal{Ai}(x),\hspace{0.5cm}x\rightarrow+\infty,\ \gamma\neq 1.
\end{equation}
Such solutions are smooth in case $\gamma<1$, with exponentially fast decay as $x\rightarrow+\infty$ and oscillatory behavior as $x\rightarrow-\infty$, see e.g. \cite{FIKN}. On the other hand for $\gamma>1$, the solution to \eqref{Absegureq} has poles on the real axis, but is still pole-free for sufficiently large positive $x$, in fact \cite{Ber} for $(\gamma,x)$ chosen from any compact subset of the set
\begin{equation}\label{excset2}
	\left\{(\gamma,x)\in\mathbb{R}^2:\ 1<\gamma<\infty,\ x>\left(\frac{3}{2}\ln\gamma\right)^{2/3}\right\}
\end{equation}
the solution $u=u(x,\gamma)$ to \eqref{Absegureq} is pole-free. This in turn leads to our second asymptotic result.
\begin{theorem}\label{theo2}  For $(\gamma,x)$ chosen from any compact subset of the set \eqref{excset2},
the Fredholm determinant $\det\left(I-\gamma K_{\textnormal{csin}}\right)$ has infinitely many zeros $\{s_n\}$ with asymptotic distribution
\begin{equation}\label{theo2result}
	\frac{8}{3}s_n^3+2xs_n+\frac{1}{\pi}\ln(\gamma-1)\ln\big(16s_n^3+4xs_n\big)-\textnormal{arg}\frac{\Gamma(1-i\kappa)}{\Gamma(i\kappa)}\sim \frac{\pi}{2}+n\pi,\ \ n\rightarrow \infty.
\end{equation}
\end{theorem}
The proofs of Theorem \ref{theo1} and \ref{theo2} are based on a Riemann-Hilbert approach. This approach (compare \cite{IIKS,DIZ}) uses the integrable form of the Fredholm operator $K_{\textnormal{csin}}$, allowing us to connect the resolvent kernel to the solution of a Riemann-Hilbert problem. The latter can be analysed rigorously via the Deift-Zhou nonlinear steepest descent method.\bigskip

We should mention that a large class of the generalized  sine-kernel determinants  has already been considered in \cite{KKMST}
(see eq.(1.6) there). In the case $\gamma < 1$ and after a proper  re-scaling, the determinant
$\det\left(I-\gamma K_{\textnormal{csin}}\right)$  can be put in the form which is  very close to the
one treated in \cite{KKMST}. However, an essential difference occurs: the fast phase function, the function  $p(\lambda)$ 
in the notations of \cite{KKMST} (see eq.(1.7)),  which appear as a result of  the re-scaling, does not satisfy 
one of the key conditions of  \cite{KKMST}; moreover, it becomes depended on the large parameter. This means that
the results of  \cite{KKMST} are not directly applicable to our case. In fact, if one formally applies the main asymptotic
formula of  \cite{KKMST} to our case, then the first two terms of our asymptotic equation (\ref{theo1result}) are reproduced
while the constant (in $s$) term is not. Most significantly, the integral term with the Painlev\'e function does not show up. 
The result of our Theorem \ref{theo2}, i.e. the asymptotic distribution (\ref{theo2result}) of the zeros of
the cubic sine-kernel determinant  can not be compared with \cite{KKMST} since the techniques in \cite{KKMST}
do not cover  the situation $\gamma\geq1$.\smallskip

We also want to point out  that the analysis of the cubic sine-kernel determinant is also of interest without the random matrix theory background: the determinant $\det\left(I-K_{\textnormal{csin}}\right)$ appears in condensed matter physics \cite{BWieg}, namely in the description of the Fermi distribution of semiclassical non-equilibrium Fermi states. In order to understand perturbations to a degenerate Fermi gas one studies the one parameter extension of determinants corresponding to the kernel \eqref{lambda3kernel}, that is
\begin{equation*}
	\det\left(I-\gamma K_{\textnormal{csin}}\right),\ \gamma\in\mathbb{R}.
\end{equation*}
Our formula \eqref{theo1result} provides therefore an exact formula for the large perturbations.\smallskip

Finally we would like to emphasize the appearance in our asymptotic results of the general Ablowitz-Segur solution of the second Painlev\'e equation and not just the usual Hastings-McLeod solution.
This fact, in particular, poses  the following  intriguing  question which Theorem \ref{theo2} does not
answer: {\it What is the large $s$ asymptotic behavior of  $\det\left(I-\gamma K_{\textnormal{csin}}\right)$ when $\gamma > 1$ and $x$ coincides with one of the poles of the corresponding
Ablowitz-Segur second Painlev\'e transcendent ?}

\bigskip

Let us finish this introduction with a brief outline for the rest of the paper. Section \ref{sec2} starts with a short review of the Riemann-Hilbert approach for the asymptotics of integrable Fredholm operators. We then apply the general framework to the Fredholm determinant $\det\left(I-\gamma K_{\textnormal{csin}}\right)$ and formulate the associated ``master'' Riemann-Hilbert problem (RHP). We also state logarithmic $s,x$ and $\gamma$ derivatives of the determinant $\det\left(I-\gamma K_{\textnormal{csin}}\right)$ and outline a derivation of an integrable system whose tau-function is represented by $\det\left(I-\gamma K_{\textnormal{csin}}\right)$. In sections \ref{sec3}-\ref{sec15}, following the Deift-Zhou roadmap, we construct the asymptotic solution of the master Riemann-Hilbert problem. At this point we will have to distinguish between $\gamma<1$ and $\gamma>1$, in fact for $\gamma>1$ we are lead to a solitonic Riemann-Hilbert problem. Using an extra ``undressing'' step, we overcome the singularities in the Riemann-Hilbert problem and derive the large zero distribution as stated in \eqref{theo2result}. The situation is similar to the one dealt with in \cite{BoI2}. The calculations of sections \ref{sec16} and \ref{sec17} provide us with the asymptotics of $\ln\det\left(I-\gamma K_{\textnormal{csin}}\right)$ including the constant term for $\gamma<1$. The large zero distribution \eqref{theo2result} will be derived in section \ref{sec18}.

\bigskip 

The authors dedicate this paper to the memory of Vladimir Savelievich Buslaev, whose pioneering works
on the asymptotic analysis of integrable systems  lie in the foundation of  their modern asymptotic theory.
\section{Riemann-Hilbert approach - setup and review}\label{sec2}

The given integral kernel \eqref{lambda3kernel} belongs to an algebra of integrable operators first introduced in \cite{IIKS}: Let $\Gamma$ be an oriented contour in the complex plane $\mathbb{C}$ such as a Jordan curve. We are interested in operators of the form $\lambda I+K$ on $L^2(\Gamma)$, where $K$ denotes an integral operator with kernel
\begin{equation}\label{IIKStheo1}
	K(\lambda,\mu) = \frac{\sum_{i=1}^Mf_i(\lambda)h_i(\mu)}{\lambda-\mu},\hspace{0.5cm}\sum_{i=1}^Mf_i(\lambda)h_i(\lambda) = 0,\ \ M\in\mathbb{Z}_{\geq 1}
\end{equation}
with functions $f_i,h_i$ which are smooth up to the boundary of $\Gamma$. Given two operators $\lambda I+K,\check{\lambda}I+\check{K}$ of this type, the composition $(\lambda I+K)(\check{\lambda}I+\check{K})$ is again of the same form, hence we have a ring. Moreover let $K^t$ denote the real adjoint of $K$, i.e.
\begin{equation*}
	K^t(\lambda,\mu) = -\frac{\sum_{i=1}^Mh_i(\lambda)f_i(\mu)}{\lambda-\mu}.
\end{equation*}
Our results are based on the following observations (see e.g. \cite{IIKS}). First an algebraic Lemma, showing that the resolvent of $I-K$ is again integrable.
\begin{lemma}\label{lemma1}
Given an operator $I-K$ on $L^2(\Gamma)$ in the previous algebra with kernel \eqref{IIKStheo1}. Suppose the inverse $(I-K)^{-1}$ exists, then $I+R = (I-K)^{-1}$ lies again in the same algebra with
\begin{equation}\label{IIKStheo2}
	R(\lambda,\mu)=\frac{\sum_{i=1}^MF_i(\lambda)H_i(\mu)}{\lambda-\mu},\hspace{0.5cm}\sum_{i=1}^MF_i(\lambda)H_i(\lambda)=0
\end{equation}
and the functions $F_i,H_i$ are given by
\begin{equation}\label{IIKStheo3}
	F_i(\lambda)=\Big((I-K)^{-1}f_i\Big)(\lambda),\hspace{0.5cm} H_i(\lambda)=\Big((I-K^t)^{-1}h_i\Big)(\lambda).
\end{equation}
\end{lemma}
Secondly an analytical Lemma, which connects integrable operators to a Riemann-Hilbert problem.
\begin{lemma}\label{lemma2}
Let $K$ be of integrable type such that $(I-K)^{-1}$ exists and let $Y=Y(z)$ denote the unique solution of the following $M\times M$ Riemann-Hilbert problem (RHP)
\begin{itemize}
	\item $Y(z)$ is analytic for $z\in\mathbb{C}\backslash \Gamma$
	\item On the contour $\Gamma$, the boundary values of the function $Y(z)$ satisfy the jump relation
	\begin{equation*}
		Y_+(z)=Y_-(z)\big(I-2\pi i f(z)h^t(z)\big),\hspace{0.5cm} z\in\Gamma
	\end{equation*}
	where $f(z)=\big(f_1(z),\ldots,f_M(z)\big)^t$ and similarly $h(z)=\big(h_1(z),\ldots,h_M(z)\big)^t$
	\item At an endpoint of the contour $\Gamma$, $Y(z)$ has no more than a logarithmic singularity
	\item As $z\rightarrow\infty$
	\begin{equation*}
		Y(z)=I+O\big(z^{-1}\big)
	\end{equation*}
\end{itemize}
Then $Y(z)$ determines the resolvent kernel via
\begin{equation}\label{IIKStheo4}
	F(z)=Y(z)f(z),\hspace{0.5cm}H(z)=\big(Y^t(z)\big)^{-1}h(z)
\end{equation}
and conversely the solution of the above RHP is expressible in terms of the function $F(z)$ using the Cauchy integral
\begin{equation}\label{IIKStheo5}
	Y(z)=I-\int\limits_{\Gamma}F(w)h^t(w)\frac{dw}{w-z}.
\end{equation}
\end{lemma} 
Let us use this general setup in the given situation \eqref{lambda3kernel}. We have
\begin{equation*}
	\gamma K_{\textnormal{csin}}(\lambda,\mu) = \frac{f^t(\lambda)h(\mu)}{\lambda-\mu},\ \ f(\lambda)=\sqrt{\frac{\gamma}{2\pi i}}\binom{e^{i(\frac{4}{3}\lambda^3+x\lambda)}}{e^{-i(\frac{4}{3}\lambda^3+x\lambda)}},\ \ h(\lambda)=\sqrt{\frac{\gamma}{2\pi i}}\binom{e^{-i(\frac{4}{3}\lambda^3+x\lambda)}}{-e^{i(\frac{4}{3}\lambda^3+x\lambda)}}
\end{equation*}
where $\sqrt{z}$ is defined on $\mathbb{C}\backslash(-\infty,0]$ with its branch fixed by the condition $\sqrt{z}>0$ if $z>0$. Lemma \ref{lemma2} provides us with the following $Y$-RHP 
\begin{itemize}
	\item $Y(\lambda)$ is analytic for $\lambda\in\mathbb{C}\backslash[-s,s]$
	\item Along the line segment $[-s,s]$, oriented from left to right, the following jump holds
		\begin{equation*}
			Y_+(\lambda)=Y_-(\lambda)\begin{pmatrix}
			1-\gamma & \gamma e^{2i(\frac{4}{3}\lambda^3+x\lambda)}\\
			-\gamma e^{-2i(\frac{4}{3}\lambda^3+x\lambda)} & 1+\gamma\\
			\end{pmatrix}, \hspace{0.5cm} \lambda\in[-s,s]
		\end{equation*}
	\item At the endpoints $\pm s$, $Y(\lambda)$ has logarithmic singularities, i.e.
	\begin{equation}\label{Yendpoint}
		Y(\lambda)=O\big(\ln(\lambda\mp s)\big),\ \ \lambda\rightarrow\pm s
	\end{equation}
	\item We have $Y(\lambda)\rightarrow I$ as $\lambda\rightarrow\infty$.
\end{itemize}
We can factorize the jump matrix
\begin{equation*}
	\begin{pmatrix}
			1-\gamma & \gamma e^{2i(\frac{4}{3}\lambda^3+x\lambda)}\\
			-\gamma e^{-2i(\frac{4}{3}\lambda^3+x\lambda)} & 1+\gamma\\
			\end{pmatrix}=e^{i(\frac{4}{3}\lambda^3+x\lambda)\sigma_3}\begin{pmatrix}
			1-\gamma & \gamma\\
			-\gamma & 1+\gamma\\
			\end{pmatrix}e^{-i(\frac{4}{3}\lambda^3+x\lambda)\sigma_3},
\end{equation*}
and employ a first transformation.


\section{First transformation of the RHP}\label{sec3}

We make the following substitution in the original $X$-RHP
\begin{equation}\label{fromYtoX}
	X(\lambda) = Y(\lambda)e^{i(\frac{4}{3}\lambda^3+x\lambda)\sigma_3},\ \ \lambda\in\mathbb{C}\backslash[-s,s].
\end{equation}
This leads to a RHP for the function $X(\lambda)$, our ``master'' RHP: 
\begin{itemize}
	\item $X(\lambda)$ is analytic for $\lambda\in\mathbb{C}\backslash[-s,s]$
	\item The following jump holds
	\begin{equation}\label{Xjump}
	X_+(\lambda)=X_-(\lambda)\begin{pmatrix}
		1-\gamma & \gamma\\
		-\gamma & 1+\gamma\\
		\end{pmatrix},\hspace{0.5cm}\lambda\in[-s,s]
	\end{equation}
	\item From \eqref{Yendpoint}, we deduce the following refined endpoint behavior
	\begin{equation}\label{Xendpoint}
	X(\lambda)=\check{X}(\lambda)\Bigg[I+\frac{\gamma}{2\pi i}\begin{pmatrix}
	-1 & 1\\
	-1 & 1\\
	\end{pmatrix}\ln\bigg(\frac{\lambda-s}{\lambda+s}\bigg)\Bigg]
\end{equation}
where $\check{X}(\lambda)$ is analytic at $\lambda=\pm s$ and the branch of the logarithm is fixed by the condition $-\pi<\textnormal{arg}\ \frac{\lambda-s}{\lambda+s}<\pi$.
	\item At infinity, $X(\lambda)$ is normalized as follows
\begin{equation}\label{Xasyinfinity}
	X(\lambda)=\Big(I+O\big(\lambda^{-1}\big)\Big)e^{i(\frac{4}{3}\lambda^3+x\lambda)\sigma_3},\hspace{0.5cm}\lambda\rightarrow\infty.
\end{equation}
\end{itemize}

The latter master problem will be solved asymptotically as $s\rightarrow\infty$ in the next sections by approximating its global solution with local model functions. Before we start this analysis in detail, we first connect the solution of lthe $X$-RHP to the Fredholm determinants $\det\left(I-\gamma K_{\textnormal{csin}}\right)$.

\section{Logarithmic derivatives - connection to $X$-RHP}\label{sec4}

We will express logarithmic derivatives of the determinant $\det(I-\gamma K_{\textnormal{csin}})$ in terms of the solution of the $X$-RHP. To this end recall the following classical identity, valid for any differentiable family of trace class operators \cite{S}
\begin{equation}\label{traceform}
    \frac{\partial}{\partial s}\ln\det\left(I-\gamma K_{\textnormal{csin}}\right)=-\textnormal{trace}\bigg(\big(I-\gamma K_{\textnormal{csin}}\big)^{-1}\frac{\partial}{\partial s} \left(\gamma K_{\textnormal{csin}}\right)\bigg).
\end{equation}
In our situation
\begin{equation*}
    \frac{\partial K_{\textnormal{csin}}}{\partial s}(\lambda,\mu)=K_{\textnormal{csin}}(\lambda,\mu)\big(\delta(\mu-s)+\delta(\mu+s)\big),
\end{equation*}
where, by definition
\begin{equation*}
    \int_{-s}^s\delta(w\mp s)f(w)dw = f(\pm s),
\end{equation*}
and therefore
\begin{equation*}
    -\textnormal{trace}\ \Big(\big(I-\gamma K_{\textnormal{csin}}\big)^{-1}\frac{\partial}{\partial s}\left(\gamma K_{\textnormal{csin}}\right)\Big) =
    -R(s,s)-R(-s,-s)
\end{equation*}
with $R(\lambda,\mu)$ denoting the kernel (see \eqref{IIKStheo2}) of the resolvent
$R=(I-\gamma K_{\textnormal{csin}})^{-1}\gamma K_{\textnormal{csin}}$. The latter derivative can be simplified using the given definitions
\begin{equation*}
	f_1(\lambda)=-h_2(\lambda),\hspace{0.5cm} f_2(\lambda)=h_1(\lambda)
\end{equation*} 
as well as the identity $\det Y(\lambda)\equiv 1$, a direct consequence of the unimodularity of the jump matrix $G(\lambda)$ and Liouville's theorem:
\begin{equation*}
	R(\lambda,\mu)=\frac{F_1(\lambda)H_1(\mu)+F_2(\lambda)H_2(\mu)}{\lambda-\mu}=\frac{F_1(\lambda)F_2(\mu)-F_2(\lambda)F_1(\mu)}{\lambda-\mu}.
\end{equation*}
Since $R(\lambda,\mu)$ is continuous along the diagonal $\lambda=\mu$ (see \eqref{IIKStheo2}) we obtain further
\begin{equation}\label{resolventRHP}
	R(s,s)=F_1'(s)F_2(s)-F_2'(s)F_1(s),\ \ R(-s,-s)=F_1'(-s)F_2(-s)-F_2'(-s)F_1(-s)
\end{equation}
provided $F_i$ is analytic at $\lambda=\pm s$. One way to see this is a follows. Use \eqref{fromYtoX} and \eqref{IIKStheo4}
\begin{equation*}
	F(\lambda)=X(\lambda)e^{-i(\frac{4}{3}\lambda^3+x\lambda)\sigma_3}f(\lambda)
\end{equation*}
to derive from \eqref{Xendpoint} the following local identity
\begin{eqnarray}\label{Fendpointlocal}
	F(\lambda)&=&\check{X}(\lambda)\Bigg[I+\frac{\gamma}{2\pi i}\begin{pmatrix}
	-1 & 1\\
	-1 & 1\\
	\end{pmatrix}\ln\bigg(\ln\frac{\lambda-s}{\lambda+s}\bigg)\Bigg]e^{-i(\frac{4}{3}\lambda^3+x\lambda)\sigma_3}f(\lambda)\nonumber\\
	&=&\check{X}(\lambda)\sqrt{\frac{\gamma}{2\pi i}}\binom{1}{1}
\end{eqnarray}	
valid in a neighborhood of $\lambda=\pm s$. But this proves analyticity of $F(\lambda)$ at the endpoints and as we shall see later on, \eqref{Fendpointlocal} will allow us to connect \eqref{traceform} via \eqref{resolventRHP} to the solution of the $X$-RHP. We summarize
\begin{prop}\label{prop1}
The logarithmic $s$-derivative of the given Fredholm determinant can be expressed as
\begin{eqnarray}\label{sidentity}
	\frac{d}{ds}\ln\det(I-\gamma K_{\textnormal{csin}})&=&-R(s,s)-R(-s,-s),\\
	 &&\ R(\pm s,\pm s)=F_1'(\pm s)F_2(\pm s)-F_2'(\pm s)F_1(\pm s)\nonumber
\end{eqnarray}
where the connection to the $X$-RHP is established through
\begin{equation*}
	F(\lambda)=\check{X}(\lambda)\sqrt{\frac{\gamma}{2\pi i}}\binom{1}{1},\ \ \lambda\rightarrow\pm s.
\end{equation*}
\end{prop}
Besides the logarithmic $s$-derivative we also differentiate with respect to $x$
\begin{equation*}
    \frac{\partial}{\partial x}\ln\det(I-\gamma K_{\textnormal{csin}})=-\textnormal{trace}\
    \bigg(\big(I-\gamma K_{\textnormal{csin}}\big)^{-1}\frac{\partial}{\partial x}\left(\gamma K_{\textnormal{csin}}\right)\bigg).
\end{equation*}
In our situation the kernel depends explicitly on $x$, indeed
\begin{equation*}
	\frac{\partial}{\partial x}\big(\gamma K_{\textnormal{csin}}(\lambda,\mu)\big)=i\big(f_1(\lambda)h_1(\mu)-f_2(\lambda)h_2(\mu)\big)
\end{equation*}
and with \eqref{IIKStheo3}
\begin{equation*}
	-\textnormal{trace}\
    \bigg(\big(I-\gamma K_{\textnormal{csin}}\big)^{-1}\frac{\partial}{\partial x}K_{\textnormal{csin}}\bigg)=-i\int\limits_{-s}^s\big(F_1(\lambda)h_1(\lambda)-F_2(\lambda)h_2(\lambda)\big)d\lambda.
\end{equation*}
On the other hand the Cauchy integral \eqref{IIKStheo4} implies
\begin{equation*}
	Y(\lambda)=I+\frac{m_1}{\lambda}+O\big(\lambda^{-2}\big),\ \ \lambda\rightarrow\infty;\ m_1=\int\limits_{-s}^sF(w)h^t(w)dw
\end{equation*}
so
\begin{equation*}
	\frac{\partial}{\partial x}\ln\det\left(I-\gamma K_{\textnormal{csin}}\right)=i\big(m_1^{22}-m_1^{11}\big),\ \ \ m_1=\big(m_1^{ij}\big)
\end{equation*}
and the connection to the $X$-RHP is established via \eqref{Xasyinfinity}. Again we summarize
\begin{prop}\label{prop2}
The logarithmic $x$-derivative of the given Fredholm determinant can be expressed as
\begin{equation}\label{xidentity}
	\frac{\partial}{\partial x}\ln\det(I-\gamma K_{\textnormal{csin}})=i\big(X_1^{22}-X_1^{11}\big)
\end{equation}
with
\begin{equation*}
	X(\lambda)\sim \Big(I+\frac{X_1}{\lambda}+\frac{X_2}{\lambda^2}+\frac{X_3}{\lambda^3}+O\big(\lambda^{-4}\big)\Big)e^{i(\frac{4}{3}\lambda^3+x\lambda)\sigma_3},\ \ \lambda\rightarrow\infty; \ X_1=\big(X_1^{ij}\big).
\end{equation*}
\end{prop}
The logarithmic $s$ and $x$ derivatives are sufficient to determine the large $s$-asymptotics of $\det(I-\gamma K_{\textnormal{csin}})$ up to the constant term. In order to evaluate the constant explicitly for $\gamma<1$, we use the following approach (see \cite{DIK} for a similar method used in the asymptotics of Toeplitz determinants).
\smallskip

Start with the logarithmic $\gamma$-derivative 
\begin{equation*}
	\frac{\partial}{\partial\gamma}\ln\det(I-\gamma K_{\textnormal{csin}}) = -\textnormal{trace}\bigg(\big(I-\gamma K_{\textnormal{csin}}\big)^{-1}K_{\textnormal{csin}}\bigg)=-\frac{1}{\gamma}\int\limits_{-s}^sR(\lambda,\lambda)d\lambda.
\end{equation*}
Our goal is to express the latter integral over the resolvent kernel in terms of the solution of the underlying $X$-RHP. Recall \eqref{resolventRHP},\eqref{IIKStheo4}, the definition of the functions $f(\lambda),h(\lambda)$ and unimodularity of $Y(\lambda)$
\begin{eqnarray*}
	R(\lambda,\lambda)&=&F_1'(\lambda)F_2(\lambda)-F_2'(\lambda)F_1(\lambda)=\frac{\gamma}{\pi}\big(4\lambda^2+x\big)+\frac{\gamma}{2\pi i}\big(Y_{11}'(\lambda)Y_{22}(\lambda)\\
	&&-Y_{11}(\lambda)Y_{22}'(\lambda)+Y_{12}'(\lambda)Y_{21}(\lambda)-Y_{21}'(\lambda)Y_{12}(\lambda)\big)+\big(Y_{11}'(\lambda)Y_{21}(\lambda)\\
	&&-Y_{11}(\lambda)Y_{21}'(\lambda)\big)f_1^2(\lambda)+\big(Y_{12}'(\lambda)Y_{22}(\lambda)-Y_{12}(\lambda)Y_{22}'(\lambda)\big)f_2^2(\lambda)
\end{eqnarray*}
where $(')$ indicates differentiation with respect to $\lambda$. Next \eqref{fromYtoX} and unimodularity allows us to rewrite the previous expression for $R(\lambda,\lambda)$ solely in terms of the $X$-RHP 
\begin{eqnarray}\label{gammaderiv1}
	R(\lambda,\lambda)&=&\frac{\gamma}{2\pi i}\Big[\big(X_{11}'(\lambda)+X_{12}'(\lambda)\big)\big(X_{21}(\lambda)+X_{22}(\lambda)\big)\nonumber\\
	&&-\big(X_{11}(\lambda)+X_{12}(\lambda)\big)\big(X_{21}'(\lambda)+X_{22}'(\lambda)\big)\Big].
\end{eqnarray}
Our next move will replace all terms involving derivatives with respect to $\lambda$ and this can be done by considering the differential equations associated with the $X$-RHP, see also section \ref{sec5} below:
\smallskip

All jump matrices in the $X$-RHP are unimodular and constant with respect to $\lambda$, thus the well-defined logarithmic derivative $X_{\lambda}X^{-1}(\lambda)$ are rational functions. Indeed using \eqref{Xasyinfinity} and \eqref{Xendpoint}
\begin{equation}\label{diffeq1}
	\frac{\partial X}{\partial\lambda} = \bigg[4i\lambda^2\sigma_3-4i\lambda\begin{pmatrix}
	0 & b\\
	-c & 0 \\
	\end{pmatrix}+\begin{pmatrix}
	d & e\\
	f & -d\\
	\end{pmatrix} +\frac{A}{\lambda-s} -\frac{B}{\lambda+s}\bigg]X\equiv \mathcal{A}(\lambda,s,x)X
\end{equation}
where
\begin{equation}\label{ABdef}
	A = \frac{\gamma}{2\pi i}\check{X}(s)\begin{pmatrix}
	-1 & 1\\
	-1 & 1\\
	\end{pmatrix}\big(\check{X}(s)\big)^{-1};\hspace{1cm}B=\frac{\gamma}{2\pi i}\check{X}(-s)\begin{pmatrix}
	-1 & 1\\
	-1 & 1\\
	\end{pmatrix}\big(\check{X}(-s)\big)^{-1}
\end{equation}
and with parameters $b,c,d,e,f$ which can be expressed in terms of the entries of $X_1$ and $X_2$
\begin{equation}\label{bcddef}
	b = 2X_{1}^{12},\hspace{0.5cm} c=2X_1^{21},\hspace{0.5cm} d=ix+8iX_1^{12}X_1^{21}
\end{equation}
\begin{equation}\label{efdef}
	e = 8i\Big(X_1^{12}X_1^{22}-X_2^{12}\Big),\hspace{0.5cm} f=-8i\Big(X_1^{21}X_1^{11}-X_2^{21}\Big).
\end{equation}
Substituting \eqref{diffeq1} into \eqref{gammaderiv1} and recalling \eqref{ABdef} we obtain with $A=(A_{ij}),B=(B_{ij})$
\begin{eqnarray}\label{gammaderiv2}
	R(\lambda,\lambda)&=&\frac{\gamma}{2\pi i}\Bigg[\bigg(8i\lambda^2+2d+\frac{A_{11}-A_{22}}{\lambda-s}-\frac{B_{11}-B_{22}}{\lambda+s}\bigg)\big(X_{11}(\lambda)+X_{12}(\lambda)\big)\nonumber\\
	&&\times\big(X_{21}(\lambda)+X_{22}(\lambda)\big)\nonumber\\
	&&\hspace{0.5cm}+\bigg(-4i\lambda b+e+\frac{A_{12}}{\lambda-s}-\frac{B_{12}}{\lambda+s}\bigg)\big(X_{21}(\lambda)+X_{22}(\lambda)\big)^2\nonumber\\
	&&\hspace{0.5cm}+\bigg(-4i\lambda c-f-\frac{A_{21}}{\lambda-s}+\frac{B_{21}}{\lambda+s}\bigg)\big(X_{11}(\lambda)+X_{12}(\lambda)\big)^2\Bigg].
\end{eqnarray}
Next we $\gamma$-differentiate the $X$-RHP in \eqref{Xjump} to obtain the following additive RHP for the function $H(\lambda)=\frac{\partial X}{\partial\gamma}(\lambda)\big(X(\lambda)\big)^{-1}$
\begin{itemize}
	\item $H(\lambda)$ is analytic for $\lambda\in\mathbb{C}\backslash[-s,s]$
	\item Along the line segment $[-s,s]$, oriented from left to right
\begin{equation*}
	H_+(\lambda) = H_-(\lambda)+X_-(\lambda)\begin{pmatrix}
	-1 & 1\\
	-1 & 1\\
	\end{pmatrix}\big(X_-(\lambda)\big)^{-1},\ \ \lambda\in[-s,s]
\end{equation*}
	\item $H(\lambda)$ has at most logarithmic singularities at the endpoints $\lambda=\pm s$
	\begin{equation}\label{Ssingular}
		H(\lambda) = \frac{\partial\check{X}}{\partial\gamma}(\lambda)\big(\check{X}(\lambda)\big)^{-1}+\check{X}(\lambda)\frac{1}{2\pi i}\begin{pmatrix}
		-1 & 1\\
		-1 & 1\\
		\end{pmatrix}\big(\check{X}(\lambda)\big)^{-1}\ln\frac{\lambda-s}{\lambda+s},\ \ \ \lambda\rightarrow\pm s
	\end{equation}
	\item As $\lambda\rightarrow\infty$, we have $H(\lambda)\rightarrow0$
\end{itemize}
If we let
\begin{equation*}
	K(\lambda) = X(\lambda)\begin{pmatrix}
	-1 & 1\\
	-1 & 1\\
	\end{pmatrix}\big(X(\lambda)\big)^{-1},\ \ \lambda\in\mathbb{C}\backslash[-s,s],
\end{equation*}
then $K_+(\lambda)=K_-(\lambda),\lambda\in[-s,s]$ and $K(\lambda)$ is bounded as $\lambda\rightarrow\pm s$. Hence $K(\lambda)$ is entire and we have a solution to the $H$-RHP
\begin{eqnarray*}
	&&H(\lambda) = \frac{1}{2\pi i}\int\limits_s^s\frac{K_-(w)}{w-\lambda}dw = \frac{1}{2\pi i}\int\limits_{-s}^s\frac{K(w)}{w-\lambda}dw\\
	&=&\frac{1}{2\pi i}\int\limits_{-s}^s\biggl(\begin{smallmatrix}
	-(X_{11}(w)+X_{12}(w))(X_{21}(w)+X_{22}(w)) & (X_{11}(w)+X_{12}(w))^2\\
	-(X_{21}(w)+X_{22}(w))^2 & (X_{11}(w)+X_{12}(w))(X_{21}(w)+X_{22}(w))\\
	\end{smallmatrix}\biggr)\frac{dw}{w-\lambda}.
\end{eqnarray*}
This solution enables us to rewrite $\int_{-s}^s R(\lambda,\lambda)d\lambda$ in \eqref{gammaderiv2}, for instance
\begin{equation*}
	\int\limits_{-s}^s\lambda^n\big(X_{11}(\lambda)+X_{12}(\lambda)\big)\big(X_{21}(\lambda)+X_{22}(\lambda)\big)d\lambda = \int\limits_{\Sigma}w^nH_{11}(w)dw,\ \ n\in\mathbb{Z}_{\geq 0}
\end{equation*}
with $\Sigma$ denoting a closed Jordan curve around the interval $[-s,s]$ and where we used
\begin{equation*}
	\lambda^n = \frac{1}{2\pi i}\int\limits_{\Sigma}\frac{w^n}{w-\lambda} dw,\ \ \lambda\in[-s,s].
\end{equation*}
Similarly
\begin{eqnarray*}
	\int\limits_{-s}^s\lambda^n\big(X_{21}(\lambda)+X_{22}(\lambda)\big)^2d\lambda &=& \int\limits_{\Sigma}w^nH_{21}(w)dw,\\ 		 \int\limits_{-s}^s\lambda^n\big(X_{11}(\lambda)+X_{12}(\lambda)\big)^2d\lambda&=&-\int\limits_{\Sigma}w^nH_{12}(w)dw
\end{eqnarray*} 
and we obtain
\begin{eqnarray}\label{gammaderiv3}
	&&\frac{\partial}{\partial\gamma}\ln\det(I-\gamma K_{\textnormal{csin}})=-\frac{1}{\gamma}\int\limits_{-s}^sR(\lambda,\lambda)d\lambda
	=-\frac{1}{2\pi i}\Bigg[8i\int\limits_{\Sigma}w^2H_{11}(w)dw\nonumber\\
	&&+\int\limits_{\Sigma}\big(2dH_{11}(w)+eH_{21}(w)+fH_{12}(w)\big)dw
	-4i\int\limits_{\Sigma}w\big(bH_{21}(w)-cH_{12}(w)\big)dw\nonumber\\
	&&-\int\limits_{-s}^s\Big((A_{11}-A_{22})K_{11}(\lambda)+A_{12}K_{21}(\lambda)+A_{21}K_{12}(\lambda)\Big)\frac{d\lambda}{\lambda-s}\nonumber\\
	&&+\int\limits_{-s}^s\Big((B_{11}-B_{22})K_{11}(\lambda)+B_{12}K_{21}(\lambda)+B_{21}K_{12}(\lambda)\Big)\frac{d\lambda}{\lambda+s}\Bigg].
\end{eqnarray}
Since 
\begin{equation*}
	K(\lambda) = \frac{2\pi i}{\gamma}A +O(\lambda-s),\ \lambda\rightarrow s,\hspace{1cm} K(\lambda) = \frac{2\pi i}{\gamma}B+O(\lambda+s),\ \lambda\rightarrow -s
\end{equation*}
and
\begin{equation*}
	(A_{11}-A_{22})A_{11}+2A_{12}A_{21}=0=(B_{11}-B_{22})B_{11}+2B_{12}B_{21},
\end{equation*}
we deduce that the last two integrals in \eqref{gammaderiv3} are indeed well-defined. To evaluate them, let
\begin{eqnarray*}
	\hat{H}(\lambda) &=& H(\lambda)-\check{X}(s)\begin{pmatrix}
	-1 & 1\\
	-1 & 1\\
	\end{pmatrix}\big(\check{X}(s)\big)^{-1}\frac{1}{2\pi i}\ln\frac{\lambda-s}{\lambda+s},\ \ \lambda\in\mathbb{C}\backslash[-s,s]\\
	\tilde{H}(\lambda) &=& H(\lambda)-\check{X}(-s)\begin{pmatrix}
	-1 & 1\\
	-1 & 1\\
	\end{pmatrix}\big(\check{X}(-s)\big)^{-1}\frac{1}{2\pi i}\ln\frac{\lambda-s}{\lambda+s},\ \ \lambda\in\mathbb{C}\backslash[-s,s].
\end{eqnarray*}
From \eqref{Ssingular} we see that $\hat{H}(\lambda)$ is bounded as $\lambda\rightarrow s$ and $\tilde{H}(\lambda)$ is bounded as $\lambda\rightarrow -s$, more precisely
\begin{equation*}
	\hat{H}(s) = \frac{\partial \check{X}}{\partial\gamma}(s)\big(\check{X}(s)\big)^{-1},\hspace{0.5cm} \tilde{H}(-s)=\frac{\partial\check{X}}{\partial\gamma}(-s)\big(\check{X}(-s)\big)^{-1},
\end{equation*}
also
\begin{equation*}
	\hat{H}_+(\lambda)=\hat{H}_-(\lambda)+K(\lambda)-K(s),\hspace{0.5cm} \tilde{H}_+(\lambda)=\tilde{H}_-(\lambda)+K(\lambda)-K(-s),\ \lambda\in[-s,s],
\end{equation*}
hence
\begin{equation}\label{gammaderiv4}
	\hat{H}(\lambda)=\frac{1}{2\pi i}\int\limits_{-s}^s\frac{K(w)-K(s)}{w-\lambda}dw,\ \ \ \tilde{H}(\lambda)=\frac{1}{2\pi i}\int\limits_{-s}^s\frac{K(w)-K(-s)}{w-\lambda}dw
\end{equation}
and we conclude
\begin{eqnarray*}
	&&\int\limits_{-s}^s\Big((A_{11}-A_{22})K_{11}(\lambda)+A_{12}K_{21}(\lambda)+A_{21}K_{12}(\lambda)\Big)\frac{d\lambda}{\lambda-s}\\
	&=&(A_{11}-A_{22})\int\limits_{-s}^s\frac{K_{11}(\lambda)-K_{11}(s)}{\lambda-s}d\lambda +A_{12}\int_{-s}^s\frac{K_{21}(\lambda)-K_{21}(s)}{\lambda-s}d\lambda\\
	&&+A_{21}\int\limits_{-s}^s\frac{K_{12}(\lambda)-K_{12}(s)}{\lambda-s}d\lambda=2\pi i\big[(A_{11}-A_{22})\hat{H}_{11}(s)+A_{12}\hat{H}_{21}(s)+A_{21}\hat{H}_{12}(s)\big].
\end{eqnarray*}
Similarly
\begin{eqnarray*}
	&&\int\limits_{-s}^s\Big((B_{11}-B_{22})K_{11}(\lambda)+B_{12}K_{21}(\lambda)+B_{21}K_{12}(\lambda)\Big)\frac{d\lambda}{\lambda+s}\\
	&=&2\pi i\big[(B_{11}-B_{22})\tilde{H}_{11}(-s)+B_{12}\tilde{H}_{21}(-s)+B_{21}\tilde{H}_{12}(-s)\big].
\end{eqnarray*}
To evaluate the remaining integrals in \eqref{gammaderiv3}, we recall
\begin{eqnarray*}
	H(\lambda) &=& \frac{1}{\lambda}(X_1)_{\gamma}+\frac{1}{\lambda^2}\big((X_2)_{\gamma}-(X_1)_{\gamma}X_1\big)\\
	&&+\frac{1}{\lambda^3}\big((X_3)_{\gamma}+(X_1)_{\gamma}(X_1^2-X_2)-(X_2)_{\gamma}X_1\big)+O\big(\lambda^{-4}\big),\ \ \lambda\rightarrow\infty
\end{eqnarray*}
and apply residue theorem
\begin{eqnarray*}
	&&\int\limits_{\Sigma}H(w)dw = 2\pi i(X_1)_{\gamma},\hspace{0.5cm}  \int\limits_{\Sigma}wH(w)dw = 2\pi i\big((X_2)_{\gamma}-(X_1)_{\gamma}X_1\big)\\
	&&\int\limits_{\Sigma}w^2H(w)dw = 2\pi i\big((X_3)_{\gamma}+(X_1)_{\gamma}(X_1^2-X_2)-(X_2)_{\gamma}X_1\big).
\end{eqnarray*}
We summarize the previous computations
\begin{prop}\label{prop3}
The logarithmic $\gamma$-derivative of the given Fredholm determinant can be expressed as
\begin{eqnarray}\label{gammaderiv5}
	&&\frac{\partial}{\partial\gamma}\ln\det\left(I-\gamma K_{\textnormal{csin}}\right) = -8i\Big((X_3)_{\gamma}+(X_1)_{\gamma}(X_1^2-X_2)-(X_2)_{\gamma}X_1\Big)^{11}\nonumber\\
	&&-2d\Big((X_1)_{\gamma}\Big)^{11}+4ib\Big((X_2)_{\gamma}-(X_1)_{\gamma}X_1\Big)^{21}-4ic\Big((X_2)_{\gamma}-(X_1)_{\gamma}X_1\Big)^{12}\nonumber\\
	&&-e\Big((X_1)_{\gamma}\Big)^{21}-f\Big((X_1)_{\gamma}\Big)^{12}+\Big((A_{11}-A_{22})\hat{H}_{11}(s)+A_{12}\hat{H}_{21}(s)+A_{21}\hat{H}_{12}(s)\Big)\nonumber\\
	&&\hspace{0.5cm}-\Big((B_{11}-B_{22})\tilde{H}_{11}(-s)+B_{12}\tilde{H}_{21}(-s)+B_{21}\tilde{H}_{12}(-s)\Big)
\end{eqnarray}
with
\begin{equation*}
	X(\lambda)\sim \Big(I+\frac{X_1}{\lambda}+\frac{X_2}{\lambda^2}+\frac{X_3}{\lambda^3}+O\big(\lambda^{-4}\big)\Big)e^{i(\frac{4}{3}\lambda^3+x\lambda)\sigma_3},\ \ \lambda\rightarrow\infty
\end{equation*}
and the functions $b,c,d,e,f,A,B$ are defined in \eqref{bcddef},\eqref{efdef} and \eqref{ABdef}.
\end{prop}


\section{Differential equations associated with the determinant $\det\left(I-\gamma K_{\textnormal{csin}}\right)$}\label{sec5}

Our considerations rely only on the underlying Riemann-Hilbert problem. Nevertheless, before we continue the asymptotical analysis, we will take a short look into the differential equations associated with the $X$-RHP.\smallskip

We already used the differential equation
\begin{equation*}
	\frac{\partial X}{\partial\lambda}=\mathcal{A}(\lambda,s,x)X
\end{equation*}
in the derivation of \eqref{gammaderiv5}. Also, the logarithmic derivatives $X_sX^{-1}(\lambda)$ and $X_xX^{-1}(\lambda)$ are rational functions, indeed
\begin{equation*}
	\frac{\partial X}{\partial s} = \bigg[-\frac{A}{\lambda-s}-\frac{B}{\lambda+s}\bigg]X \equiv \Theta(\lambda,s,x)X
\end{equation*}
and furthermore
\begin{equation*}
	\frac{\partial X}{\partial x}=\bigg[i\lambda\sigma_3-i\begin{pmatrix}
	0 & b\\
	-c & 0\\
	\end{pmatrix}\bigg]X \equiv \Pi(\lambda,s,x)X.
\end{equation*}
Hence we arrive at the Lax-system for the function $X$
\begin{equation*}
    \left\{
      \begin{array}{c}
        \frac{\partial X}{\partial \lambda}=\mathcal{A}(\lambda,s,x)X \\
        \\
        \frac{\partial X}{\partial s}=\Theta(\lambda,s,x)X,  \\
        \\
        \frac{\partial X}{\partial x}=\Pi(\lambda,s,x)X.
      \end{array}
    \right.
\end{equation*}
Considering the compatibility conditions
 \begin{equation}\label{compa}
    \mathcal{A}_s-\Theta_{\lambda}=[\Theta,\mathcal{A}],\hspace{1cm}
\mathcal{A}_x-\Pi_{\lambda}=[\Pi,\mathcal{A}],\hspace{1cm}
\Theta_x-\Pi_s=[\Pi,\Theta]
\end{equation}
we are lead to a system of eighteen nonlinear ordinary differential equations. Since it is possible to express the previous derivatives of $\ln\det(I-\gamma K_{\textnormal{csin}})$ solely in terms of the unknowns $b,c,d,e,f,A$ and $B$. Since it is possible to express the previous derivatives of $\ln\det\left(I-\gamma K_{\textnormal{csin}}\right)$ solely in terms of the unknowns $b,c,d,e,f,A$ and $B$ one can then try to derive a differential equation for the Fredholm determinant using \eqref{compa}. We plan to address this issue in a future publication.
\section{Second transformation of the RHP - rescaling and opening of lenses}\label{sec6}

Let us scale the variables in \eqref{fromYtoX} as $\lambda=zs$, so that
\begin{equation}\label{fromXtoT}
	T(z) = X(zs)e^{-s^3\vartheta(z)\sigma_3},\hspace{0.5cm} z\in\mathbb{C}\backslash[-1,1],\ \ \vartheta(z)=i\left(\frac{4}{3}z^3+\frac{xz}{s^2}\right)
\end{equation}
solves the following RHP, which up to the rescaling $\lambda=zs$, is identical to the initial $Y$-RHP
\begin{itemize}
	\item $T(z)$ is analytic for $z\in\mathbb{C}\backslash[-1,1]$
	\item Along the line segment $[-1,1]$ oriented from left to right
	\begin{equation*}
		T_+(z)=T_-(z)e^{s^3\vartheta(z)\sigma_3}\begin{pmatrix}
		1-\gamma & \gamma\\
		-\gamma & 1+\gamma\\
		\end{pmatrix}e^{-s^3\vartheta(z)\sigma_3},\hspace{0.5cm}z\in[-1,1]
	\end{equation*}
	\item In a neighborhood of the endpoints $z=\pm 1$
	\begin{equation*}
		T(z)e^{s^3\vartheta(z)\sigma_3} = \check{X}(zs)\Bigg[I+\frac{\gamma}{2\pi i}\begin{pmatrix}
	-1 & 1\\
	-1 & 1\\
	\end{pmatrix}\ln\bigg(\frac{z-1}{z+1}\bigg)\Bigg]
	\end{equation*}
	\item At infinity
	\begin{equation}\label{thetafunction}
		T(z) = I+O\big(z^{-1}\big),\hspace{0.5cm}z\rightarrow\infty
	\end{equation}
\end{itemize}
We now move to a RHP formulated according to the sign-diagrm of $\textnormal{Re}\ \vartheta(z)$ which is depicted in Figure \ref{fig1}. In this Figure we choose $x$ from a compact subset of the real line, $s>0$ is sufficiently large and
\begin{equation*}
	z_{\pm} = \pm i\sqrt{\frac{3x}{4s^2}}
\end{equation*}
denote the two vertices of the depicted algebraic curves.
\begin{figure}[tbh]
  \begin{center}
  \psfragscanon
  \psfrag{1}{\footnotesize{$\textnormal{Re}\ \vartheta<0$}}
  \psfrag{2}{\footnotesize{$\textnormal{Re}\ \vartheta<0$}}
  \psfrag{3}{\footnotesize{$\textnormal{Re}\ \vartheta>0$}}
  \psfrag{4}{\footnotesize{$\textnormal{Re}\ \vartheta>0$}}
  \psfrag{5}{\footnotesize{$\textnormal{Re}\ \vartheta<0$}}
  \psfrag{6}{\footnotesize{$\textnormal{Re}\ \vartheta<0$}}
  \psfrag{7}{\footnotesize{$\textnormal{Re}\ \vartheta>0$}}
  \psfrag{8}{\footnotesize{$\textnormal{Re}\ \vartheta>0$}}
  \psfrag{9}{\footnotesize{$\textnormal{Re}\ \vartheta>0$}}
  \psfrag{10}{\footnotesize{$\textnormal{Re}\ \vartheta<0$}}
  \psfrag{11}{\footnotesize{$\textnormal{Re}\ \vartheta<0$}}
  \psfrag{12}{\footnotesize{$\textnormal{Re}\ \vartheta>0$}}
  \psfrag{13}{\footnotesize{$z_+$}}
  \psfrag{14}{\footnotesize{$z_-$}}
  \psfrag{15}{\footnotesize{$z_-$}}
  \psfrag{16}{\footnotesize{$z_+$}}
  \includegraphics[width=10cm,height=6cm]{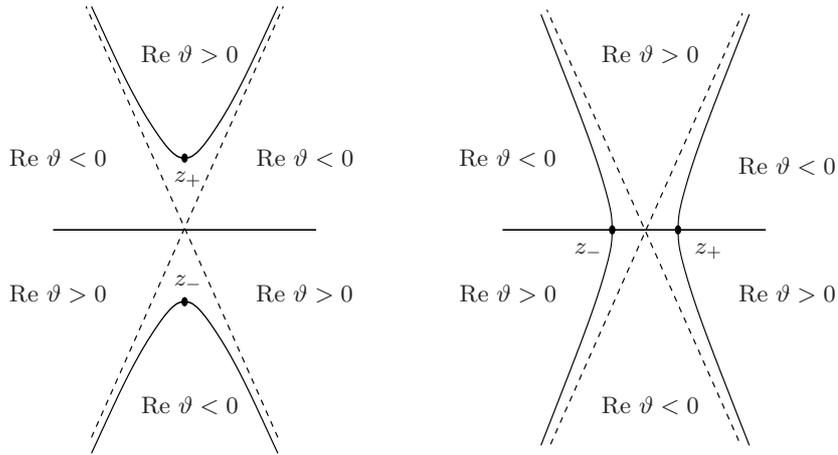}
  \end{center}
  \caption{Sign-diagram for the function $\textnormal{Re}\ \vartheta(z)$. In the left picture we indicate the location of $z_{\pm}$ as $x>0$ and in the right picture for a particular choice of $x<0$. Along the solid lines $\textnormal{Re}\ \vartheta(z)=0$ and the dashed lines resemble $\textnormal{arg}\ z =\pm\frac{\pi}{3},\pm\frac{2\pi}{3}$}
  \label{fig1}
\end{figure}
With the matrix factorization
\begin{eqnarray}\label{LDUfact}
	\begin{pmatrix}
		1-\gamma & \gamma \\
		-\gamma & 1+\gamma\\
	\end{pmatrix}&=& \begin{pmatrix}
	1 & 0 \\
	-\gamma(1-\gamma)^{-1} & 1\\
	\end{pmatrix}(1-\gamma)^{\sigma_3}\begin{pmatrix}
	1 & \gamma(1-\gamma)^{-1}\\
	0 & 1\\
	\end{pmatrix}\nonumber\\
	&=& S_LS_DS_U,
\end{eqnarray}
valid as long as $\gamma\neq 1$, we perform {\it opening of lenses} as follows. Let $\mathcal{L}_j^{\pm}$ denote the {\it upper (lower) lense}, shown in Figure \ref{fig2}, which is bounded by the contours $\gamma_{jk}^{\pm}$ with
\begin{eqnarray*}
	\gamma_{12}^+ = \left\{z\in\mathbb{C}:\ \textnormal{arg}\,z=\frac{\pi}{6}\right\},&&\gamma_{21}^+=\left\{z\in\mathbb{C}:\ \textnormal{arg}\,z=\frac{5\pi}{6}\right\},\\
	\gamma_{32}^- =\left\{z\in\mathbb{C}:\ \textnormal{arg}\,z=-\frac{5\pi}{6}\right\},&&\gamma_{41}^-=\left\{z\in\mathbb{C}:\ \textnormal{arg}\,z=-\frac{\pi}{6}\right\}.
\end{eqnarray*}
\begin{figure}[tbh]
  \begin{center}
  \psfragscanon
  \psfrag{1}{\footnotesize{$\mathcal{L}_1^+$}}
  \psfrag{2}{\footnotesize{$\mathcal{L}_2^+$}}
  \psfrag{3}{\footnotesize{$\mathcal{L}_3^-$}}
  \psfrag{4}{\footnotesize{$\mathcal{L}_4^-$}}
  \psfrag{5}{\footnotesize{$\gamma_{11}^+$}}
  \psfrag{6}{\footnotesize{$\gamma_{12}^+$}}
  \psfrag{7}{\footnotesize{$\gamma_{21}^+$}}
  \psfrag{8}{\footnotesize{$\gamma_{22}^+$}}
  \psfrag{9}{\footnotesize{$\gamma_{31}^-$}}
  \psfrag{10}{\footnotesize{$\gamma_{32}^-$}}
  \psfrag{11}{\footnotesize{$\gamma_{41}^-$}}
  \psfrag{12}{\footnotesize{$\gamma_{42}^-$}}
  \psfrag{13}{\footnotesize{$+1$}}
  \psfrag{14}{\footnotesize{$-1$}}
  \includegraphics[width=9cm,height=5cm]{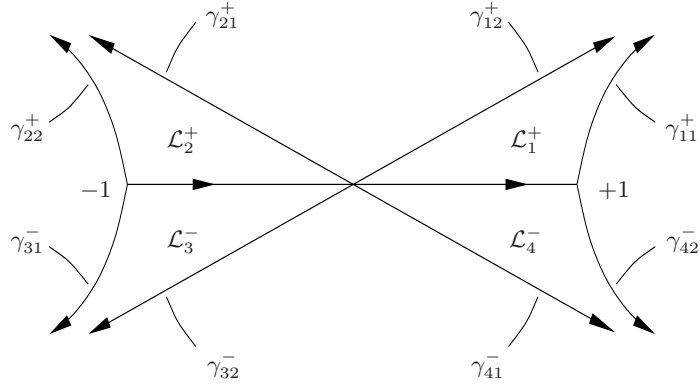}
  \end{center}
  \caption{The second transformation - opening of lenses}
  \label{fig2}
\end{figure}

Define
\begin{equation}\label{openingxpos}
	S(z) = \left\{
                                 \begin{array}{ll}
                                  T(z)S_{U}^{-1} , & \hbox{$z\in\mathcal{L}_1^+\cup\mathcal{L}_2^+$,} \smallskip \\
                                  T(z)S_L , & \hbox{$z\in\mathcal{L}_3^-\cup\mathcal{L}_4^-$,} \smallskip\\
                                  T(z), & \hbox{otherwise,}
                                 \end{array}
                               \right.\equiv S(z)L(z)
\end{equation}
then $S(z)$ solves the following RHP
\begin{itemize}
	\item $S(z)$ is analytic for  $z\in\mathbb{C}\backslash([-1,1]\cup \mathcal{C})$, with $\mathcal{C}=\bigcup_{i,j}\big(\gamma_{ji}^+\cup \gamma_{ji}^-\big)$
	\item The following jumps hold, with orientation fixed as in Figure \ref{fig2}
	\begin{equation}\label{openingjumps}
		S_+(z) = S_-(z)e^{s^3\vartheta(z)\sigma_3}\widehat{G}_Se^{-s^3\vartheta(z)\sigma_3},\hspace{0.5cm} z\in\mathbb{C}\backslash\big([-1,1]\cup\mathcal{C}\big)
	\end{equation}
	where the piecewise constant matrix $\widehat{G}_S$ is given by
	\begin{equation*}
		 \widehat{G}_S=\left\{
                                 \begin{array}{ll}
                                  S_U^{-1} , & \hbox{$z\in\gamma_{11}^+\cup\gamma_{21}^+$,} \smallskip \\
                                  S_U , & \hbox{$z\in\gamma_{12}^+\cup\gamma_{22}^+$,} \smallskip\\
                                  (1-\gamma)^{\sigma_3}, &\hbox{$z\in[-1,1]$} \smallskip\\
                                  S_L, & \hbox{$z\in\gamma_{31}^-\cup\gamma_{41}^-$,} \smallskip\\
                                  S_L^{-1}, & \hbox{$z\in\gamma_{32}^-\cup\gamma_{42}^-$.} 
                                 \end{array}
                               \right.
  \end{equation*}
  \item As $z\rightarrow\pm 1$, we have
  \begin{equation}\label{Ssingularendp}
  S(z)L^{-1}(z)e^{s^3\vartheta(z)\sigma_3} = \check{X}(zs)\Bigg[I+\frac{\gamma}{2\pi i}\begin{pmatrix}
	-1 & 1\\
	-1 & 1\\
	\end{pmatrix}\ln\left(\frac{z-1}{z+1}\right)\Bigg]
	\end{equation}
	\item At infinity, $S(z) =I+O(z^{-1}),\hspace{0.5cm} z\rightarrow\infty$.
\end{itemize}
Let us analyse the behavior of the jump matrix $G_S(z)$ in \eqref{openingjumps} along the infinite branches as $s\rightarrow\infty$. To this end recall the sign-diagram of the function $\textnormal{Re}\,\vartheta(z)$, depicted in Figure \ref{fig1}. We have in the upper half-plane
\begin{equation*}
	G_S(z) = e^{s^3\vartheta(z)\sigma_3}\begin{pmatrix}
	1 & a\\
	0 & 1\\
	\end{pmatrix}e^{-s^3\vartheta(z)\sigma_3},\hspace{0.5cm}z\in\gamma_{jk}^+,\ \ j,k=1,2
\end{equation*}
with a constant $a\in\mathbb{C}$. Since we choose $x$ from a compact subset of the real line and $s>0$ is sufficiently large, $\textnormal{Re}\,\vartheta(z)$ is always negatice on $\gamma_{jk}^+,j,k=1,2$ outside small neighborhoods around the origin and the endpoints $z=\pm 1$. Hence for such $z$
\begin{equation}\label{ap1}
	G_S(z)\longrightarrow I,\hspace{0.5cm}s\rightarrow\infty
\end{equation}
uniformly on any compact subset of the set \eqref{excset1} and the stated convergence is in fact exponentially fast. As similar statement holds on the infinite branches in the lower half-plane. There
\begin{equation*}
	G_S(z) = e^{s^3\vartheta(z)\sigma_3}\begin{pmatrix}
	1 & 0\\
	b & 1\\
	\end{pmatrix}e^{-s^3\vartheta(z)\sigma_3},\hspace{0.5cm}z\in\gamma_{jk}^-,\ \ j,k=1,2
\end{equation*}
again with some constant $b\in\mathbb{C}$. In this situation $\textnormal{Re}\,\vartheta(z)>0$ outside a small neigbhorhood of the origin and therefore
\begin{equation}\label{ap2}
	G_S(z)\longrightarrow I,\hspace{0.5cm}s\rightarrow\infty
\end{equation}
also uniformly on any compact subset of the set \eqref{excset1}. From \eqref{ap1} and \eqref{ap2} we expect, and this will be justified rigorously, that as $s\rightarrow\infty$, $S(z)$ converges to a solution of the model RHP, in which we only have to deal with the diagonal jump matrix $S_D$ on the line segment $[-1,1]$. Let us now consider this model problem.


\section{The model RHP}\label{sec7}

Find the piecewise analytic $2\times 2$ matrix valued function $M(z)$ such that
\begin{itemize}
	\item $M(z)$ is analytic for $\mathbb{C}\backslash[-1,1]$
	\item The following jump condition holds along $[-1,1]$
	\begin{equation*}
		M_+(z)=M_-(z)S_D,\hspace{0.5cm} z\in[-1,1]
	\end{equation*}
	where
	\begin{equation*}
		S_D=\big(1-\gamma\big)^{\sigma_3}
	\end{equation*}
	\item $M(z)$ has a most logarithmic singularities at the endpoints $z=\pm 1$
	\item As $z\rightarrow\infty$ we have $M(z)= I+O\big(z^{-1}\big),\ z\rightarrow\infty$
\end{itemize}
Only assuming that $\gamma\neq 1$ we can always solve this diagonal and thus quasi-scalar RHP at once (see \cite{FIKN})
\begin{eqnarray}\label{Xdiagonalpara}
	M(z)&=&\exp\bigg[\frac{1}{2\pi i}\int\limits_{-1}^1\frac{\ln(1-\gamma)^{\sigma_3}}{w-z}dw\bigg]=\bigg(\frac{z+1}{z-1}\bigg)^{-\nu\sigma_3}\\
	&&\nu = \frac{1}{2\pi i}\ln(1-\gamma),\hspace{0.5cm}\textnormal{arg}\ (1-\gamma)\in(-\pi,\pi]\nonumber
\end{eqnarray}
and $\big(\frac{z+1}{z-1}\big)^{\nu}$ is defined on $\mathbb{C}\backslash [-1,1]$ with its branch fixed by the condition $\big(\frac{z+1}{z-1}\big)^{\nu}\rightarrow 1$ as $z\rightarrow\infty$.
\begin{remark} We bring the reader's attention to the important fact, that in case $\gamma<1$, we have $\textnormal{arg}\,(1-\gamma)=0$ and $\nu$ is therefore purely imaginary. However if $\gamma>1$, then $\textnormal{arg}\,(1-\gamma)=\pi$ and $\nu$ equals
\begin{equation*}
	\nu=\frac{1}{2\pi i}\ln\left(\gamma-1\right)+\frac{1}{2}\equiv \nu_0+\frac{1}{2},\hspace{0.5cm}\nu_0\in i\mathbb{R}.
\end{equation*}
Later on we will see that this difference will have a very substantial impact on the whole steepest descent analysis.
\end{remark}

\section{Construction of a parametrix at the origin $z=0$}\label{sec8}
In this section we construct a parametrix for the origin $z=0$. This construction involves the Ablowitz-Segur solution $u=u(x,\gamma)$ of the boundary value problem
\begin{equation*}
	u_{xx}=xu+2u^3,\hspace{0.5cm}u(x)\sim\gamma\textnormal{Ai}(x),\ \ x\rightarrow+\infty,
\end{equation*}
where $\textnormal{Ai}(x)$ is the Airy-function. Viewing $x,u$ and $u_x=\frac{du(x)}{dx}$ as real parameters, consider the $2\times 2$ system of linear ordinary differential equations,
\begin{equation}\label{Absegur1}
	\frac{\partial \Psi}{\partial \zeta} = A(\zeta,x)\Psi,\ \ A(\zeta,x)=-4i\zeta^2\sigma_3+4i\zeta\begin{pmatrix}
	0 & u\\
	-u & 0\\
	\end{pmatrix}+\begin{pmatrix}
	-ix-2iu^2 & -2u_x\\
	-2u_x & ix +2iu^2\\
	\end{pmatrix}.
\end{equation}
The given system has precisely one irregular singular point of Poincar\'e rank $3$ at infinity, thus from the classical theory of ordinary differential equations in the complex plane, we obtain existence of seven canonical solutions $\Psi_n(\zeta)$ which are fixed uniquely by their asymptotics (for more detail see \cite{FIKN})
\begin{equation*}
	\Psi_n(\zeta)\sim\Big(I+O\big(\zeta^{-1}\big)\Big)e^{-i(\frac{4}{3}\zeta^3+x\zeta)\sigma_3},\hspace{0.5cm}\zeta\rightarrow\infty,\ \zeta\in\Omega_n
\end{equation*}
in the canonical sectors $\Omega_n$ given in Figure \ref{fig6}
\begin{equation*}
	\Omega_n=\Big\{\zeta\in\mathbb{C}\ |\ \textnormal{arg}\
    \zeta\in\Big(\frac{\pi}{3}(n-2),\frac{\pi}{3}n\Big),
    n=1,\ldots,7\Big\}.
\end{equation*}
\begin{figure}[tbh]
  \begin{center}
  \psfragscanon
  \psfrag{1}{\footnotesize{$\Omega_2$}}
  \psfrag{2}{\footnotesize{$0<\textnormal{arg}\ \zeta<\frac{2\pi}{3}$}}
  \psfrag{3}{\footnotesize{$\Omega_3$}}
  \psfrag{4}{\footnotesize{$\frac{\pi}{3}<\textnormal{arg}\ \zeta<\pi$}}
  \psfrag{5}{\footnotesize{$\Omega_4$}}
  \psfrag{6}{\footnotesize{$\frac{2\pi}{3}<\textnormal{arg}\ \zeta<\frac{4\pi}{3}$}}
  \psfrag{7}{\footnotesize{$\pi<\textnormal{arg}\ \zeta<\frac{5\pi}{3}$}}
  \psfrag{8}{\footnotesize{$\Omega_5$}}
  \psfrag{9}{\footnotesize{$\Omega_6$}}
  \psfrag{10}{\footnotesize{$\frac{4\pi}{3}<\textnormal{arg}\
  \zeta<2\pi$}}
  \psfrag{11}{\footnotesize{$\Omega_7$}}
  \psfrag{12}{\footnotesize{$\frac{5\pi}{3}<\textnormal{arg}\ \zeta<\frac{7\pi}{3}$}}
  \psfrag{13}{\footnotesize{$\Omega_1$}}
  \psfrag{14}{\footnotesize{$-\frac{\pi}{3}<\textnormal{arg}\ \zeta<\frac{\pi}{3}$}}
  \includegraphics[width=11cm,height=6cm]{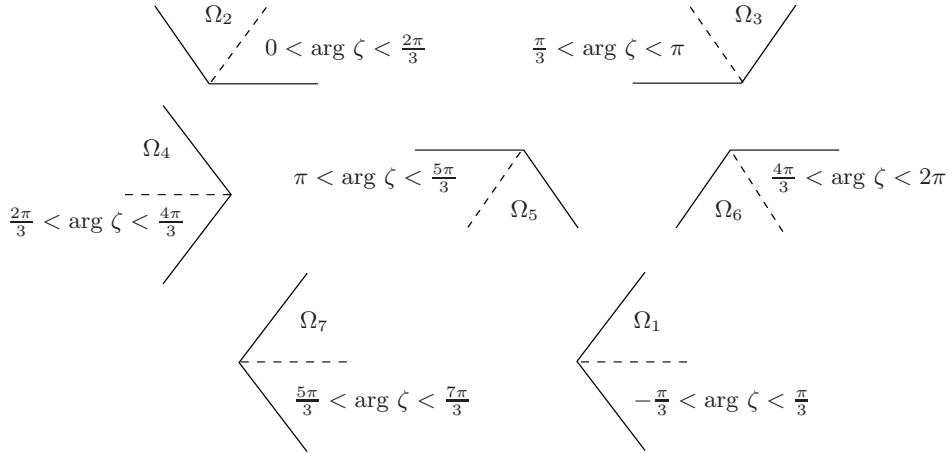}
  \end{center}
  \caption{Canonical sectors of system \eqref{Absegur1} with the dashed lines indicating where Re $\zeta^3=0$}
  \label{fig6}
\end{figure}

Moreover the presence of an irregular singularity gives us a non-trivial Stokes phenomenon described by the Stokes matrices $S_n$:
\begin{equation}\label{Stokesmat}
	S_n=\big(\Psi_n(\zeta)\big)^{-1}\Psi_{n+1}(\zeta).
\end{equation}
In the given situation \eqref{Absegur1} with the Ablowitz-Segur solution $u=u(x,\gamma)$, these multipliers are
\begin{equation}\label{Stokesch}
	S_1 = \begin{pmatrix}
	1 & 0\\
	-i\gamma & 1\\
	\end{pmatrix},\ S_2=\begin{pmatrix}
	1 & 0\\
	0 & 1\\
	\end{pmatrix},\ S_4=\begin{pmatrix}
	1 & i\gamma\\
	0 & 1\\
	\end{pmatrix},\ S_3=\bar{S}_1,\ S_5=\bar{S}_2,\ S_6=\bar{S}_4.
\end{equation}
We now choose $P_{II}(\zeta)$ to be the first canonical solution of \eqref{Absegur1} corresponding to the specific choice \eqref{Stokesch}, i.e.
\begin{equation*}
	P_{II}(\zeta) = \Psi_1(\zeta),\hspace{0.5cm}\textnormal{arg}\,\zeta\in\left(-\frac{\pi}{6},\frac{\pi}{6}\right)
\end{equation*}
and assemble a piecewise analytic model function as follows
\begin{equation}\label{Absegur2}
	\widetilde{P}_{II}^{RH}(\zeta)=\left\{
                                   \begin{array}{ll}
                                     P_{II}(\zeta), & \hbox{$\textnormal{arg}\ \zeta\in(-\frac{\pi}{6},\frac{\pi}{6})\cup(\frac{5\pi}{6},\frac{7\pi}{6})$,} \smallskip\\
                                     P_{II}(\zeta)S_1, & \hbox{$\textnormal{arg}\ \zeta\in(\frac{\pi}{6},\frac{5\pi}{6})$,} \smallskip\\
                                     P_{II}(\zeta)S_4, & \hbox{$\textnormal{arg}\ \zeta\in(\frac{7\pi}{6},\frac{11\pi}{6})$.}
                                   \end{array}
                                 \right.
\end{equation}
One checks directly that $\widetilde{P}_{II}^{RH}(\zeta)$ solves the model RHP shown in Figure \ref{fig7}
\begin{figure}[tbh]
  \begin{center}
  \psfragscanon
  \psfrag{1}{$\bigl(\begin{smallmatrix}
  1 & 0\\
  -i\gamma & 1\\
  \end{smallmatrix}\bigr)$}
  \psfrag{2}{$\bigl(\begin{smallmatrix}
  1 & -i\gamma\\
  0 & 1\\
  \end{smallmatrix}\bigr)$}
  \psfrag{3}{$\bigl(\begin{smallmatrix}
  1 & 0\\
  i\gamma & 1\\
  \end{smallmatrix}\bigr)$}
  \psfrag{4}{$\bigl(\begin{smallmatrix}
  1 & i\gamma\\
  0 & 1\\
  \end{smallmatrix}\bigr)$}
  \includegraphics[width=7cm,height=3cm]{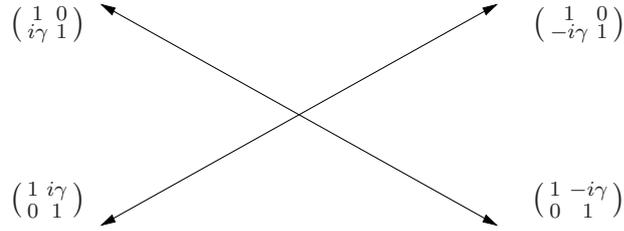}
  \end{center}
  \caption{The model RHP near $z=0$ which can be solved explicitly using the real-valued Ablowitz-Segur solution of the second Painlev\'e equation}
  \label{fig7}
\end{figure}
\begin{itemize}
	\item $\widetilde{P}_{II}^{RH}(\zeta)$ is analytic for $\zeta\in\mathbb{C}\backslash \{\textnormal{arg}\ \zeta=\frac{\pi}{6},\frac{5\pi}{6},\frac{7\pi}{6},\frac{11\pi}{6}\}$
	\item On the infinite rays, the following jumps hold
	\begin{eqnarray*}
		\big(\widetilde{P}_{II}^{RH}(\zeta)\big)_+&=&\big(\widetilde{P}_{II}^{RH}(\zeta)\big)_-S_1,\hspace{0.8cm}\textnormal{arg}\ \zeta=\frac{\pi}{6}\\
		\big(\widetilde{P}_{II}^{RH}(\zeta)\big)_+&=&\big(\widetilde{P}_{II}^{RH}(\zeta)\big)_-S_1^{-1},\hspace{0.5cm}\textnormal{arg}\ \zeta=\frac{5\pi}{6}\\
		\big(\widetilde{P}_{II}^{RH}(\zeta)\big)_+&=&\big(\widetilde{P}_{II}^{RH}(\zeta)\big)_-S_4,\hspace{0.75cm}\textnormal{arg}\ \zeta=\frac{7\pi}{6}\\
		\big(\widetilde{P}_{II}^{RH}(\zeta)\big)_+&=&\big(\widetilde{P}_{II}^{RH}(\zeta)\big)_-S_4^{-1},\hspace{0.5cm}\textnormal{arg}\ \zeta=\frac{11\pi}{6}\\
	\end{eqnarray*}
	\item In terms of the previous choice \eqref{Stokesch}, we have the following uniform asymptotics, valid in a full neighborhood of infinity
	\begin{equation}\label{Absegur3}
		\widetilde{P}_{II}^{RH}(\zeta)\sim\Big(I+\frac{m_1}{\zeta}+\frac{m_2}{\zeta^2}+\frac{m_3}{\zeta^3}+O\big(\zeta^{-4}\big)\Big)e^{-i(\frac{4}{3}\zeta^3+x\zeta)\sigma_3},
		\hspace{0.5cm}\zeta\rightarrow\infty
	\end{equation}
	with
	\begin{equation*}
		m_1=\frac{1}{2}\begin{pmatrix}
		-iv & u\\
		u & iv\\
		\end{pmatrix},\hspace{0.3cm} m_2=\frac{1}{8}\begin{pmatrix}
		u^2-v^2 & 2i(u_x+uv)\\
		-2i(u_x+uv) & u^2-v^2\\
		\end{pmatrix},
	\end{equation*}
	and
	\begin{equation*}
		m_3 = \frac{1}{48}\begin{pmatrix}
		i(v^3-3vu^2+2(xv-uu_x)) & -3(u(u^2+v^2)+2(vu_x+xu))\\
		-3(u(u^2+v^2)+2(vu_x+xu)) & -i(v^3-3vu^2+2(xv-uu_x))\\
		\end{pmatrix}.
	\end{equation*}
	where we introduced
	\begin{equation*}
		 v=(u_x)^2-xu^2-u^4.
	\end{equation*}
\end{itemize}
\begin{remark} In view of our discussion in section \ref{sec1}, the Ablowitz-Segur solution $u=u(x,\gamma)$ might have poles on the real line. However this can only happen in case $\gamma>1$ and in this situation we restrict ourselves to values of $(\gamma,x)$ from \eqref{excset2}. Thus in either of our cases, $u=u(x,\gamma)$ is smooth in $x$ and therefore the model function $\widetilde{P}_{II}^{RH}(\zeta)=\widetilde{P}_{II}^{RH}(\zeta;x)$ well defined.
\end{remark}
The model function \eqref{Absegur2} will now be used to construct the parametrix to the solution of the original $S$-RHP in a neighborhood of $z=0$. First set
\begin{equation*}
	P_{II}^{RH}(\zeta)=\left\{
                                   \begin{array}{ll}
                                     e^{\pi i\nu\sigma_3}\widetilde{P}_{II}^{RH}(\zeta)e^{-\pi i\nu\sigma_3}, & \hbox{$\textnormal{Im}\ \zeta>0$;} \smallskip\\
                                     e^{\pi i\nu\sigma_3}\widetilde{P}_{II}^{RH}(\zeta)e^{\pi i\nu\sigma_3}, & \hbox{$\textnormal{Im}\ \zeta<0$;}
                                   \end{array}
                                 \right.
\end{equation*}
where $\nu$ was introduced in \eqref{Xdiagonalpara}. This leads to a model RHP with jumps on the positive oriented real line
\begin{equation*}
	\big(P_{II}^{RH}(\zeta)\big)_+=\big(P_{II}^{RH}(\zeta)\big)_-(1-\gamma)^{-\sigma_3}
\end{equation*}
as well as on the infinite rays $\textnormal{arg}\ \zeta=\frac{\pi}{6},\frac{5\pi}{6},\frac{7\pi}{6},\frac{11\pi}{6}$
\begin{eqnarray*}
		\big(P_{II}^{RH}(\zeta)\big)_+&=&\big(P_{II}^{RH}(\zeta)\big)_-\begin{pmatrix}
		1 & 0 \\
		-i\gamma(1-\gamma)^{-1} & 1\\
		\end{pmatrix},\hspace{0.4cm}\textnormal{arg}\ \zeta=\frac{\pi}{6}\\
		\big(P_{II}^{RH}(\zeta)\big)_+&=&\big(P_{II}^{RH}(\zeta)\big)_-\begin{pmatrix}
		1 & 0\\
		i\gamma(1-\gamma)^{-1} & 1\\
		\end{pmatrix},\hspace{0.5cm}\textnormal{arg}\ \zeta=\frac{5\pi}{6}\\
		\big(P_{II}^{RH}(\zeta)\big)_+&=&\big(P_{II}^{RH}(\zeta)\big)_-\begin{pmatrix}
		1 & i\gamma(1-\gamma)^{-1}\\
		0 & 1\\
		\end{pmatrix},\hspace{0.75cm}\textnormal{arg}\ \zeta=\frac{7\pi}{6}\\
		\big(P_{II}^{RH}(\zeta)\big)_+&=&\big(P_{II}^{RH}(\zeta)\big)_-\begin{pmatrix}
		1 & -i\gamma(1-\gamma)^{-1}\\
		0 & 1\\
		\end{pmatrix},\hspace{0.5cm}\textnormal{arg}\ \zeta=\frac{11\pi}{6}\\
\end{eqnarray*}
and with behavior at infinity
\begin{equation*}
	P_{II}^{RH}(\zeta)=\Big(I+\frac{\tilde{m}_1}{\zeta}+\frac{\tilde{m}_2}{\zeta^2}+\frac{\tilde{m}_3}{\zeta^3}+O\big(\zeta^{-4}\big)\Big)e^{-i(\frac{4}{3}\zeta^3+x\zeta)\sigma_3}
											\left\{\begin{array}{ll}
                                     I, & \hbox{$\textnormal{Im}\ \zeta>0$;} \smallskip\\
                                     e^{2\pi i\nu\sigma_3}, & \hbox{$\textnormal{Im}\ \zeta<0$;}
                                   \end{array}
                                \right.
\end{equation*}
where
\begin{equation*}
	\tilde{m}_1 = \frac{1}{2}\begin{pmatrix}
	-iv & ue^{2\pi i\nu}\\
	ue^{-2\pi i\nu}& iv\\
	\end{pmatrix},\hspace{0.5cm} \tilde{m}_2=\frac{1}{8}\begin{pmatrix}
	u^2-v^2&2i(u_x+uv)e^{2\pi i\nu}\\
	-2i(u_x+uv)e^{-2\pi i\nu} & u^2-v^2\\
	\end{pmatrix}.
\end{equation*}
and
\begin{equation*}
	\tilde{m}_3=\frac{1}{48}\begin{pmatrix}
		i(v^3-3vu^2+2(xv-uu_x)) & -3(u(u^2+v^2)+2(vu_x+xu))e^{2\pi i\nu}\\
		-3(u(u^2+v^2)+2(vu_x+xu))e^{-2\pi i\nu} & -i(v^3-3vu^2+2(xv-uu_x))\\
		\end{pmatrix}.
	\end{equation*}
Secondly define
\begin{equation}\label{Absegurchange}
	\zeta(z)=sz,\hspace{0.5cm} |z|<r
\end{equation}
which yields a locally conformal change of variables and which allows us to define the parametrix $U(z)$ near $z=0$ by the formula
\begin{equation}\label{Absegurparametrix}
	U(z)=\sigma_1B_0(z)e^{-i\frac{\pi}{4}\sigma_3}P_{II}^{RH}\big(\zeta(z)\big)e^{i(\frac{4}{3}\zeta(z)+x\zeta(z))\sigma_3}e^{i\frac{\pi}{4}\sigma_3}\sigma_1,\hspace{0.25cm} |z|<r,\hspace{0.25cm}\sigma_1 = \begin{pmatrix}
	0 & 1\\
	1 & 0\\
	\end{pmatrix}
\end{equation}
with $\zeta(z)$ as in \eqref{Absegurchange} and the matrix multiplier
\begin{equation}\label{Absegurmultiplier}
	B_0(z)=\bigg(\frac{z+1}{z-1}\bigg)^{\nu\sigma_3}\left\{\begin{array}{ll}
                                     I, & \hbox{$\textnormal{Im}\ z>0$;} \smallskip\\
                                     e^{-2\pi i\nu\sigma_3}, & \hbox{$\textnormal{Im}\ z<0$;}
                                   \end{array}
                                \right.\hspace{1cm} B_0(0)=e^{-i\pi\nu\sigma_3}.
\end{equation}
By construction, in particular since $B_0(z)$ is analytic in a neighborhood of $z=0$, the parametrix $U(z)$ has jumps along the curves depicted in Figure \ref{fig8}, and we can always locally match the latter curves with the jump curves of the original RHP. Also these jumps are described by the same Stokes matrices as in the original $S$-RHP, since the previously stated jumps of $P_{II}^{RH}(\zeta)$ will be conjugated with $e^{i\frac{\pi}{4}\sigma_3}\sigma_1e^{-s^3\vartheta(z)\sigma_3}$ which precisely matches model jumps with original ones. Thus the ratio of $S(z)$ with $U(z)$ is locally analytic, i.e.
\begin{equation*}
	S(z)=N_0(z)U(z),\hspace{0.5cm}|z|<r<\frac{1}{2}
\end{equation*}
\begin{figure}[tbh]
  \begin{flushleft}
  \psfragscanon
  \psfrag{1}{\footnotesize{$e^{s^3\vartheta(z)\sigma_3}S_Ue^{-s^3\vartheta(z)\sigma_3}$}}
  \psfrag{2}{\footnotesize{$e^{s^3\vartheta(z)\sigma_3}S_Le^{-s^3\vartheta(z)\sigma_3}$}}
  \psfrag{3}{\footnotesize{$e^{s^3\vartheta(z)\sigma_3}S_U^{-1}e^{-s^3\vartheta(z)\sigma_3}$}}
  \psfrag{4}{\footnotesize{$e^{s^3\vartheta(z)\sigma_3}S_L^{-1}e^{-s^3\vartheta(z)\sigma_3}$}}
  \psfrag{5}{\footnotesize{$S_D$}}
  \psfrag{6}{\footnotesize{$S_D$}}
  \hspace{1.5cm}\includegraphics[width=9cm,height=3cm]{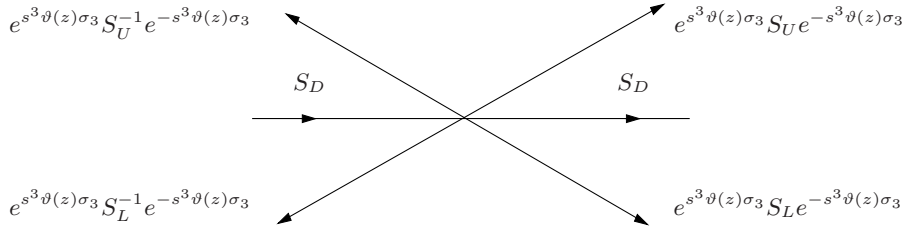}
  \end{flushleft}
  \caption{Jump graph of the parametrix $U(z)$}
  \label{fig8}
\end{figure}

Let us explain the role of the left multiplier $B_0(z)$ in the definition \eqref{Absegurparametrix}. Observe that
\begin{equation*}
	\sigma_1B_0(z)\sigma_1 = M(z)\left\{\begin{array}{ll}
                                     I, & \hbox{$\textnormal{Im}\ \zeta>0$;} \smallskip\\
                                     e^{2\pi i\nu\sigma_3}, & \hbox{$\textnormal{Im}\ \zeta<0$.}
                                   \end{array}
                                \right.
\end{equation*}
This relation together with the asymptotic condition \eqref{Absegur3} implies that,
\begin{eqnarray}\label{Absegurmatch}
	&&U(z) = \sigma_1B_0(z)e^{-i\frac{\pi}{4}\sigma_3}\Big(I+\frac{\tilde{m}_1}{\zeta}+\frac{\tilde{m}_2}{\zeta^2}+\frac{\tilde{m}_3}{\zeta^3}+O\big(\zeta^{-4}\big)\Big)e^{i\frac{\pi}{4}\sigma_3}B_0^{-1}(z)\sigma_1M(z)\nonumber\\
	&&= \bigg[I+\frac{i}{2\zeta}B_0(z)^{-1}\begin{pmatrix}
	v & ue^{-2\pi i\nu}\\
	-ue^{2\pi i\nu} & -v\\
	\end{pmatrix}B_0(z)\nonumber\\
	&&+\frac{1}{8\zeta^2}B_0(z)^{-1}\begin{pmatrix}
	u^2-v^2 & 2(u_x+uv)e^{-2\pi i\nu}\\
	 2(u_x+uv)e^{2\pi i\nu}& u^2-v^2\\
	 \end{pmatrix}B_0(z)+\frac{i}{48\zeta^3}B_0(z)^{-1}\nonumber\\
	 &&\times
	 \begin{pmatrix}
	 -(v^3-3vu^2+2(xv-uu_x)) & -3(u(u^2+v^2)+2(vu_x+xu))e^{-2\pi i\nu}\\
	 3(u(u^2+v^2)+2(vu_x+xu))e^{2\pi i\nu} & v^3-3vu^2+2(xv-uu_x)\\
	 \end{pmatrix}B_0(z)\nonumber\\
	 &&+O\big(\zeta^{-3}\big)\bigg]M(z)
\end{eqnarray}
as $s\rightarrow\infty$ and $0<r_1\leq |z|\leq r_2<1$ (so $|\zeta|\rightarrow\infty$). Since the function $\zeta(z)$ is of order $O\left(s^1\right)$ on the latter annulus and $B_0(z)$ is bounded, equation \eqref{Absegurmatch} yields the matching relation between the model functions $U(z)$ and $M(z)$,
\begin{equation*}
	U(z) = \big(I+o(1)\big)M(z),\hspace{0.5cm} s\rightarrow\infty,\ \ 0 < r_1 \leq |z|\leq r_2<1
\end{equation*}
which is crucial for the successful implementation of the nonlinear steepest descent method as we shall see later on. This is the reason for chosing the left multiplier $B_0(z)$ in \eqref{Absegurparametrix} in the form \eqref{Absegurmultiplier}.

\section{Construction of a parametrix at the edge point $z=+1$}\label{sec9}

Fix a small neighborhood $\mathcal{U}$ of the point $z=+1$ and observe that
\begin{equation*}
	 \vartheta(z) = \vartheta(1)+i\Big(4+\frac{x}{s^2}\Big)(z- 1)+O\big((z- 1)^2\big),
\end{equation*}
as $z\in\mathcal{U}$ and from \eqref{Ssingularendp}
\begin{equation*}
	S(z)=O\big(\ln(z-1)\big),\hspace{0.5cm} z\in\mathcal{U}.
\end{equation*}
Both observations suggest to use the confluent hypergeometric function $U(a,b;\zeta)$ for our construction. This idea can be justified rigorously as follows. Recall that the listed confluent hypergeometric function is defined as unique solution to Kummer's equation
\begin{equation*}
	zw''+(b-z)w'-aw=0
\end{equation*}
satisfying the asymptotic condition as $\zeta\rightarrow\infty$ and $-\frac{3\pi}{2}<\textnormal{arg}\ \zeta<\frac{3\pi}{2}$ (see \cite{BE})
\begin{equation*}
	U(a,b;\zeta)\sim \zeta^{-a}\bigg(1-\frac{a(1+a-b)}{\zeta}+\frac{a(a+1)(1+a-b)(2+a-b)}{2\zeta^2}+O\big(\zeta^{-3}\big)\bigg).
\end{equation*}
Secondly, using the notation $U(a,\zeta)\equiv U(a,1;\zeta)$, the following monodromy relation holds on the entire universal covering of the punctured plane
\begin{equation}\label{conflumono}
	U(1-a,e^{i\pi}\zeta)=e^{2\pi ia}U(1-a,e^{-i\pi}\zeta)-e^{i\pi a}\frac{2\pi i}{\Gamma^2(1-a)}U(a,\zeta)e^{-\zeta}
\end{equation}
and finally we have an expansion at the origin (compare to \eqref{Xendpoint})
\begin{equation}\label{confluorigin}
	U(a,\zeta)= c_0+c_1\ln\zeta +c_2\zeta+c_3\zeta\ln\zeta +O\big(\zeta^2\ln\zeta),\hspace{0.5cm}\zeta\rightarrow 0
\end{equation}
with coefficients $c_i$ given as
\begin{equation*}
	c_0=-\frac{1}{\Gamma(a)}\big(\psi(a)+2\gamma_E),\ c_1=-\frac{1}{\Gamma(a)},\ c_2=-\frac{a}{\Gamma(a)}\big(\psi(a+1)+2\gamma_E-2\big),\ c_3=-\frac{a}{\Gamma(a)}
\end{equation*}
where $\gamma_E$ is Euler's constant and $\psi(z) =\frac{\Gamma'(z)}{\Gamma(z)}$. Remembering the latter properties we now introduce the following matrix-valued function on the punctured plane $\zeta\in\mathbb{C}\backslash\{0\}$ (cf. \cite{IK})
\begin{eqnarray}\label{PCH}
    P_{CH}(\zeta)&=&\begin{pmatrix}
                    U(\nu,e^{i\frac{\pi}{2}}\zeta)e^{2\pi i\nu}e^{-i\frac{\zeta}{2}} & -U(1-\nu,e^{-i\frac{\pi}{2}}\zeta)e^{\pi i\nu}e^{i\frac{\zeta}{2}}\frac{\Gamma(1-\nu)}{\Gamma(\nu)} \\
                    -U(1+\nu,e^{i\frac{\pi}{2}}\zeta)e^{\pi i\nu}e^{-i\frac{\zeta}{2}}\frac{\Gamma(1+\nu)}{\Gamma(-\nu)} & U(-\nu,e^{-i\frac{\pi}{2}}\zeta)e^{i\frac{\zeta}{2}}  \\
                  \end{pmatrix}\nonumber\\
                  &&\times e^{-i\frac{\pi}{2}(\frac{1}{2}-\nu)\sigma_3},\hspace{0.5cm}-\pi<\textnormal{arg}\,\zeta\leq\pi.
\end{eqnarray}
Let us collect the following asymptotic expansions. First in the sector $-\frac{\pi}{2}<\textnormal{arg}\ \zeta <\frac{\pi}{2}$ 
\begin{eqnarray*}
	P_{CH}(\zeta) &=& \Bigg[I+\frac{i}{\zeta}\begin{pmatrix}
	\nu^2& -\frac{\Gamma(1-\nu)}{\Gamma(\nu)}e^{\pi i\nu}\\
	\frac{\Gamma(1+\nu)}{\Gamma(-\nu)}e^{-\pi i\nu} & -\nu^2\\
	\end{pmatrix}\\
	&&+\frac{1}{\zeta^2}\begin{pmatrix}
	-\frac{\nu^2}{2}(1+\nu)^2 & -\frac{\Gamma(1-\nu)}{\Gamma(\nu)}(1-\nu)^2e^{\pi i\nu} \\
	-\frac{\Gamma(1+\nu)}{\Gamma(-\nu)}(1+\nu)^2e^{-\pi i\nu}& -\frac{\nu^2}{2}(1-\nu)^2\\
	\end{pmatrix}+O\big(\zeta^{-3}\big)\Bigg]\zeta^{-\nu\sigma_3}\\
	&&\times e^{-i\frac{\zeta}{2}\sigma_3}e^{-i\frac{\pi}{2}(\frac{1}{2}-\nu)\sigma_3}\begin{pmatrix}
	e^{i\frac{3\pi}{2}\nu} & 0 \\
	0 & e^{-i\frac{\pi}{2}\nu}\\
	\end{pmatrix},\hspace{0.5cm}\zeta\rightarrow\infty.
\end{eqnarray*}
For another sector, say $\frac{\pi}{4}<\textnormal{arg}\ \zeta<\frac{5\pi}{4}$, we use \eqref{conflumono} in the first column of \eqref{PCH} and obtain instead 
\begin{eqnarray*}
	P_{CH}(\zeta) &=& \Bigg[I+\frac{i}{\zeta}\begin{pmatrix}
	\nu^2& -\frac{\Gamma(1-\nu)}{\Gamma(\nu)}e^{\pi i\nu}\\
	\frac{\Gamma(1+\nu)}{\Gamma(-\nu)}e^{-\pi i\nu} & -\nu^2\\
	\end{pmatrix}\\
	&&+\frac{1}{\zeta^2}\begin{pmatrix}
	-\frac{\nu^2}{2}(1+\nu)^2 & -\frac{\Gamma(1-\nu)}{\Gamma(\nu)}(1-\nu)^2e^{\pi i\nu} \\
	-\frac{\Gamma(1+\nu)}{\Gamma(-\nu)}(1+\nu)^2e^{-\pi i\nu}& -\frac{\nu^2}{2}(1-\nu)^2\\
	\end{pmatrix}+O\big(\zeta^{-3}\big)\Bigg]\\
	&&\times \zeta^{-\nu\sigma_3}e^{-i\frac{\zeta}{2}\sigma_3}e^{-i\frac{\pi}{2}(\frac{1}{2}-\nu)\sigma_3}\begin{pmatrix}
	e^{i\frac{3\pi}{2}\nu} & 0 \\
	0 & e^{-i\frac{\pi}{2}\nu}\\
	\end{pmatrix}\begin{pmatrix}
	1 & 0\\
	-\frac{2\pi e^{i\pi\nu}}{\Gamma(1-\nu)\Gamma(\nu)} & 1\\
	\end{pmatrix},\hspace{0.2cm}\zeta\rightarrow\infty,
\end{eqnarray*}
as well as for $-\frac{5\pi}{4}<\textnormal{arg}\ \zeta<-\frac{\pi}{4}$ with a similar argument in the second column of \eqref{PCH}
\begin{eqnarray*}
	P_{CH}(\zeta) &=& \Bigg[I+\frac{i}{\zeta}\begin{pmatrix}
	\nu^2& -\frac{\Gamma(1-\nu)}{\Gamma(\nu)}e^{\pi i\nu}\\
	\frac{\Gamma(1+\nu)}{\Gamma(-\nu)}e^{-\pi i\nu} & -\nu^2\\
	\end{pmatrix}\\
	&&+\frac{1}{\zeta^2}\begin{pmatrix}
	-\frac{\nu^2}{2}(1+\nu)^2 & -\frac{\Gamma(1-\nu)}{\Gamma(\nu)}(1-\nu)^2e^{\pi i\nu} \\
	-\frac{\Gamma(1+\nu)}{\Gamma(-\nu)}(1+\nu)^2e^{-\pi i\nu}& -\frac{\nu^2}{2}(1-\nu)^2\\
	\end{pmatrix}+O\big(\zeta^{-3}\big)\Bigg]\\
	&&\times\zeta^{-\nu\sigma_3}e^{-i\frac{\zeta}{2}\sigma_3}e^{-i\frac{\pi}{2}(\frac{1}{2}-\nu)\sigma_3}\begin{pmatrix}
	e^{i\frac{3\pi}{2}\nu} & 0 \\
	0 & e^{-i\frac{\pi}{2}\nu}\\
	\end{pmatrix}\begin{pmatrix}
	1 & \frac{2\pi e^{-3\pi i\nu}}{\Gamma(1-\nu)\Gamma(\nu)}\\
	0 & 1\\
	\end{pmatrix},\hspace{0.2cm}\zeta\rightarrow\infty.
\end{eqnarray*}
Also of interest is the following exact monodromy relation
\begin{equation}\label{PCHmono}
    P_{CH}(\zeta)=P_{CH}(e^{-2\pi i}\zeta)\begin{pmatrix}
                                            e^{-2\pi i\nu}+\Big(\frac{2\pi}{\Gamma(1-\nu)\Gamma(\nu)}\Big)^2 & -\frac{2\pi e^{-i\pi\nu}}{\Gamma(1-\nu)\Gamma(\nu)} \\
                                            -\frac{2\pi e^{3\pi i\nu}}{\Gamma(1-\nu)\Gamma(\nu)} & e^{2\pi i\nu} \\
                                          \end{pmatrix}.
\end{equation}
Let us now assemble the model function 
\begin{equation}\label{PCHRH}
    P_{CH}^{RH}(\zeta)=\left\{
                         \begin{array}{ll}
                           P_{CH}(\zeta)\Bigl(\begin{smallmatrix}
                          1 & 0 \\
                          \frac{2\pi e^{i\pi\nu}}{\Gamma(1-\nu)\Gamma(\nu)} & 1 \\
                        \end{smallmatrix}\Bigr)\Bigl(\begin{smallmatrix}
                          e^{-\frac{3\pi}{2}i\nu} & 0 \\
                          0 & e^{\frac{\pi}{2}i\nu} \\
                        \end{smallmatrix}\Bigr), & \hbox{arg $\zeta\in(\frac{\pi}{3},\pi)$,}\bigskip \\
                          P_{CH}(\zeta)\Bigl(\begin{smallmatrix}
                          1 & -\frac{2\pi e^{-3\pi i\nu}}{\Gamma(1-\nu)\Gamma(\nu)} \\
                          0 & 1 \\
                        \end{smallmatrix}\Bigr)\Bigl(\begin{smallmatrix}
                          e^{-\frac{3\pi}{2}i\nu} & 0 \\
                          0 & e^{\frac{\pi}{2}i\nu} \\
                        \end{smallmatrix}\Bigr) , & \hbox{arg $\zeta\in(-\pi,-\frac{\pi}{3})$,}\bigskip \\
                          P_{CH}(\zeta)\Bigl(\begin{smallmatrix}
                          e^{-\frac{3\pi}{2}i\nu} & 0 \\
                          0 & e^{\frac{\pi}{2}i\nu} \\
                        \end{smallmatrix}\Bigr), & \hbox{arg $\zeta\in(-\frac{\pi}{3},\frac{\pi}{3})$.}
                         \end{array}
                       \right.
\end{equation}
which solves the RHP depicted in Figure \ref{fig9}.
\begin{figure}[tbh]
  \begin{center}
  \psfragscanon
  \psfrag{1}{\footnotesize{$(1-\gamma)^{-\sigma_3}$}}
  \psfrag{2}{$\Bigl(\begin{smallmatrix}
  1 & 0\\
  \frac{2\pi e^{-i\pi\nu}}{\Gamma(1-\nu)\Gamma(\nu)} & 1\\
  \end{smallmatrix}\Bigr)$}
  \psfrag{3}{$\Bigl(\begin{smallmatrix}
  1 & \frac{2\pi e^{-i\pi\nu}}{\Gamma(1-\nu)\Gamma(\nu)}\\
  0 & 1\\
  \end{smallmatrix}\Bigr)$}
  \includegraphics[width=4cm,height=3cm]{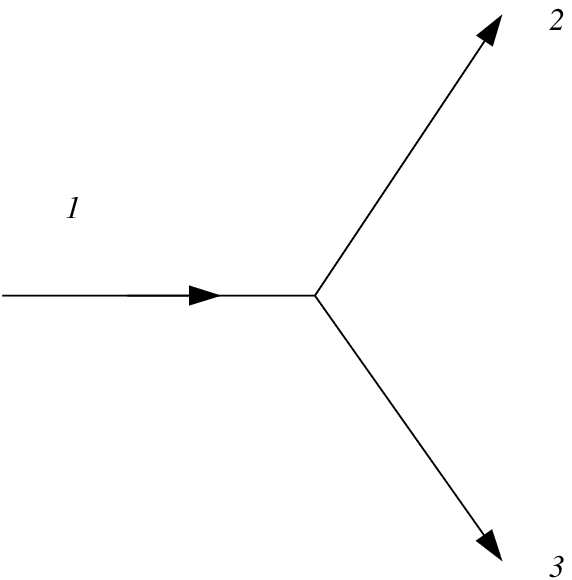}
  \end{center}
  \caption{The model RHP near $z=+1$ which can be solved explicitly using confluent hypergeometric functions}
  \label{fig9}
\end{figure}

In more detail, $P_{CH}^{RH}(\zeta)$ possesses the following analytic properties
\begin{itemize}
    \item $P_{CH}^{RH}(\zeta)$ is analytic for $\zeta\in\mathbb{C}\backslash\{\textnormal{arg}\ \zeta=-\pi,-\frac{\pi}{3},\frac{\pi}{3}\}$
    \item The following jumps are valid, orienting the jump contours as shown in Figure \ref{fig9}
		\begin{eqnarray*}
    	\big(P_{CH}^{RH}(\zeta)\big)_+&=&\big(P_{CH}^{RH}(\zeta)\big)_-(1-\gamma)^{-\sigma_3},\ \ \ \textnormal{arg}\ \zeta=-\pi\\ \smallskip
      \big(P_{CH}^{RH}(\zeta)\big)_+&=&\big(P_{CH}^{RH}(\zeta)\big)_-\begin{pmatrix}
                                                                     1 & 0 \\
                                                                     \frac{2\pi e^{-i\pi\nu}}{\Gamma(1-\nu)\Gamma(\nu)} & 1 \\
                                                                   \end{pmatrix},\
\ \ \ \textnormal{arg}\ \zeta=\frac{\pi}{3}\smallskip\\
    \big(P_{CH}^{RH}(\zeta)\big)_+&=&\big(P_{CH}^{RH}(\zeta)\big)_-\begin{pmatrix}
                                                                     1 & \frac{2\pi e^{-i\pi\nu}}{\Gamma(1-\nu)\Gamma(\nu)} \\
                                                                     0 & 1 \\
                                                                   \end{pmatrix},\
\ \ \ \textnormal{arg}\ \zeta=-\frac{\pi}{3}
\end{eqnarray*}
		and in virtue of the classical identity
		\begin{equation*}
			\Gamma(1-\nu)\Gamma(\nu)=\frac{\pi}{\sin\pi\nu}
		\end{equation*}
		the entry of the latter triangular matrices equal
		\begin{equation*}
				\frac{2\pi e^{-i\pi\nu}}{\Gamma(1-\nu)\Gamma(\nu)} = i\gamma(1-\gamma)^{-1}.
		\end{equation*}
		\item As $\zeta\rightarrow\infty$, the model function
$P_{CH}^{RH}(\zeta)$ shows the following asymptotic behavior in the
whole neighborhood of infinity
\begin{eqnarray*}
	P_{CH}^{RH}(\zeta) &=& \Bigg[I+\frac{i}{\zeta}\begin{pmatrix}
	\nu^2& -\frac{\Gamma(1-\nu)}{\Gamma(\nu)}e^{\pi i\nu}\\
	\frac{\Gamma(1+\nu)}{\Gamma(-\nu)}e^{-\pi i\nu} & -\nu^2\\
	\end{pmatrix}\\
	&&+\frac{1}{\zeta^2}\begin{pmatrix}
	-\frac{\nu^2}{2}(1+\nu)^2 & -\frac{\Gamma(1-\nu)}{\Gamma(\nu)}(1-\nu)^2e^{\pi i\nu} \\
	-\frac{\Gamma(1+\nu)}{\Gamma(-\nu)}(1+\nu)^2e^{-\pi i\nu}& -\frac{\nu^2}{2}(1-\nu)^2\\
	\end{pmatrix}+O\big(\zeta^{-3}\big)\Bigg]\\
	&&\times\zeta^{-\nu\sigma_3}e^{-i\frac{\zeta}{2}\sigma_3}e^{-i\frac{\pi}{2}(\frac{1}{2}-\nu)\sigma_3},\hspace{0.5cm}\zeta\rightarrow\infty.
\end{eqnarray*}
\end{itemize}
The model function $P_{CH}^{RH}(\zeta)$ will now be used in the construction of the parametrix near $z=+1$. Define
\begin{equation}\label{PCHRHchangeright}
	\zeta(z)=-2is^3\big(\vartheta(z)-\vartheta(1)\big),\hspace{0.5cm}|z-1|<r,
\end{equation}
and notice
\begin{equation*}
	\zeta(z) = 2s^3\Big(4+\frac{x}{s^2}\Big)(z-1)\big(1+O(z-1)\big),\hspace{0.5cm}z\rightarrow 1
\end{equation*}
yielding local conformality. This change of variables allows us to define the right parametrix $V(z)$ near $z=+1$ as follows
\begin{equation}\label{PCHRHparametrixright}
	V(z) = \sigma_1e^{-i\frac{\pi}{4}\sigma_3}B_r(z)e^{i\frac{\pi}{2}(\frac{1}{2}-\nu)\sigma_3}e^{-s^3\vartheta(1)\sigma_3}P_{CH}^{RH}\big(\zeta(z)\big)e^{(\frac{i}{2}\zeta(z)+s^3\vartheta(1))\sigma_3}e^{i\frac{\pi}{4}\sigma_3}\sigma_1,
\end{equation}
with $\zeta(z)$ in \eqref{PCHRHchangeright} and the multiplier
\begin{equation*}
	B_r(z)=\bigg(\zeta(z)\frac{z+1}{z-1}\bigg)^{\nu\sigma_3},\hspace{0.5cm}B_r(1)=\big(16s^3+4xs\big)^{\nu\sigma_3}.
\end{equation*}
\begin{figure}[tbh]
  \begin{flushleft}
  \psfragscanon
  \psfrag{1}{\footnotesize{$(1-\gamma)^{\sigma_3}$}}
  \psfrag{2}{$e^{s^3\vartheta(z)\sigma_3}\bigl(\begin{smallmatrix}
  1 & -\gamma(1-\gamma)^{-1}\\
  0 & 1\\
  \end{smallmatrix}\bigr)e^{-s^3\vartheta(z)\sigma_3}$}
  \psfrag{3}{$e^{s^3\vartheta(z)\sigma_3}\bigl(\begin{smallmatrix}
  1 & 0\\
  \gamma(1-\gamma)^{-1} & 1\\
  \end{smallmatrix}\bigr)e^{-s^3\vartheta(z)\sigma_3}$}
  \hspace{4cm}\includegraphics[width=4cm,height=3cm]{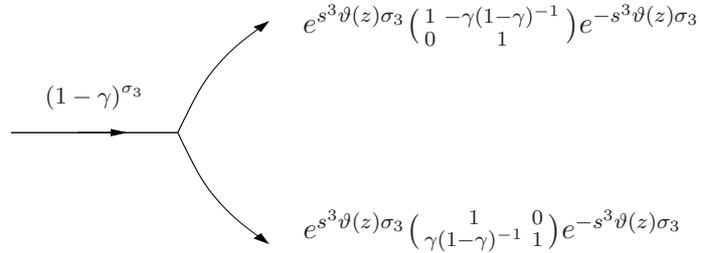}
  \end{flushleft}
  \caption{Transformation of parametrix jumps to original jumps}
  \label{fig10}
\end{figure}
Also here, following from analyticity of $B_r(z)$ and conjugation with $\sigma_1e^{-i\frac{\pi}{4}\sigma_3}e^{-s^3\vartheta(z)\sigma_3}$ in \eqref{PCHRHparametrixright}, parametrix jumps and jump curves (see Figure \ref{fig10}) match original jumps and locally original jump curves. Furthermore, and we shall elaborate this in full detail very soon, the singular endpoint behavior of parametrix $U(z)$ matches \eqref{Ssingularendp}, i.e.
\begin{equation}\label{Uendpoint}
	V(z)=O\big(\ln(z-1)\big),\hspace{0.5cm}z\rightarrow +1.
\end{equation}
Hence the ratio of $S(z)$ with $U(z)$ is locally analytic, i.e.
\begin{equation*}
	S(z) = N_r(z)V(z),\hspace{0.5cm}|z-1|<\frac{1}{2}
\end{equation*}
and again the role of the left multiplier $B_r(z)$ follows from the following asymptotical matchup
\begin{eqnarray}\label{conflurightmatchup}
	V(z)&=&\sigma_1e^{-i\frac{\pi}{4}\sigma_3}B_r(z)e^{i\frac{\pi}{2}(\frac{1}{2}-\nu)\sigma_3}e^{-s^3\vartheta(1)\sigma_3} \Bigg[I+\frac{i}{\zeta}\begin{pmatrix}
	\nu^2& -\frac{\Gamma(1-\nu)}{\Gamma(\nu)}e^{\pi i\nu}\\
	\frac{\Gamma(1+\nu)}{\Gamma(-\nu)}e^{-\pi i\nu} & -\nu^2\\
	\end{pmatrix}\nonumber\\
	&&+\frac{1}{\zeta^2}\begin{pmatrix}
	-\frac{\nu^2}{2}(1+\nu)^2 & -\frac{\Gamma(1-\nu)}{\Gamma(\nu)}(1-\nu)^2e^{\pi i\nu} \\
	-\frac{\Gamma(1+\nu)}{\Gamma(-\nu)}(1+\nu)^2e^{-\pi i\nu}& -\frac{\nu^2}{2}(1-\nu)^2\\
	\end{pmatrix}+O\big(\zeta^{-3}\big)\Bigg]\nonumber\\
	&&\times\zeta^{-\nu\sigma_3}e^{s^3\vartheta(1)\sigma_3}e^{-i\frac{\pi}{2}(\frac{1}{2}-\nu)\sigma_3}e^{i\frac{\pi}{4}\sigma_3}\sigma_1\nonumber\\
	&=&\Bigg[I+\frac{i}{\zeta}\begin{pmatrix}
	-\nu^2 & \frac{\Gamma(1+\nu)}{\Gamma(-\nu)}e^{2s^3\vartheta(1)}\beta_r^{-2}(z)\\
	-\frac{\Gamma(1-\nu)}{\Gamma(\nu)}e^{-2s^3\vartheta(1)}\beta_r^2(z) & \nu^2\\
	\end{pmatrix}\nonumber\\
	&&+\frac{1}{\zeta^2}\begin{pmatrix}
	-\frac{\nu^2}{2}(1-\nu)^2 & -\frac{\Gamma(1+\nu)}{\Gamma(-\nu)}(1+\nu)^2e^{2s^3\vartheta(1)}\beta_r^{-2}(z)\\
	-\frac{\Gamma(1-\nu)}{\Gamma(\nu)}(1-\nu)^2e^{-2s^3\vartheta(1)}\beta_r^2(z) & -\frac{\nu^2}{2}(1+\nu)^2\\
	\end{pmatrix}\nonumber\\
	&&+O\big(\zeta^{-3}\big)\Bigg]M(z)
\end{eqnarray}
as $s\rightarrow\infty$ valid on the annulus $0<r_1\leq |z-1|\leq r_2<1$ (hence $|\zeta|\rightarrow\infty$) using the abbreviation
\begin{equation*}
	\beta_r(z)=\bigg(\zeta(z)\frac{z+1}{z-1}\bigg)^{\nu}.
\end{equation*}
If we are dealing with the case $\gamma<1$, then
\begin{equation*}
	\beta_r^{\pm 2}(z)\frac{1}{\zeta} = O\left(s^{-3\pm 6\textnormal{Re}\,\nu}\right)=o(1),\hspace{0.5cm}s\rightarrow\infty.
\end{equation*}
This would mean that equation \eqref{conflurightmatchup} yields the matching relation between the model functions $V(z)$ and $M(z)$,
\begin{equation}\label{matchgoodr}
	V(z) = \big(I+o(1)\big)M(z),\hspace{0.5cm}s\rightarrow\infty,\ 0<r_1\leq|z-1|\leq r_2<\frac{1}{2}
\end{equation}
which is again crucial for the successful implementation of the nonlinear steepest descent method. However, if $\gamma>1$, then
\begin{equation*}
	\nu=\frac{1}{2\pi i}\ln(\gamma-1)+\frac{1}{2}\equiv \nu_0+\frac{1}{2}
\end{equation*}
and hence
\begin{equation*}
	\beta_r^2(z)\frac{1}{\zeta} = \hat{\beta}_r^2(z)\frac{z+1}{z-1}=O(1),\hspace{0.5cm}s\rightarrow\infty;\hspace{0.5cm}\hat{\beta}_r(z)=\left(\zeta(z)\frac{z+1}{z-1}\right)^{\nu_0}.
\end{equation*}
With this notation, we have to replace \eqref{matchgoodr} in case $\gamma>1$ by
\begin{equation}\label{matchbadr}
	V(z) = E_r(z)\big(I+o(1)\big)M(z),\hspace{0.5cm}s\rightarrow\infty,\ 0<r_1\leq|z-1|\leq r_2<\frac{1}{2}
\end{equation}
where
\begin{equation*}
	E_r(z) = \begin{pmatrix}
		1 & 0\\
		-i\frac{\Gamma(1-\nu)}{\Gamma(\nu)}e^{-2s^3\vartheta(1)}\hat{\beta}_r^2(z)\frac{z+1}{z-1} & 1\\
		\end{pmatrix}.
\end{equation*} 
The appearance of the nontrivial matrix term $E_r(z)$ instead of the unit matrix in estimate \eqref{matchbadr} yields a very serious change in the further asymptotic analysis comparing with the matching case \eqref{matchgoodr}. We will proceed with this analysis in sections \ref{sec13} and \ref{sec14}.


\section{Construction of a parametrix at the edge point $z=-1$}\label{sec10}
For now, we introduce the model RHP near the other endpoint $z=-1$. Opposed to \eqref{PCH} consider
\begin{eqnarray*}
    \widetilde{P}_{CH}(\zeta)&=& \begin{pmatrix}
                               U(-\nu,e^{-i\frac{3\pi}{2}}\zeta)e^{-i\frac{\zeta}{2}} & U(1+\nu,e^{-i\frac{\pi}{2}}\zeta)
                                e^{\pi i\nu}e^{i\frac{\zeta}{2}}\frac{\Gamma(1+\nu)}{\Gamma(-\nu)} \\
                               U(1-\nu,e^{-i\frac{3\pi}{2}}\zeta)e^{\pi i\nu}e^{-i\frac{\zeta}{2}}\frac{\Gamma(1-\nu)}{\Gamma(\nu)} &
U(\nu,e^{-i\frac{\pi}{2}}\zeta)e^{2\pi i\nu}e^{i\frac{\zeta}{2}} \\
                             \end{pmatrix}\\
                             &&\times e^{i\frac{\pi}{2}(\frac{1}{2}-\nu)\sigma_3}=\sigma_2P_{CH}(e^{-i\pi}\zeta)\sigma_2,\ \ \
0<\textnormal{arg}\ \zeta\leq 2\pi.
\end{eqnarray*}
and
\begin{equation}\label{PCHRHleft}
    \widetilde{P}_{CH}^{RH}(\zeta)=\left\{
                                 \begin{array}{ll}
                                   \widetilde{P}_{CH}(\zeta)\Bigl(\begin{smallmatrix}
                          1 & 0 \\
                          \frac{2\pi e^{-3i\pi\nu}}{\Gamma(1-\nu)\Gamma(\nu)} & 1 \\
                        \end{smallmatrix}\Bigr)\Bigl(\begin{smallmatrix}
                          e^{i\frac{\pi}{2}\nu} & 0 \\
                          0 & e^{-i\frac{3\pi}{2}\nu} \\
                        \end{smallmatrix}\Bigr), & \hbox{arg $\zeta\in(0,\frac{2\pi}{3})$,} \bigskip\\
                                   \widetilde{P}_{CH}(\zeta)\Bigl(\begin{smallmatrix}
                          1 & -\frac{2\pi e^{i\pi\nu}}{\Gamma(1-\nu)\Gamma(\nu)} \\
                          0 & 1 \\
                        \end{smallmatrix}\Bigr)\Bigl(\begin{smallmatrix}
                          e^{i\frac{\pi}{2}\nu} & 0 \\
                          0 & e^{-i\frac{3\pi}{2}\nu} \\
                        \end{smallmatrix}\Bigr), & \hbox{arg $\zeta\in(\frac{4\pi}{3},2\pi)$,} \bigskip \\
                                   \widetilde{P}_{CH}(\zeta)\Bigl(\begin{smallmatrix}
                          e^{i\frac{\pi}{2}\nu} & 0 \\
                          0 & e^{-i\frac{3\pi}{2}\nu} \\
                        \end{smallmatrix}\Bigr), & \hbox{arg $\zeta\in(\frac{2\pi}{3},\frac{4\pi}{3})$.}
                                 \end{array}
                               \right.
\end{equation}
The model function $\widetilde{P}_{CH}^{RH}(\zeta)$ solves the RHP of Figure \ref{fig11}
\begin{figure}[tbh]
  \begin{center}
  \psfragscanon
  \psfrag{1}{\footnotesize{$(1-\gamma)^{-\sigma_3}$}}
  \psfrag{2}{$\Bigl(\begin{smallmatrix}
  1 & 0\\
  -\frac{2\pi e^{-i\pi\nu}}{\Gamma(1-\nu)\Gamma(\nu)} & 1\\
  \end{smallmatrix}\Bigr)$}
  \psfrag{3}{$\Bigl(\begin{smallmatrix}
  1 & -\frac{2\pi e^{-i\pi\nu}}{\Gamma(1-\nu)\Gamma(\nu)}\\
  0 & 1\\
  \end{smallmatrix}\Bigr)$}
  \includegraphics[width=6cm,height=3cm]{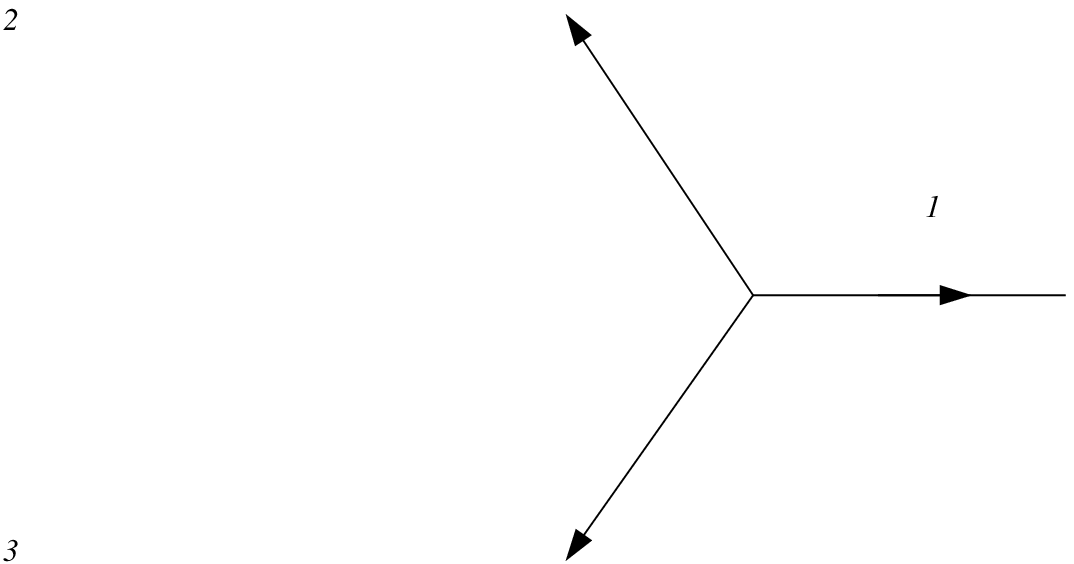}
  \end{center}
  \caption{The model RHP near $z=-1$ which can be solved explicitly using confluent hypergeometric functions}
  \label{fig11}
\end{figure}
\begin{itemize}
	\item $\widetilde{P}_{CH}^{RH}(\zeta)$ is analytic for $\zeta\in\mathbb{C}\backslash\{\textnormal{arg}\ \zeta=0,\frac{2\pi}{3},\frac{4\pi}{3}\}$
	\item Along the contour in Figure \ref{fig11}, the following jumps are valid (recall \eqref{PCHmono} and the symmetry relation $\widetilde{P}_{CH}^{RH}(\zeta)=\sigma_2P_{CH}(e^{-i\pi}\zeta)\sigma_2$)
	\begin{eqnarray*}
    \big(\widetilde{P}_{CH}^{RH}(\zeta)\big)_+&=&\big(\widetilde{P}_{CH}^{RH}(\zeta)\big)_-e^{-2\pi
i\nu\sigma_3},\ \ \ \textnormal{arg}\ \zeta=0\smallskip\\
    \big(\widetilde{P}_{CH}^{RH}(\zeta)\big)_+ &=&
\big(\widetilde{P}_{CH}^{RH}(\zeta)\big)_-\begin{pmatrix}
                                        1 & 0 \\
                                        -\frac{2\pi e^{-i\pi\nu}}{\Gamma(1-\nu)\Gamma(\nu)} & 1 \\
                                      \end{pmatrix},\ \ \
\textnormal{arg}\ \zeta=\frac{2\pi}{3}\smallskip\\
    \big(\widetilde{P}_{CH}^{RH}(\zeta)\big)_+&=&\big(\widetilde{P}_{CH}^{RH}(\zeta)\big)_-\begin{pmatrix}
                                                                                     1 & -\frac{2\pi e^{-i\pi\nu}}{\Gamma(1-\nu)\Gamma(\nu)} \\
                                                                                     0 & 1 \\
                                                                                   \end{pmatrix},\
\ \ \textnormal{arg}\ \zeta=\frac{4\pi}{3}
\end{eqnarray*}
	\item From symmetry $\widetilde{P}_{CH}^{RH}(\zeta)=\sigma_2P_{CH}(e^{-i\pi}\zeta)\sigma_2$ and the asymptotic information derived earlier for $P_{CH}(\zeta)$ in the different sectors, we deduce the following behavior, valid in a full neighborhood of infinity
\begin{eqnarray*}
	\widetilde{P}_{CH}^{RH}(\zeta)&=&\Bigg[I+\frac{i}{\zeta}\begin{pmatrix}
	\nu^2 & \frac{\Gamma(1+\nu)}{\Gamma(-\nu)}e^{-\pi i\nu}\\
	-\frac{\Gamma(1-\nu)}{\Gamma(\nu)}e^{\pi i\nu} & -\nu^2\\
	\end{pmatrix}\\
	&&+\frac{1}{\zeta^2}\begin{pmatrix}
	-\frac{\nu^2}{2}(1-\nu)^2 & \frac{\Gamma(1+\nu)}{\Gamma(-\nu)}(1+\nu)^2e^{-\pi i\nu}\\
	\frac{\Gamma(1-\nu)}{\Gamma(\nu)}(1-\nu)^2e^{\pi i\nu} & -\frac{\nu^2}{2}(1+\nu)^2\\
	\end{pmatrix}+O\big(\zeta^{-3}\big)\Bigg]\\
	&&\times\big(e^{-i\pi}\zeta\big)^{\nu\sigma_3}e^{-i\frac{\zeta}{2}\sigma_3}e^{i\frac{\pi}{2}(\frac{1}{2}-\nu)\sigma_3},\hspace{0.5cm}\zeta\rightarrow\infty
\end{eqnarray*}
\end{itemize}
Now, similarly to what we did in the construction of $V(z)$, define
\begin{equation}\label{PCHRHchangeleft}
	\zeta(z)=-2is^3\big(\vartheta(z)-\vartheta(-1)\big),\hspace{0.5cm}|z+1|<r
\end{equation}
with
\begin{equation*}
	\zeta(z)= 2s^3\Big(4+\frac{x}{s^2}\Big)(z+1)\big(1+O(z+1)\big),\hspace{0.5cm}z\rightarrow -1,
\end{equation*}
hence a locally conformal change of variables. With this change the parametrix $W(z)$ near the left endpoint $z=-1$ will be defined as
\begin{equation}\label{PCHRHparametrixleft}						W(z)=\sigma_1e^{-i\frac{\pi}{4}\sigma_3}B_l(z)e^{-i\frac{\pi}{2}(\frac{1}{2}-\nu)\sigma_3}e^{-s^3\vartheta(-1)\sigma_3}\widetilde{P}_{CH}^{RH}\big(\zeta(z)\big)e^{(\frac{i}{2}\zeta(z)+s^3\vartheta(-1))\sigma_3}e^{i\frac{\pi}{4}\sigma_3}\sigma_1,
\end{equation}
with $\zeta(z)$ as in \eqref{PCHRHchangeleft} and 
\begin{equation*}
	B_l(z)=\bigg(e^{-i\pi}\zeta(z)\frac{z-1}{z+1}\bigg)^{-\nu\sigma_3},\hspace{0.5cm} B_l(-1) = \big(16s^3+4xs\big)^{-\nu\sigma_3}.
\end{equation*}
Since 
\begin{equation*}
	\frac{2\pi e^{-i\pi\nu}}{\Gamma(1-\nu)\Gamma(\nu)}=i\gamma(1-\gamma)^{-1}
\end{equation*}
the stated conjugation with $\sigma_1e^{-i\frac{\pi}{4}\sigma_3}e^{-s^3\vartheta(z)\sigma_3}$ will again match parametrix jumps with original jumps locally on the original jump contour (see Figure \ref{fig12}).
\begin{figure}[tbh]
  \begin{center}
  \psfragscanon
  \psfrag{1}{\footnotesize{$(1-\gamma)^{\sigma_3}$}}
  \psfrag{2}{$e^{s^3\vartheta(z)\sigma_3}\bigl(\begin{smallmatrix}
  1 & \gamma(1-\gamma)^{-1}\\
  0 & 1\\
  \end{smallmatrix}\bigr)e^{-s^3\vartheta(z)\sigma_3}$}
  \psfrag{3}{$e^{s^3\vartheta(z)\sigma_3}\bigl(\begin{smallmatrix}
  1 & 0\\
  -\gamma(1-\gamma)^{-1} & 1\\
  \end{smallmatrix}\bigr)e^{-s^3\vartheta(z)\sigma_3}$}
  \includegraphics[width=9cm,height=3cm]{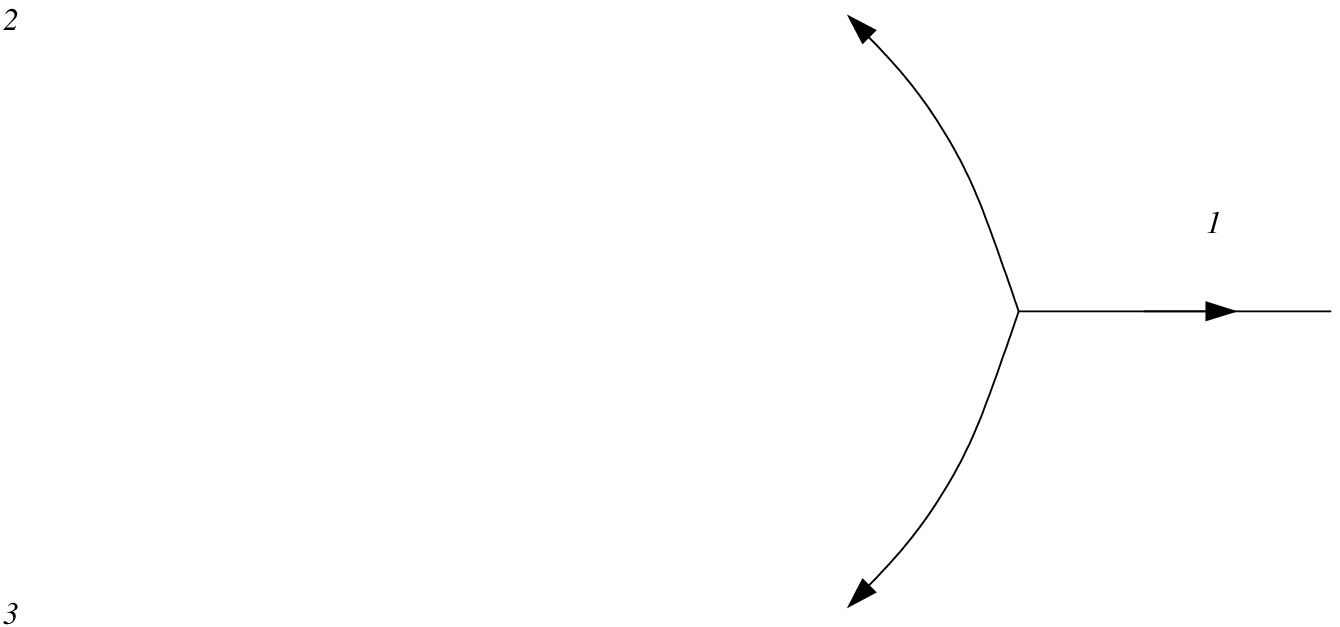}
  \end{center}
  \caption{Transformation of parametrix jumps to original jumps}
  \label{fig12}
\end{figure}
Moreover the singular endpoint behavior
\begin{equation}\label{Wendpoint}
	W(z) = O\big(\ln(z+1)\big),\hspace{0.5cm}z\rightarrow -1
\end{equation}
matches \eqref{Xendpoint}, leading to a locally analytic ratio of $S(z)$ with $W(z)$
\begin{equation*}
	S(z)=N_l(z)W(z),\hspace{0.5cm}|z+1|<\frac{1}{2}.
\end{equation*}
Also here the left multiplier $B_l(z)$ provides us with an asymptotical matchup
\begin{eqnarray}\label{confluleftmatchup}
	W(z) &=& \Bigg[I+\frac{i}{\zeta}\begin{pmatrix}
	-\nu^2 & \frac{\Gamma(1-\nu)}{\Gamma(\nu)}e^{2s^3\vartheta(-1)}\beta_l^2(z)\\
	-\frac{\Gamma(1+\nu)}{\Gamma(-\nu)}e^{-2s^3\vartheta(-1)}\beta_l^{-2}(z)& \nu^2\\
	\end{pmatrix}\nonumber\\
	&&+\frac{1}{\zeta^2}\begin{pmatrix}
	-\frac{\nu^2}{2}(1+\nu)^2 & -\frac{\Gamma(1-\nu)}{\Gamma(\nu)}(1-\nu)^2e^{2s^3\vartheta(-1)}\beta_l^2(z)\\
	-\frac{\Gamma(1+\nu)}{\Gamma(-\nu)}(1+\nu)^2e^{-2s^3\vartheta(-1)}\beta_l^{-2}(z)& -\frac{\nu^2}{2}(1-\nu)^2\\
	\end{pmatrix}\nonumber\\
	&&+O\big(\zeta^{-3}\big)\Bigg]M(z)
\end{eqnarray}
as $s\rightarrow\infty$ valid on the annulus $0<r_1\leq|z+1|\leq r_2<1$ (thus $|\zeta|\rightarrow\infty$) and we introduced the abbreviation
\begin{equation*}
	\beta_l(z) = \bigg(e^{-i\pi}\zeta(z)\frac{z-1}{z+1}\bigg)^{\nu}.
\end{equation*}
Similar to the previous situation this implies on the annulus for $\gamma<1$ 
\begin{equation*}
	W(z)=\big(I+o(1)\big)M(z),\hspace{0.5cm}s\rightarrow\infty
\end{equation*}
whereas in case $\gamma>1$
\begin{equation}\label{matchbadl}
	W(z)=E_l(z)\big(I+o(1)\big)M(z),\hspace{0.5cm}s\rightarrow\infty
\end{equation}
with
\begin{equation*}
	E_l(z)=\begin{pmatrix}
	1 & -i\frac{\Gamma(1-\nu)}{\Gamma(\nu)}e^{2s^3\vartheta(-1)}\hat{\beta}_l^2(z)\frac{z-1}{z+1}\\
	0 & 1\\
	\end{pmatrix},\hspace{0.5cm}\hat{\beta}_l(z) = \bigg(e^{-i\pi}\zeta(z)\frac{z-1}{z+1}\bigg)^{\nu_0}.
\end{equation*}
At this point we can use the model functions $M(z),U(z),V(z)$ and $W(z)$ to employ a further transformation.


\section{Third transformation of the RHP}\label{sec11}

We put in this transformation
\begin{equation}\label{ratioRHP}
	R(z) = S(z)\left\{
                                   \begin{array}{ll}
                                     \big(V(z)\big)^{-1}, & \hbox{$|z-1|<r_1$,} \\
                                     \big(U(z)\big)^{-1}, & \hbox{$|z|<r_2$,} \\
                                     \big(W(z)\big)^{-1}, & \hbox{$|z+1|<r_1$,} \\
                                     \big(M(z)\big)^{-1}, & \hbox{$|z-1|>r_1,|z+1|>r_1,|z|>r_2$,} \\
                                   \end{array}
                                 \right.
\end{equation}
where $0<r_1,r_2<\frac{1}{2}$ is fixed. With $C_{0,r,l}$ denoting the clockwise oriented circles shown in Figure \ref{fig15}, the ratio-function $R(z)$ solves the following RHP
\begin{figure}[tbh]
  \begin{center}
  \psfragscanon
  \psfrag{1}{\footnotesize{$\gamma_1$}}
  \psfrag{2}{\footnotesize{$\gamma_2$}}
  \psfrag{3}{\footnotesize{$\gamma_3$}}
  \psfrag{4}{\footnotesize{$\gamma_4$}}
  \psfrag{5}{\footnotesize{$\gamma_5$}}
  \psfrag{6}{\footnotesize{$\gamma_6$}}
  \psfrag{7}{\footnotesize{$\gamma_7$}}
  \psfrag{8}{\footnotesize{$\gamma_8$}}
  \psfrag{9}{\footnotesize{$C_0$}}
  \psfrag{10}{\footnotesize{$C_l$}}
  \psfrag{11}{\footnotesize{$C_r$}}
  \includegraphics[width=8cm,height=5cm]{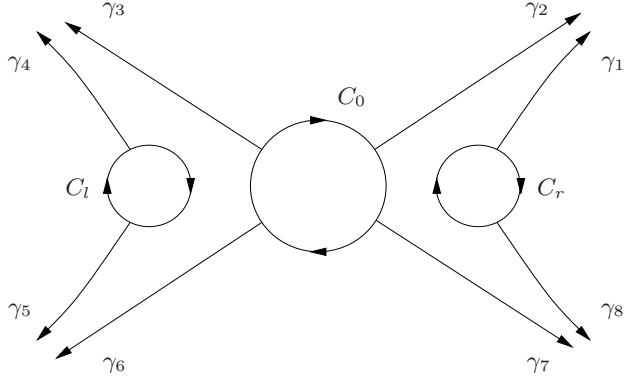}
  \end{center}
  \caption{The jump graph for the ratio-function $R(z)$}
  \label{fig15}
\end{figure}
\begin{itemize}
	\item $R(z)$ is analytic for $z\in\mathbb{C}\backslash\big\{ C_{0,r,l}\cup \bigcup_{i=1}^8 \gamma_i\big\}$
	\item For the jumps, along the infinite branches $\gamma_i$
	\begin{equation*}
		R_+(z) = R_-(z)M(z)e^{s^3\vartheta(z)\sigma_3}\widehat{G}_Se^{-s^3\vartheta(z)\sigma_3}\big(M(z)\big)^{-1},
	\end{equation*}
	with $\widehat{G}_S$ denoting the corresponding jump matrices from \eqref{openingjumps}. On the clockwise orientied circles $C_0$ and $C_{r,l}$, the jumps are described by the equation
	\begin{eqnarray*}
		R_+(z) &=& R_-(z)U(z)\big(M(z)\big)^{-1},\ z\in C_0,\\
		R_+(z) &=& R_-(z)V(z)\big(M(z)\big)^{-1},\ z\in C_r, \\
		R_+(z) &=& R_-(z)W(z)\big(M(z)\big)^{-1},\ z\in C_l.
	\end{eqnarray*}
	\item $R(z)$ is analytic at $z=\pm 1$. This observation will follow directly from \eqref{Uendpoint} and \eqref{Wendpoint}, which will be proven in section \ref{sec16}
	\item In a neighborhood of infinity, we have $R(z)\rightarrow I$.
\end{itemize}	
We emphasize that, by construction, $R(z)$ has no jumps inside of the circles $C_{0,r,l}$ and across the line segment in between. In order to apply the Defit-Zhou nonlinear steepest descent method to the ratio-RHP, all its jump matrices have to be close to the unit matrix, as $s\rightarrow\infty$, compare \cite{DZ1}. Hence it is now important to recall the previosuly stated behavior of the jump matrices as $s\rightarrow\infty$: As mentioned before, due to the triangularity of $S_U,S_L$ and the sign-diagram of $\textnormal{Re}\,\vartheta(z)$, the jump matrices corresponding to the infinite parts $\bigcup_{i=1}^8\gamma_i$ of the $R$-jump contour are in fact exponentially close to the unit matrix
\begin{equation}\label{esti1}
	\|Me^{s^3\vartheta(\cdot)\sigma_3}\widehat{G}_Se^{-s^3\vartheta(\cdot)\sigma_3}\big(M\big)^{-1}-I\|_{L^2\cap L^{\infty}(\gamma_i)}\leq c_1\left\{\begin{array}{ll}
                                     e^{-c_2s^3|z|}, & \hbox{$\textnormal{from}\ C_0$;} \smallskip\\
                                     e^{-c_3s^3|z\mp 1|}, & \hbox{$\textnormal{from}\ C_{r,l}$.}
                                   \end{array}
                                \right.
\end{equation}
as $s\rightarrow\infty$ with constants $c_i>0$ whose values are not important. Also by virtue of \eqref{Absegurmatch}, $U(z)\big(M(z)\big)^{-1}$ approaches the unit matrix as $s\rightarrow\infty$,
\begin{equation}\label{esti2}
	\| U\big(M\big)^{-1}-I\|_{L^2\cap L^{\infty}(C_0)}\leq c_4s^{-1}
\end{equation}
with a constant $c_4>0$. The jumps on $C_{r,l}$ however have to be treated more carefully. As we already mentioned in sections \ref{sec9} and \ref{sec10}, estimates \eqref{conflurightmatchup} and \eqref{confluleftmatchup} yield in case $\gamma<1$
\begin{equation}\label{esti3}
	\| V\big(M\big)^{-1}-I\|_{L^2\cap L^{\infty}(C_r)}\leq c_5s^{-3},\ \ \|W\big(M\big)^{-1}-I\|_{L^2\cap L^{\infty}(C_l)}\leq c_6 s^{-3}
\end{equation}
as $s\rightarrow\infty$. The estimations \eqref{esti1}, \eqref{esti2} and \eqref{esti3}, which are uniform on any compact subset of the set \eqref{excset1}
\begin{equation*}
		\{(\gamma,x)\in\mathbb{R}^2:\ -\infty<\gamma<1,\ -\infty<x<\infty\},
\end{equation*}
enable us to solve the ratio-RHP iteratively in that particular situation.


\section{Solution of the RHP for $R(z)$ via iteration,\ $\gamma<1$}\label{sec12}

Let us denote with $G_R$ the jump matrix in the ratio-RHP and with $\Sigma_R$ the underlying contour. The stated RHP for the function $R(z)$
\begin{itemize}
	\item $R(z)$ is analytic for $z\in\mathbb{C}\backslash\Sigma_R$
	\item Along the contour depicted in Figure \ref{fig15}
	\begin{equation*}
		R_+(z) = R_-(z)G_R(z),\hspace{0.5cm}z\in\Sigma_R.
	\end{equation*} 
	\item As $z\rightarrow\infty$, we have $R(z) = I+O\big(z^{-1}\big)$.
\end{itemize} 
is equivalent to the singular integral equation
\begin{equation}\label{integraleq1}
	R_-(z) = I+\frac{1}{2\pi i}\int\limits_{\Sigma_R}R_-(w)\big(G_R(w)-I\big)\frac{dw}{w-z_-}
\end{equation} 
and by the previous estimates \eqref{esti1},\eqref{esti2} and \eqref{esti3}, we have
\begin{equation}\label{esti4}
	\|G_{R}-I\|_{L^2\cap L^{\infty}(\Sigma_R)}\leq c_6s^{-1}
\end{equation}
uniformly on any compact subset of the set \eqref{excset1}. By standard arguments (see \cite{DZ1}), we know that for sufficiently large $s$, the relevant integral operator is contracting and equation \eqref{integraleq1} can be solved iteratively in $L^2(\Sigma_R)$. Moreover, its unique solution satisfies
\begin{equation}\label{esti5}
	\|R_--I\|_{L^2(\Sigma_R)}\leq cs^{-1},\hspace{0.5cm} s\rightarrow\infty.
\end{equation}
The latter information is all we need to compute the asymptotic expansion for the Fredholm determinant $\det(I-\gamma K_{\textnormal{csin}})$ in case $\gamma<1$. Before we derive the relevant asymptotics let us first discuss the situation $\gamma>1$. In this case
\begin{equation}\label{esti6}
	\| V\big(M\big)^{-1}-I\|_{L^2\cap L^{\infty}(C_r)}\nrightarrow 0,\ \ \| W\big(M\big)^{-1}\|_{L^2\cap L^{\infty}(C_l)}\nrightarrow 0
\end{equation}
and we need to employ further transformations.


\section{Fourth transformation of the RHP - undressing}\label{sec13}

The presence of the multipliers $E_r(z)$ and $E_l(z)$ in \eqref{matchbadr} and \eqref{matchbadl} requires further transformations leading to a singular or solitonic type of Riemann-Hilbert problem. Following \cite{FT,BoI2}, we will show how to deal with the singular structure.
\smallskip

A key observation for our next move is that the jump matrices $G_r(z)=V(z)\big(M(z)\big)^{-1}$ and $G_l(z)=W(z)\big(M(z)\big)^{-1}$ admit the following algebraic factorizations
\begin{eqnarray}
	G_r(z) &=& E_r(z)\widehat{G}_r(z)=\begin{pmatrix}
	1 & 0\\
	-i\frac{\Gamma(1-\nu)}{\Gamma(\nu)}e^{-2s^3\vartheta(1)}\hat{\beta}_r^2(z)\frac{z+1}{z-1}& 1\\
	\end{pmatrix}\nonumber\\
	&&\times\Bigg[I+\frac{i}{\zeta}\begin{pmatrix}
	-\nu^2 & 0 \\
	-i\frac{\Gamma(1-\nu)}{\Gamma(\nu)}e^{-2s^3\vartheta(1)}\frac{z+1}{z-1}(2\nu-1) & \nu^2\\
	\end{pmatrix} +O\big(\zeta^{-2}\big)\Bigg],\label{facto1}\\
	G_l(z) &=& E_l(z)\widehat{G}_l(z)=\begin{pmatrix}
	1 & -i\frac{\Gamma(1-\nu)}{\Gamma(\nu)}e^{2s^3\vartheta(-1)}\hat{\beta}_l^2(z)\frac{z-1}{z+1}\\
	0 & 1\\
	\end{pmatrix}\nonumber\\
	&&\times\Bigg[I+\frac{i}{\zeta}\begin{pmatrix}
	-\nu^2 & -i\frac{\Gamma(1-\nu)}{\Gamma(\nu)}e^{2s^3\vartheta(-1)}\hat{\beta}_l^2(z)(1-2\nu)\\
	0 & \nu^2\\
	\end{pmatrix}+O\big(\zeta^{-2}\big)\Bigg]\label{facto2}
\end{eqnarray}
as $s\rightarrow\infty$ and $0<r_1\leq |z\mp 1|\leq r_2<1$. We observe that $\|\widehat{G}_{r,l}-I\|\rightarrow 0$ as $s\rightarrow\infty$; in fact, since $|\zeta(z)|\geq cs^3$ on $C_r\cup C_l$, we have that
\begin{equation*}
	\|\widehat{G}_{r,l}-I\|_{L^2\cap L^{\infty}(C_r,l)}\leq c_7s^{-3},\hspace{0.5cm}s\rightarrow\infty.
\end{equation*}
Hence the natural idea is to pass from the function $R(z)$ to the function $P(z)$ defined by the equations
\begin{equation}\label{undressing}
	P(z)=\left\{
                 \begin{array}{ll}
                   R(z)E_r(z), & \hbox{$|z-1|<r_1$,} \\
                   R(z)E_l(z), & \hbox{$|z+1|<r_1$,} \\
                   R(z), & \hbox{$|z\mp 1|>r_1$.} 
                 \end{array}
               \right.
\end{equation}
with $0<r_1<\frac{1}{2}$ chosen as in \eqref{ratioRHP}. By definition, the function $P(z)$ solves the following RHP:
\begin{itemize}
	\item $P(z)$ is analytic for $z\in\mathbb{C}\backslash \big( C_{0,r,l} \cup\{\pm 1\}\cup\bigcup_{i=1}^8\gamma_i\big)$
	\item $P_+(z)=P_-(z)G_P(z)$, where
	\begin{equation*}
		G_P(z)=\left\{
                 \begin{array}{ll}
                   \widehat{G}_{r,l}(z), & \hbox{$z\in C_{r,l}$,} \\
                   U(z)\big(M(z)\big)^{-1}, & \hbox{$z\in C_0$,} \\
                   M(z)e^{s^3\vartheta(z)\sigma_3}\widehat{G}_Se^{-s^3\vartheta(z)\sigma_3}\big(M(z)\big)^{-1}, & \hbox{$z\in\gamma_i,\,i=1,\ldots,8$.} 
                 \end{array}
               \right.
  \end{equation*}
	\item $P(z)$ has first-order poles at $z=\pm 1$. More precisely let $P(z)=\big(P^{(1)}(z),P^{(2)}(z)\big)$ with $P^{(i)}(z)$ denoting the columns of the corresponding $2\times 2$ matrix valued function. We obtain from \eqref{facto1},\eqref{facto2} and \eqref{undressing}
	\begin{eqnarray}
		\textnormal{res}_{z=+1}P^{(1)}(z) &=&P^{(2)}(1)\bigg(-2i\frac{\Gamma(1-\nu)}{\Gamma(\nu)}e^{-2s^3\vartheta(1)}\hat{\beta}_r^2(1)\bigg)\label{res1}\\
		\textnormal{res}_{z=-1}P^{(2)}(z) &=&P^{(1)}(-1)\bigg(2i\frac{\Gamma(1-\nu)}{\Gamma(\nu)}e^{2s^3\vartheta(-1)}\hat{\beta}_l^2(-1)\bigg)\label{res2}.
	\end{eqnarray} 
	\item As $z\rightarrow\infty$, we have $P(z)\rightarrow I$.
\end{itemize}
First the latter four properties determine $P(z)$ uniquely.

\begin{prop}\label{prop4} The stated singular Riemann-Hilbert problem for $P(z)$ has a unique solution.
\end{prop}
\begin{proof} The residue relations \eqref{res1},\eqref{res2} imply
\begin{equation}\label{singbehavior}
	P(z)=\left\{
                 \begin{array}{ll}
                   \hat{P}^{(+)}(z)\begin{pmatrix}
                   1 & 0\\
                   -\frac{2p}{z-1} & 1\\
                   \end{pmatrix}, & \hbox{$|z-1|<r$;} \\
                   \hat{P}^{(-)}(z)\begin{pmatrix}
                   1 & \frac{2p}{z+1}\\
                   0 & 1\\
                   \end{pmatrix}, & \hbox{$|z+1|<r$,} 
                 \end{array}
               \right.\hspace{0.5cm} p = i\frac{\Gamma(1-\nu)}{\Gamma(\nu)}e^{-2s^3\vartheta(1)}\hat{\beta}_r^2(1)
\end{equation}
where $\hat{P}^{(\pm)}(z)$ are analytic at $z=\pm 1$. Hence one establishes $\det P(z)\equiv 1$ via Liouville theorem using the normalization at infinity und unimodularity of the jump matrices. From this and representation \eqref{singbehavior}, the ratio of any two solutions $P_1(z)$ and $P_2(z)$ of the given $P$-RHP, i.e.
\begin{equation*}
	P_1(z)\big(P_2(z)\big)^{-1},
\end{equation*}
is an entire function approaching identity at infinity; hence $P_1=P_2$, showing uniqueness.
\end{proof}
\smallskip

Secondly, all jump matrices in the $P$-RHP approach the identity matrix as $s\rightarrow\infty$; however $P(z)$ has singularities at $z=\pm 1$ whose structure is described by the residue relations \eqref{res1} and \eqref{res2}. This type of Riemann-Hilbert problem is a known one in the theory of integrable systems. The way to deal with such RHPs is to use a certain ``dressing procedure'' which reduces the problem to the one without the pole singularities. 


\section{Fifth and final transformation of the RHP - dressing}\label{sec14}
We put
\begin{equation}\label{dressing}
	P(z) = (zI+B)Q(z)\begin{pmatrix}
	\frac{1}{z-1} & 0\\
	0 &\frac{1}{z+1}\\
	\end{pmatrix},
\end{equation}
where $B\in\mathbb{C}^{2\times 2}$ is constant and see immediately that $Q(z)$ solves the following RHP:
\begin{itemize}
	\item $Q(z)$ is analytic for $z\in\mathbb{C}\backslash \big(C_{0,r,l}\cup\bigcup_{i=1}^8\gamma_i\big)$
	\item $Q_+(z) = Q_-(z)G_Q(z)$, where
	\begin{equation*}
		G_Q(z)=\begin{pmatrix}
		\frac{1}{z-1} & 0\\
		0 &\frac{1}{z+1}\\
		\end{pmatrix}\widehat{G}_{r,l}(z)\begin{pmatrix}
		z-1 & 0\\
		0 & z+1\\
		\end{pmatrix},\hspace{0.5cm} z\in C_{r,l}
	\end{equation*}
	and
	\begin{equation*}
		G_Q(z) = \begin{pmatrix}
		\frac{1}{z-1} & 0\\
		0 &\frac{1}{z+1}\\
		\end{pmatrix}U(z)\big(M(z)\big)^{-1}\begin{pmatrix}
		z-1 & 0\\
		0 & z+1\\
		\end{pmatrix},\hspace{0.5cm} z\in C_0
	\end{equation*}
	as well as
	\begin{equation*}
		G_Q(z) = \begin{pmatrix}
		\frac{1}{z-1} & 0\\
		0 &\frac{1}{z+1}\\
		\end{pmatrix}M(z)e^{s^3\vartheta(z)\sigma_3}\widehat{G}_Se^{-s^3\vartheta(z)\sigma_3}\big(M(z)\big)^{-1}\begin{pmatrix}
		z-1 & 0\\
		0 & z+1\\
		\end{pmatrix},\hspace{0.5cm} z\in\gamma_i.
	\end{equation*}
	\item $Q(z)\rightarrow I$, as $z\rightarrow\infty$
\end{itemize}
The $Q$-jump matrix $G_Q(z)$ is uniformly close to the unit matrix: therefore the $Q$-RHP admits direct asymptotic analysis, which will be performed after we determined the unknown matrix $B$. Using the conditions \eqref{res1} and \eqref{res2}
\begin{eqnarray*}
	\textnormal{res}_{z=+1}P^{(1)}(z)&=&(I+B)Q^{(1)}(1) = (I+B)Q^{(2)}(1)\big(-p\big),\\
	\textnormal{res}_{z=-1}P^{(2)}(z)&=&(-I+B)Q^{(2)}(-1)=(-I+B)Q^{(1)}(-1)\big(-p\big),
\end{eqnarray*}
so
\begin{equation}\label{Bmatrix}
	B=\Bigg(Q(1)\begin{pmatrix}
	1\\
	p\\
	\end{pmatrix},Q(-1)\begin{pmatrix}
	p\\
	1\\
	\end{pmatrix}\Bigg)\begin{pmatrix}
	-1 & 0\\
	0 & 1\\
	\end{pmatrix}\bigg(Q(1)\begin{pmatrix}
	1\\
	p\\
	\end{pmatrix},Q(-1)\begin{pmatrix}
	p\\
	1\\
	\end{pmatrix}\Bigg)^{-1}
\end{equation}
where $p$ was introduced in \eqref{singbehavior}. Let us see for which values of $s$ the latter matrix inverse is well-defined.


\section{Solution of the RHP for $Q(z)$ via iteration,\ $\gamma>1$}\label{sec15}

Since 
\begin{equation}\label{esti7}
	\|G_Q-I\|_{L^2\cap L^{\infty}(\Sigma_R)}\leq c_8s^{-1},\ \ s\rightarrow\infty
\end{equation}
we can solve the $Q$-RHP via iteration. This problem is equivalent to the singular integral equation
\begin{equation}\label{integraleq2}
	Q_-(z) = I+\frac{1}{2\pi i}\int\limits_{\Sigma_R}Q_-(w)\big(G_Q(w)-I\big)\frac{dw}{w-z_-}
\end{equation}
which can be solved via iteration in $L^2(\Sigma_R)$, its unique solution satisfies
\begin{equation}\label{esti8}
	\|Q_--I\|_{L^2(\Sigma_R)}\leq \tilde{c}s^{-1},\ \ s\rightarrow\infty.
\end{equation}
Combining the integral representation
\begin{equation*}
	Q(z) = I+\frac{1}{2\pi i}\int\limits_{\Sigma_R}Q_-(w)\big(G_Q(w)-I\big)\frac{dw}{w-z},\hspace{0.5cm} z\notin \Sigma_R
\end{equation*}
with \eqref{esti7} and \eqref{esti8}, we conclude
\begin{equation*}
	Q(\pm 1) = I+O\big(s^{-1}\big),\hspace{0.5cm}s\rightarrow\infty
\end{equation*}
Hence the matrix inverse in the right hand side of \eqref{Bmatrix} exists for all sufficiently large $s$ lying outside of the zero set of the function
\begin{equation*}
	1-p^2
\end{equation*}
which consists of the points $\{s_n\}$ defined by the equation
\begin{equation*}
	\frac{8}{3}s_n^3+2xs_n+\frac{1}{\pi}\ln(\gamma-1)\ln\big(16s^3+4xs\big)-\textnormal{arg}\ \frac{\Gamma(1-\nu)}{\Gamma(\nu)} = \frac{\pi}{2}+ n\pi,\hspace{0.5cm} n=1,2,\ldots
\end{equation*}
and which will eventually form the zeros of the Fredholm determinant as written in Theorem \ref{theo2result}. From now on, when dealing with the situation $\gamma>1$, we shall assume that $s$ stays away from the small neighborhood of the points $s_n$. We have now gathered enough information to prove Theorem \ref{theo1} and \ref{theo2}.


\section{Asymptotics of $\ln\det(I-\gamma K_{\textnormal{csin}})$ - preliminary steps}\label{sec16}

In order to prove the stated theorems we will use Proposition \ref{prop1},\ref{prop2} and \ref{prop3}, which in particular requires us to connect $\check{X}(\pm s)$ and $\check{X}'(\pm s)$ to the solution of either the $R$-RHP or the $Q$-RHP, see \eqref{sidentity} and \eqref{gammaderiv5}. To this end recall \eqref{Ssingularendp}, \eqref{ratioRHP} and Figure \ref{fig2}, which implies for $|z-1|<r_1$
\begin{equation}\label{comparison1}
	R(z)V(z)L^{-1}(z)e^{s^3\vartheta(z)\sigma_3} = \check{X}(zs)\Bigg[I+\frac{\gamma}{2\pi i}\begin{pmatrix}
	-1 & 1\\
	-1 & 1\\
	\end{pmatrix}\ln\bigg(\frac{z-1}{z+1}\bigg)\Bigg].
\end{equation}
On the other hand for $|z+1|<r_1$
\begin{equation}\label{comparison2}
	R(z)W(z)L^{-1}(z)e^{s^3\vartheta(z)\sigma_3} = \check{X}(zs)\Bigg[I+\frac{\gamma}{2\pi i}\begin{pmatrix}
	-1 & 1\\
	-1 & 1\\
	\end{pmatrix}\ln\bigg(\frac{z-1}{z+1}\bigg)\Bigg].
\end{equation}
This shows that the required values of $\check{X}(\pm s)$ and $\check{X}'(\pm s)$ can be determined via comparison in \eqref{comparison1} and \eqref{comparison2} once we know the local expansions of $V(z)$, respectively $W(z)$ at $z=\pm 1$. Our starting point is \eqref{confluorigin}
\begin{eqnarray}\label{PCHlocal}
	P_{CH}(\zeta)&=&\Bigg[\begin{pmatrix}
	d_1(\zeta,\nu)e^{2\pi i\nu} & -d_2(\zeta,1-\nu)e^{\pi i\nu}\frac{\Gamma(1-\nu)}{\Gamma(\nu)}\\
	-d_1(\zeta,1+\nu)e^{\pi i\nu}\frac{\Gamma(1+\nu)}{\Gamma(-\nu)}&d_2(\zeta,-\nu)\\
	\end{pmatrix}\nonumber \\
	&&+\zeta\begin{pmatrix}
	d_3(\zeta,\nu)e^{2\pi i\nu} & -d_4(\zeta,1-\nu)e^{\pi i\nu}\frac{\Gamma(1-\nu)}{\Gamma(\nu)}\\
	-d_3(\zeta,1+\nu)e^{\pi i\nu}\frac{\Gamma(1+\nu)}{\Gamma(-\nu)}& d_4(\zeta,-\nu)\\
	\end{pmatrix}+O\big(\zeta^2\ln\zeta\big)\Bigg]\nonumber\\
	&&\times e^{-i\frac{\pi}{2}(\frac{1}{2}-\nu)\sigma_3},\hspace{0.5cm}\zeta\rightarrow 0
\end{eqnarray}
with (recall \eqref{confluorigin})
\begin{equation*}
	d_1(\zeta,\nu) = c_0(\nu)+c_1(\nu)\Big(\ln\zeta+i\frac{\pi}{2}\Big),\hspace{0.5cm} d_2(\zeta,\nu)=c_0(\nu)+c_1(\nu)\Big(\ln\zeta-i\frac{\pi}{2}\Big)
\end{equation*}
and
\begin{equation*}
	d_3(\zeta,\nu) =-\frac{i}{2}d_1(\zeta,\nu)+i\bigg(c_2(\nu)+c_3(\nu)\Big(\ln\zeta+i\frac{\pi}{2}\Big)\bigg)
\end{equation*}
as well as
\begin{equation*}
	d_4(\zeta,\nu) = \frac{i}{2}d_2(\zeta,\nu)-i\bigg(c_2(\nu)+c_3(\nu)\Big(\ln\zeta-i\frac{\pi}{2}\Big)\bigg).
\end{equation*}
At this point use the changes of variables
\begin{equation*}
	\zeta = \zeta(z) = -2is^3\big(\vartheta(z)-\vartheta(1)\big),\ |z-1|<r_1,\hspace{0.5cm} \lambda=zs
\end{equation*}
and deduce from \eqref{PCHlocal} for $-\frac{\pi}{3}<\textnormal{arg}(\lambda-s)<\frac{\pi}{3}$ using \eqref{PCHRH} 
\begin{eqnarray*}
	P_{CH}^{RH}\big(\zeta(z)\big) &=& \Big[P_1\big(\ln(\lambda-s)\big)+(\lambda-s)P_2\big(\ln(\lambda-s)\big) +O\big((\lambda-s)^2\ln(\lambda-s)\big)\Big]\\
	&&\times\begin{pmatrix}
	e^{-i\frac{3\pi}{2}\nu} & 0\\
	0 & e^{i\frac{\pi}{2}\nu}\\
	\end{pmatrix},\ \ \lambda\rightarrow s
\end{eqnarray*}
where the matrix functions $P_1(\lambda)= \big(P_1^{ij}(\lambda)\big)$ and $P_2(\lambda)=\big(P_2^{ij}(\lambda)\big)$ can be determined from \eqref{PCHlocal}. For the remaining sectors $-\pi<\textnormal{arg}\,(\lambda-s)<-\frac{\pi}{3}$ and $\frac{\pi}{3}<\textnormal{arg}\,(\lambda-s)<\pi$ we can derive similar expansions, they differ from the latter only by multiplication with a triangular matrix. Now we combine the last expansion with \eqref{PCHRHparametrixright} and \eqref{ratioRHP}, so as $\lambda\rightarrow s$ in the sector $-\frac{\pi}{3}<\textnormal{arg}(\lambda-s)<\frac{\pi}{3}$, the left hand side of \eqref{comparison1} reads as
\begin{eqnarray*}
	&&R(z)V(z)L^{-1}(z)e^{s^3\vartheta(z)\sigma_3}\bigg|_{z=\frac{\lambda}{s}} = R(z)\sigma_1e^{-i\frac{\pi}{4}\sigma_3}B_r(z)e^{i\frac{\pi}{2}(\frac{1}{2}-\nu)\sigma_3}e^{-s^3\vartheta(1)\sigma_3}P_{CH}^{RH}\big(\zeta(z)\big)\\
	&&\times e^{i\frac{\pi}{4}\sigma_3}\sigma_1\bigg|_{z=\frac{\lambda}{s}}=R(1)\sigma_1e^{-i\frac{\pi}{4}\sigma_3}B_r(1)e^{i\frac{\pi}{2}(\frac{1}{2}-\nu)\sigma_3}
	e^{-s^3\vartheta(1)\sigma_3}P_1\big(\ln(\lambda-s)\big)\\
	&&\times\begin{pmatrix}
	e^{-i\frac{3\pi}{2}\nu} & 0\\
	0 & e^{i\frac{\pi}{2}\nu}\\
	\end{pmatrix}e^{i\frac{\pi}{4}\sigma_3}\sigma_1+(\lambda-s)\bigg[\Big(R'(1)\sigma_1e^{-i\frac{\pi}{4}\sigma_3}B_r(1)+R(1)\sigma_1e^{-i\frac{\pi}{4}\sigma_3}B_r'(1)\Big)\\
	&&e^{i\frac{\pi}{2}(\frac{1}{2}-\nu)\sigma_3}e^{-s^3\vartheta(1)\sigma_3}\frac{1}{s}P_1\big(\ln(\lambda-s)\big)
	\begin{pmatrix}
	e^{-i\frac{3\pi}{2}\nu} & 0\\
	0 & e^{i\frac{\pi}{2}\nu}\\	\end{pmatrix}e^{i\frac{\pi}{4}\sigma_3}\sigma_1+R(1)\sigma_1e^{-i\frac{\pi}{4}\sigma_3}B_r(1)\\
	&&\times e^{i\frac{\pi}{2}(\frac{1}{2}-\nu)\sigma_3} e^{-s^3\vartheta(1)\sigma_3}P_2\big(\ln(\lambda-s)\big)\begin{pmatrix}
	e^{-i\frac{3\pi}{2}\nu} & 0\\
	0 & e^{i\frac{\pi}{2}\nu}\\
	\end{pmatrix}e^{i\frac{\pi}{4}\sigma_3}\sigma_1\bigg]\\
	&&+O\big((\lambda-s)^2\ln(\lambda-s)\big).
\end{eqnarray*}
In the other sectors $-\pi<\textnormal{arg}\ (\lambda-s)<-\frac{\pi}{3}$ and $\frac{\pi}{3}<\textnormal{arg}\ (\lambda-s)<\pi$ we can derive similar identities, they differ from the latter one only by multiplication with a triangular matrix, see \eqref{PCHRH}. The right hand side in \eqref{comparison1} implies in the sector $-\frac{\pi}{3}<\textnormal{arg}(\lambda-s)<\frac{\pi}{3}$
\begin{eqnarray*}
	&&R(z)V(z)L^{-1}(z)e^{s^3\vartheta(z)\sigma_3}\Big|_{z=\frac{\lambda}{s}} = \check{X}(zs)\Bigg[I+\frac{\gamma}{2\pi i}\begin{pmatrix}
	-1 & 1\\
	-1 & 1\\
	\end{pmatrix}\ln\bigg(\frac{\lambda-s}{\lambda+s}\bigg)\Bigg]\\
	&&=\Big(\check{X}(s)+(\lambda-s)\check{X}'(s)+O\big((\lambda-s)^2\big)\Big)\Bigg[I+\frac{\gamma}{2\pi i}\begin{pmatrix}
	-1 & 1\\
	-1 & 1\\
	\end{pmatrix}\ln\bigg(\frac{\lambda-s}{\lambda+s}\bigg)\Bigg]\\
	&&=\begin{pmatrix}
	\alpha_{11}-\frac{\gamma}{2\pi i}(\check{X}_{11}(s)+\check{X}_{12}(s))\ln(\lambda-s) & \alpha_{12}+\frac{\gamma}{2\pi i}(\check{X}_{11}(s)+\check{X}_{12}(s))\ln(\lambda-s)\\
	\alpha_{21}-\frac{\gamma}{2\pi i}(\check{X}_{21}(s)+\check{X}_{22}(s))\ln(\lambda-s)& \alpha_{22}+\frac{\gamma}{2\pi i}(\check{X}_{21}(s)+\check{X}_{22}(s))\ln(\lambda-s)\\
	\end{pmatrix}\\
	&&+(\lambda-s)\\
	&&\times\begin{pmatrix}
		\beta_{11}-\frac{\gamma}{2\pi i}(\check{X}_{11}'(s)+\check{X}_{12}'(s))\ln(\lambda-s) & \beta_{12}+\frac{\gamma}{2\pi i}(\check{X}_{11}'(s)+\check{X}_{12}'(s))\ln(\lambda-s)\\
		\beta_{21}-\frac{\gamma}{2\pi i}(\check{X}_{21}'(s)+\check{X}_{22}'(s))\ln(\lambda-s) & \beta_{22}+\frac{\gamma}{2\pi i}(\check{X}_{21}'(s)+\check{X}_{22}'(s))\ln(\lambda-s)\\
		\end{pmatrix}\\
		&&+O\big((\lambda-s)^2\ln(\lambda-s)\big),\hspace{0.5cm}\lambda\rightarrow s
\end{eqnarray*}
with constants $\alpha_{ij}$ depending on $\check{X}_{ij}(s)$ and $\beta_{ij}$ involving both $\check{X}_{ij}(s)$ and $\check{X}_{ij}'(s)$. Comparing now left hand side and right hand side in \eqref{comparison1} we deduce
\begin{eqnarray*}
	\check{X}_{11}(s) &=& \Big[R(1)\big(B_r(1)\big)^{-1}\Big]_{11}e^{i\frac{\pi}{2}\nu}e^{s^3\vartheta(1)}\Big(c_0(-\nu)+c_1(-\nu)\bigg(\ln(16s^3+4xs)-i\frac{\pi}{2}\bigg)\Big)\\
	&&-\Big[R(1)\big(B_r(1)\big)^{-1}\Big]_{12}e^{i\frac{\pi}{2}\nu}e^{-s^3\vartheta(1)}\Big(c_0(1-\nu)\frac{\Gamma(1-\nu)}{\Gamma(\nu)}\\
	&&+c_1(\nu)\bigg(\ln(16s^3+4xs)-i\frac{\pi}{2}\bigg)\Big)\\
	\check{X}_{12}(s)&=& -\Big[R(1)\big(B_r(1)\big)^{-1}\Big]_{11}e^{i\frac{\pi}{2}\nu}e^{s^3\vartheta(1)}\Big(c_0(1+\nu)\frac{\Gamma(1+\nu)}{\Gamma(-\nu)}\\
	&&+c_1(-\nu)\bigg(\ln(16s^3+4xs)+i\frac{\pi}{2}\bigg)\Big)\\
	&&+\Big[R(1)\big(B_r(1)\big)^{-1}\Big]_{12}e^{i\frac{\pi}{2}\nu}e^{-s^3\vartheta(1)}\Big(c_0(\nu)+c_1(\nu)\bigg(\ln(16s^3+4xs)+i\frac{\pi}{2}\bigg)\Big)
\end{eqnarray*}
hence in particular
\begin{equation*}
	\check{X}_{11}(s)+\check{X}_{12}(s) = \frac{2\pi i}{\gamma}e^{i\frac{\pi}{2}\nu}\Bigg( \Big[R(1)\big(B_r(1)\big)^{-1}\Big]_{11}\frac{e^{s^3\vartheta(1)}}{\Gamma(-\nu)}-\Big[R(1)\big(B_r(1)\big)^{-1}\Big]_{12}\frac{e^{-s^3\vartheta(1)}}{\Gamma(\nu)}\Bigg).
\end{equation*}
Furthermore
\begin{eqnarray*}		
	\check{X}_{21}(s)&=&\Big[R(1)\big(B_r(1)\big)^{-1}\Big]_{21}e^{i\frac{\pi}{2}\nu}e^{s^3\vartheta(1)}\Big(c_0(-\nu)+c_1(-\nu)\bigg(\ln(16s^3+4xs)-i\frac{\pi}{2}\bigg)\Big)\\
	&&-\Big[R(1)\big(B_r(1)\big)^{-1}\Big]_{22}e^{i\frac{\pi}{2}\nu}e^{-s^3\vartheta(1)}\Big(c_0(1-\nu)\frac{\Gamma(1-\nu)}{\Gamma(\nu)}\\
	&&+c_1(\nu)\bigg(\ln(16s^3+4xs)-i\frac{\pi}{2}\bigg)\Big)
\end{eqnarray*}
and
\begin{eqnarray*}
	\check{X}_{22}(s)&=&-\Big[R(1)\big(B_r(1)\big)^{-1}\Big]_{21}e^{i\frac{\pi}{2}\nu}e^{s^3\vartheta(1)}\Big(c_0(1+\nu)\frac{\Gamma(1+\nu)}{\Gamma(-\nu)}\\
	&&+c_1(-\nu)\bigg(\ln(16s^3+4xs)+i\frac{\pi}{2}\bigg)\Big)\\
	&&+\Big[R(1)\big(B_r(1)\big)^{-1}\Big]_{22}e^{i\frac{\pi}{2}\nu}e^{-s^3\vartheta(1)}\Big(c_0(\nu)+c_1(\nu)\bigg(\ln(16s^3+4xs)+i\frac{\pi}{2}\bigg)\Big)
\end{eqnarray*}
which implies
\begin{equation*}
	\check{X}_{21}(s)+\check{X}_{22}(s) = \frac{2\pi i}{\gamma}e^{i\frac{\pi}{2}\nu}\Bigg(
	\Big[R(1)\big(B_r(1)\big)^{-1}\Big]_{21}\frac{e^{s^3\vartheta(1)}}{\Gamma(-\nu)}-\Big[R(1)\big(B_r(1)\big)^{-1}\Big]_{22}\frac{e^{-s^3\vartheta(1)}}{\Gamma(\nu)}\Bigg).
\end{equation*}
Comparing after that terms of $O((\lambda-s)\ln(\lambda-s))$ we obtain
\begin{eqnarray*}\label{compar3}
	&&\check{X}_{11}'(s)+\check{X}_{12}'(s) = \frac{2\pi i}{\gamma}e^{i\frac{\pi}{2}\nu}\Bigg(\bigg[\Big(R'(1)-R(1)\frac{\nu \sigma_3}{2}\frac{3+\frac{x}{4s^2}}{1+\frac{x}{4s^2}}\Big)\big(B_r(1)\big)^{-1}\bigg]_{11}\frac{e^{s^3\vartheta(1)}}{s\Gamma(-\nu)}\\
	&&-\bigg[\Big(R'(1)-R(1)\frac{\nu \sigma_3}{2}\frac{3+\frac{x}{4s^2}}{1+\frac{x}{4s^2}}\Big)\big(B_r(1)\big)^{-1}\bigg]_{12}\frac{e^{-s^3\vartheta(1)}}{s\Gamma(\nu)}+(8s^2+2x)\nonumber\\
	&&\times\bigg\{\Big[R(1)\big(B_r(1)\big)^{-1}\Big]_{11}\Big(\frac{i}{2}+i\nu\Big)\frac{e^{s^3\vartheta(1)}}{\Gamma(-\nu)}
	-\Big[R(1)\big(B_r(1)\big)^{-1}\Big]_{12}\Big(i\nu-\frac{i}{2}\Big)\frac{e^{-s^3\vartheta(1)}}{\Gamma(\nu)}\bigg\}\Bigg)\nonumber
\end{eqnarray*}
and
\begin{eqnarray*}\label{compar4}
	&&\check{X}_{21}'(s)+\check{X}_{22}'(s) =\frac{2\pi i}{\gamma}e^{i\frac{\pi}{2}\nu}\Bigg(\bigg[\Big(R'(1)-R(1)\frac{\nu \sigma_3}{2}\frac{3+\frac{x}{4s^2}}{1+\frac{x}{4s^2}}\Big)\big(B_r(1)\big)^{-1}\bigg]_{21}\frac{e^{s^3\vartheta(1)}}{s\Gamma(-\nu)}\\
	&&-\bigg[\Big(R'(1)-R(1)\frac{\nu \sigma_3}{2}\frac{3+\frac{x}{4s^2}}{1+\frac{x}{4s^2}}\Big)\big(B_r(1)\big)^{-1}\bigg]_{22}\frac{e^{-s^3\vartheta(1)}}{s\Gamma(\nu)}+(8s^2+2x)\nonumber\\
	&&\times\bigg\{\Big[R(1)\big(B_r(1)\big)^{-1}\Big]_{21}\Big(\frac{i}{2}+i\nu\Big)\frac{e^{s^3\vartheta(1)}}{\Gamma(-\nu)}
	-\Big[R(1)\big(B_r(1)\big)^{-1}\Big]_{22}\Big(i\nu-\frac{i}{2}\Big)\frac{e^{-s^3\vartheta(1)}}{\Gamma(\nu)}\bigg\}\Bigg)\nonumber.
\end{eqnarray*}
Although we derived the previous identities from a comparison in the sector $-\frac{\pi}{3}<\textnormal{arg}(\lambda-s)<\frac{\pi}{3}$, the same identities also hold in the other two sectors.
There one uses the correct triangular matrices in \eqref{PCHRH} on the left hand side. Thus all previously derived identities follow in fact from a comparison in a full neighborhood of $\lambda=+s$. Moving on to a neighborhood of $\lambda=-s$, a completely similar analysis provides us with
\begin{eqnarray*}
	\check{X}_{11}(-s)+\check{X}_{12}(-s) &=& \frac{2\pi i}{\gamma}e^{i\frac{\pi}{2}\nu}\Bigg(\Big[R(-1)\big(B_l(-1)\big)^{-1}\Big]_{12}\frac{e^{-s^3\vartheta(-1)}}{\Gamma(-\nu)}\\
	&&\Big[R(-1)\big(B_l(-1)\big)^{-1}\Big]_{11}\frac{e^{s^3\vartheta(-1)}}{\Gamma(\nu)}\Bigg).
\end{eqnarray*}
and
\begin{eqnarray*}
	\check{X}_{21}(-s)+\check{X}_{22}(-s) &=& \frac{2\pi i}{\gamma}e^{i\frac{\pi}{2}\nu}\Bigg(\Big[R(-1)\big(B_l(-1)\big)^{-1}\Big]_{22}\frac{e^{-s^3\vartheta(-1)}}{\Gamma(-\nu)}\\
	&&-\Big[R(-1)\big(B_l(-1)\big)^{-1}\Big]_{21}\frac{e^{s^3\vartheta(-1)}}{\Gamma(\nu)}\Bigg).
\end{eqnarray*}
Also, comparing terms of $O((\lambda+s)\ln(\lambda+s))$, 
\begin{eqnarray*}
	&&\check{X}_{11}'(-s)+\check{X}_{12}'(-s) =\frac{2\pi i}{\gamma}e^{i\frac{\pi}{2}\nu}\Bigg(\Big[\bigg(R'(-1)-R(-1)\frac{\nu\sigma_3}{2}\frac{3+\frac{x}{4s^2}}{1+\frac{x}{4s^2}}\bigg)\big(B_l(-1)\big)^{-1}\Big]_{12}\\
	&&\times\frac{e^{-s^3\vartheta(-1)}}{s\Gamma(-\nu)}-\Big[\bigg(R'(-1)-R(-1)\frac{\nu\sigma_3}{2}\frac{3+\frac{x}{4s^2}}{1+\frac{x}{4s^2}}\bigg)\big(B_l(-1)\big)^{-1}\Big]_{11}\frac{e^{s^3\vartheta(-1)}}{s\Gamma(\nu)}\\	&&+(8s^2+2x)\bigg\{\Big[R(-1)\big(B_l(-1)\big)^{-1}\Big]_{11}\Big(\frac{i}{2}+i\nu\Big)\frac{e^{s^3\vartheta(-1)}}{\Gamma(\nu)}\\
	&&-\Big[R(-1)\big(B_l(-1)\big)^{-1}\Big]_{12}\Big(\frac{3i}{2}+i\nu\Big)\frac{e^{-s^3\vartheta(-1)}}{\Gamma(-\nu)}\bigg\}\Bigg)
\end{eqnarray*}
and
\begin{eqnarray*}
	&&\check{X}_{21}'(-s)+\check{X}_{22}'(-s)=\frac{2\pi i}{\gamma}e^{i\frac{\pi}{2}\nu}\Bigg(\Big[\bigg(R'(-1)-R(-1)\frac{\nu\sigma_3}{2}\frac{3+\frac{x}{4s^2}}{1+\frac{x}{4s^2}}\bigg)\big(B_l(-1)\big)^{-1}\Big]_{22}\\
	&&\times\frac{e^{-s^3\vartheta(-1)}}{s\Gamma(-\nu)}-\Big[\bigg(R'(-1)-R(-1)\frac{\nu\sigma_3}{2}\frac{3+\frac{x}{4s^2}}{1+\frac{x}{4s^2}}\bigg)\big(B_l(-1)\big)^{-1}\Big]_{21}\frac{e^{s^3\vartheta(-1)}}{s\Gamma(\nu)}\\
	&&+(8s^2+2x)\bigg\{\Big[R(-1)\big(B_l(-1)\big)^{-1}\Big]_{21}\Big(\frac{i}{2}+i\nu\Big)\frac{e^{s^3\vartheta(-1)}}{\Gamma(\nu)}\\
	&&-\Big[R(-1)\big(B_l(-1)\big)^{-1}\Big]_{22}\Big(\frac{3i}{2}+i\nu\Big)\frac{e^{-s^3\vartheta(-1)}}{\Gamma(-\nu)}\bigg\}\Bigg).
\end{eqnarray*}
We finish this section by evaluating the resolvent kernel at $\lambda=\pm s$. Recall \eqref{Fendpointlocal} and deduce
\begin{eqnarray*}
	F_1(\pm s) = \sqrt{\frac{\gamma}{2\pi i}}\big(\check{X}_{11}(\pm s)+\check{X}_{12}(\pm s)\big),&&F_2(\pm s) = \sqrt{\frac{\gamma}{2\pi i}}\big(\check{X}_{21}(\pm s)+\check{X}_{22}(\pm s)\big)\\
	F_1'(\pm s) = \sqrt{\frac{\gamma}{2\pi i}}\big(\check{X}_{11}'(\pm s)+\check{X}_{12}'(\pm s)\big),&& F_2'(\pm s) = \sqrt{\frac{\gamma}{2\pi i}}\big(\check{X}_{21}'(\pm s)+\check{X}_{22}'(\pm s)\big).
\end{eqnarray*}
Since
\begin{equation*}
	R(s,s) = F_1'(s)F_2(s)-F_2'(s)F_1(s)
\end{equation*}
we use the previously derived identities and obtain
\begin{eqnarray*}
	&&R(s,s)=\frac{2\pi i}{\gamma}e^{i\pi\nu}\Bigg(\Big[R'_{11}(1)R_{21}(1)-R'_{21}(1)R_{11}(1)\Big](16s^3+4xs)^{-2\nu}\frac{e^{2s^3\vartheta(1)}}{s\Gamma^2(-\nu)}\\
	&&+\Big[R'_{12}(1)R_{22}(1)-R'_{22}(1)R_{12}(1)\Big](16s^3+4xs)^{2\nu}\frac{e^{-2s^3\vartheta(1)}}{s\Gamma^2(\nu)}\\
	&&-\bigg[R'_{11}(1)R(1)_{22}-R'_{22}(1)R_{11}(1)+R'_{12}(1)R_{21}(1)-R'_{21}(1)R_{12}(1)\\
	&&-\Big(R_{11}(1)R_{22}(1)-R_{21}(1)R_{12}(1)\Big)\nu\frac{3+\frac{x}{4s^2}}{1+\frac{x}{4s^2}}\bigg]\frac{1}{s\Gamma(\nu)\Gamma(-\nu)}\\
	&&-\Big[R_{11}(1)R_{22}(1)-R_{21}(1)R_{12}(1)\Big]i\frac{8s^2+2x}{\Gamma(\nu)\Gamma(-\nu)}\Bigg).
\end{eqnarray*}
To simplify this expression, we have
\begin{prop}\label{prop5} $R(z)$ is unimodular for any $x,\gamma\in\mathbb{R}$, i.e. $\det R(z)\equiv 1$.
\end{prop}
\begin{proof} Since $\det P_{II}^{RH}(\zeta)\equiv 1$, we have $\det U(z) = 1$. Similarly $\det P_{CH}^{RH}(\zeta) = 1=\det\tilde{P}_{CH}^{RH}(\zeta)$, leading to $\det V(z) = 1= \det W(z)$. Thus the ratio RHP has a unimodular jump matrix $G_R(z)$, which means that the function $\det R(z)$ is entire. By normalization at infinity we end up with
\begin{equation*}
	\det R(z) = R_{11}(z)R_{22}(z)-R_{12}(z)R_{21}(z) \equiv 1,\ \ z\in\mathbb{C}.
\end{equation*}
\end{proof}
In light of the last proposition, we obtain
\begin{eqnarray}\label{RSSplus}
	&&R(s,s) = -ie^{i\pi\nu}\frac{2\pi i}{\gamma}\frac{8s^2+2x}{\Gamma(\nu)\Gamma(-\nu)}+\frac{2\pi i}{\gamma}\frac{\nu e^{i\pi\nu}}{s\Gamma(\nu)\Gamma(-\nu)}\frac{3+\frac{x}{4s^2}}{1+\frac{x}{4s^2}}-\frac{2\pi i}{\gamma}\frac{e^{i\pi\nu}}{s\Gamma(\nu)\Gamma(-\nu)}\nonumber\\
	&&\Big[R'_{11}(1)R_{22}(1)-R'_{22}(1)R_{11}(1)+R'_{12}(1)R_{21}(1)-R'_{21}(1)R_{12}(1)\Big]\nonumber\\
	&&+\frac{2\pi i}{\gamma}e^{i\pi\nu}\bigg(\Big[R'_{11}(1)R_{21}(1)-R'_{21}(1)R_{11}(1)\Big](16s^3+4xs)^{-2\nu}\frac{e^{2s^3\vartheta(1)}}{s\Gamma^2(-\nu)}\nonumber\\
	&&+\Big[R'_{12}(1)R_{22}(1)-R'_{22}(1)R_{12}(1)\Big](16s^3+4xs)^{2\nu}\frac{e^{-2s^3\vartheta(1)}}{s\Gamma^2(\nu)}\bigg).
\end{eqnarray}
Similarly, using Proposition \ref{prop5} once more,
\begin{eqnarray}\label{RSSminus}
	&&R(-s,-s) =-ie^{i\pi\nu}\frac{2\pi i}{\gamma}\frac{8s^2+2x}{\Gamma(\nu)\Gamma(-\nu)}+\frac{2\pi i}{\gamma}\frac{\nu e^{i\pi\nu}}{s\Gamma(\nu)\Gamma(-\nu)}\frac{3+\frac{x}{4s^2}}{1+\frac{x}{4s^2}}-\frac{2\pi i}{\gamma}\frac{e^{i\pi\nu}}{s\Gamma(\nu)\Gamma(-\nu)}\nonumber\\
	&&\Big[R'_{11}(-1)R_{22}(-1)-R'_{22}(-1)R_{11}(-1)+R'_{12}(-1)R_{21}(-1)-R'_{21}(-1)R_{12}(-1)\Big]\nonumber\\
	&&+\frac{2\pi i}{\gamma}e^{i\pi\nu}\bigg(\Big[R'_{11}(-1)R_{21}(-1)-R'_{21}(-1)R_{11}(-1)\Big](16s^3+4xs)^{2\nu}\frac{e^{-2s^3\vartheta(1)}}{s\Gamma^2(\nu)}\nonumber\\	&&+\Big[R'_{12}(-1)R_{22}(-1)-R'_{22}(-1)R_{12}(-1)\Big](16s^3+4xs)^{-2\nu}\frac{e^{2s^3\vartheta(1)}}{s\Gamma^2(-\nu)}\bigg).
\end{eqnarray}


\section{Asymptotics of $\ln\det(I-\gamma K_{\textnormal{csin}})$ - proof of theorem \ref{theo1}}\label{sec17}

The stated asymptotics \eqref{theo1result} is a direct consequence of Proposition \ref{prop3}. Since the relevant estimations \eqref{esti4} and \eqref{esti5} were uniform on any compact subset of the set \eqref{excset1}, the asymptotic series for 
\begin{equation*}
	\frac{\partial}{\partial\gamma}\ln\det\left(I-\gamma K_{\textnormal{csin}}\right)
\end{equation*}
can be integrated with respect to $\gamma$, which leads to \eqref{theo1result} including the constant term $\chi$. We trace back the transformations
\begin{eqnarray*}
	X(\lambda)\mapsto T(z)\mapsto S(z)\mapsto R(z)
\end{eqnarray*}
and deduce first asymptotic expansions for the coefficients $X_1,X_2$ and $X_3$ in \eqref{gammaderiv5}. First
\begin{eqnarray}\label{X1exact}
	X_1 &=&\lim_{\lambda\rightarrow\infty}\big(\lambda\Big(X(\lambda)e^{-i(\frac{4}{3}\lambda^3+x\lambda)\sigma_3}-I\Big)\big) = -2\nu\sigma_3s +\frac{is}{2\pi}\int\limits_{\Sigma_R}R_-(w)\big(G_R(w)-I\big)dw\nonumber\\
	&=&-2\nu\sigma_3s +\frac{is}{2\pi}\int\limits_{\Sigma_R}\big(R_-(w)-I\big)\big(G_R(w)-I\big)dw +\frac{is}{2\pi}\int\limits_{\Sigma_R}\big(G_R(w)-I\big)dw
\end{eqnarray}
and we have for $z\in\Sigma_R$ from \eqref{integraleq1},\eqref{esti5} and \eqref{Absegurmatch}
\begin{eqnarray*}
	R_-(z)-I &=& \frac{1}{2\pi i}\int\limits_{C_0}\big(G_R(w)-I\big)\frac{dw}{w-z_-} +O\big(s^{-2}\big)\\
	&=&\frac{1}{2\pi i}\int\limits_{C_0}\big(B_0(w)\big)^{-1}\begin{pmatrix}
	v & ue^{-2\pi i\nu}\\
	-ue^{2\pi i\nu} & -v\\
	\end{pmatrix}B_0(w)\frac{idw}{2sw(w-z_-)} +O\big(s^{-2}\big)\\
	&=&\frac{i}{2sz}\bigg[\begin{pmatrix}
	v & u\\
	-u & -v\\
	\end{pmatrix} -\big(B_0(z)\big)^{-1}\begin{pmatrix}
	v & ue^{-2\pi i\nu}\\
	-ue^{2\pi i\nu} & -v\\
	\end{pmatrix}B_0(z)\bigg]+O\big(s^{-2}\big).
\end{eqnarray*}
The latter expansion will be improved via iteration
\begin{eqnarray*}
	R_-(z)-I&=&\frac{1}{2\pi i}\int\limits_{C_0}\big(R_-(w)-I\big)\big(G_R(w)-I\big)\frac{dw}{w-z_-}\\
	&&+\frac{1}{2\pi i}\int\limits_{C_0}\big(G_R(w)-I\big)\frac{dw}{w-z_-}+O\big(s^{-3}\big),\ \ s\rightarrow\infty.
\end{eqnarray*}
If $[\cdot,\cdot]$ denotes the usual matrix commutator on $\mathbb{C}^{2\times 2}$, \eqref{Absegurmultiplier} implies for any constant matrix $\Lambda\in\mathbb{C}^{2\times 2}$
\begin{eqnarray*}
	\big(B_0(w)\big)^{-1}\Lambda B_0(w) &=& \big(B_0(0)\big)^{-1}\Lambda B_0(0) +2\nu w\Big[\big(B_0(0)\big)^{-1}
\Lambda B_0(0),\sigma_3\Big]\\
&&+4\nu^2 w^2\Big[\big(B_0(0)\big)^{-1}\Lambda B_0(0),\sigma_3\Big]\sigma_3+O\big(w^3\big),\ \ w\rightarrow 0.
\end{eqnarray*}
This in turn leads to
\begin{eqnarray*}
	&&R_-(z)-I =\frac{i}{2sz}\bigg[\begin{pmatrix}
	v & u\\
	-u & -v\\
	\end{pmatrix} -\big(B_0(z)\big)^{-1}\begin{pmatrix}
	v & ue^{-2\pi i\nu}\\
	-ue^{2\pi i\nu} & -v\\
	\end{pmatrix}B_0(z)\bigg]\\
	&&-\frac{1}{4s^2z^2}\begin{pmatrix}
	v & u\\
	-u & -v\\
	\end{pmatrix}\bigg[\begin{pmatrix}
	v & u\\
	-u & -v\\
	\end{pmatrix} -\big(B_0(z)\big)^{-1}\begin{pmatrix}
	v & ue^{-2\pi i\nu}\\
	-ue^{2\pi i\nu} & -v\\
	\end{pmatrix}B_0(z)\bigg]\\
	&&+\frac{1}{8s^2z^2}\bigg[\begin{pmatrix}
	u^2-v^2 & 2(u_x+uv)\\
	2(u_x+uv) & u^2-v^2\\
	\end{pmatrix}-\big(B_0(z)\big)^{-1}\\
	&&\times\begin{pmatrix}
	u^2-v^2 & 2(u_x+uv)\\
	2(u_x+uv) & u^2-v^2\\
	\end{pmatrix}B_0(z)\bigg]+\frac{\nu}{s^2z}\begin{pmatrix}
	u^2 & -u_x\\
	u_x & -u^2\\
	\end{pmatrix}+O\big(s^{-3}\big),\ \ s\rightarrow\infty
\end{eqnarray*}
for any $z\in\Sigma_R$. Back to \eqref{X1exact}, one starts with
\begin{eqnarray*}
	&&I_1 = \int\limits_{\Sigma_R}\big(R_-(w)-I\big)\big(G_R(w)-I\big)dw = \frac{2\pi i\nu}{s^2}\begin{pmatrix}
	-u^2 & -uv\\
	uv & u^2 \\
	\end{pmatrix}\\
	&&+\frac{i(-2\pi i)\nu^2}{s^3}\begin{pmatrix}
	2uu_x & -2u^3+v(u_x+uv)+\frac{u}{2}(u^2-v^2)\\
	2u^3-v(u_x+uv)-\frac{u}{2}(u^2-v^2) & -2uu_x\\
	\end{pmatrix}\\
	&&+O\big(s^{-4}\big),\ \ s\rightarrow\infty
\end{eqnarray*}
and moves on to 
\begin{eqnarray*}
	I_2&=&\int\limits_{\Sigma_R}\big(G_R(w)-I\big)dw=\frac{(-2\pi i)i}{2s}\begin{pmatrix}
	v & u\\
	-u & -v\\
	\end{pmatrix}+\frac{(-2\pi i)\nu}{s^2}\begin{pmatrix}
	0 & -(u_x+uv)\\
	u_x+uv & 0\\
	\end{pmatrix}\\
	&&+\frac{i(-2\pi i)\nu^2}{s^3}\begin{pmatrix}
	0 & -(\frac{u}{2}(u^2+v^2)+vu_x+xu)\\
	\frac{u}{2}(u^2+v^2)+vu_x+xu & 0\\
	\end{pmatrix}\\
	&&+\frac{(-2\pi i)i}{4s^3+xs}\begin{pmatrix}
	-\nu^2 & \sqrt{-\nu^2}\cos\sigma\\
	-\sqrt{-\nu^2}\cos\sigma & \nu^2\\
	\end{pmatrix} + O\big(s^{-4}\big)
\end{eqnarray*}
with
\begin{equation}\label{sigmadef}
	\sigma = \sigma(s,x,\gamma) = \frac{8}{3}s^3+2xs +\frac{\ln|1-\gamma|}{\pi}\ln\big(16s^3+4xs\big)-\textnormal{arg}\frac{\Gamma(1-\nu)}{\Gamma(\nu)}.
\end{equation}
Together
\begin{eqnarray*}
	&&I_1+I_2 = \frac{(-2\pi i)i}{2s}\begin{pmatrix}
	v & u\\
	-u & -v\\
	\end{pmatrix} +\frac{(-2\pi i)\nu}{s^2}\begin{pmatrix}
	u^2 & -u_x\\
	u_x & -u^2\\
	\end{pmatrix}\\
	&&+\frac{(-2\pi i)i\nu^2}{s^3}\begin{pmatrix}
	2uu_x & -u_{xx}\\
	u_{xx} & -2uu_x\\
	\end{pmatrix}+\frac{(-2\pi i)i}{4s^3+xs}\begin{pmatrix}
	-\nu^2 & \sqrt{-\nu^2}\cos\sigma\\
	-\sqrt{-\nu^2}\cos\sigma & \nu^2\\
	\end{pmatrix}+O\big(s^{-4}\big),
\end{eqnarray*}
where we used that $u=u(x,\gamma)$ solves the second Painlev\'e equation. The previous line leads us to the following expansion for $X_1$
\begin{eqnarray}\label{X1second}
	&&X_1 = -2\nu\sigma_3 s+\frac{i}{2}\begin{pmatrix}
	v & u\\
	-u & -v\\
	\end{pmatrix}+\frac{\nu}{s}\begin{pmatrix}
	u^2 & -u_x\\
	u_x & -u^2\\
	\end{pmatrix} +\frac{i\nu^2}{s^2}\begin{pmatrix}
	2uu_x & -u_{xx}\\
	u_{xx}& -2uu_x\\
	\end{pmatrix}\nonumber\\
	&&+\frac{i}{4s^2+x}\begin{pmatrix}
	-\nu^2 & \sqrt{-\nu^2}\cos\sigma\\
	-\sqrt{-\nu^2}\cos\sigma & \nu^2\\
	\end{pmatrix} + O\big(s^{-3}\big),\ \ s\rightarrow\infty.
\end{eqnarray}
Secondly
\begin{equation}\label{X2first}
	X_2 = 2\nu^2s^2 I-\frac{i\nu s^2}{\pi}\big(I_1+I_2\big)\sigma_3 +\frac{is^2}{2\pi}\int\limits_{\Sigma_R}R_-(w)\big(G_R(w)-I\big)wdw.
\end{equation}
We need to compute
\begin{eqnarray*}
	I_3 &=& \int\limits_{\Sigma_R}\big(R_-(w)-I\big)\big(G_R(w)-I\big)wdw\\
	&=&\frac{(-2\pi i)i\nu}{2s^3}\begin{pmatrix}
	u^2v+uu_x& \frac{u}{2}(u^2-v^2)\\
	\frac{u}{2}(u^2-v^2) & u^2v+uu_x\\
	\end{pmatrix} +O\big(s^{-4}\big)
\end{eqnarray*}
as well as
\begin{eqnarray*}
	I_4 &=& \int\limits_{\Sigma_R}\big(G_{R}(w)-I\big)wdw=\frac{(-2\pi i)}{8s^2}\begin{pmatrix}
	u^2-v^2 & 2(u_x+uv)\\
	2(u_x+uv) & u^2-v^2\\
	\end{pmatrix}\\
	&&+\frac{(-2\pi i)i\nu}{4s^3}\begin{pmatrix}
	0 & u(u^2+v^2)+2(vu_x+xu)\\
	u(u^2+v^2)+2(vu_x+xu) & 0\\
	\end{pmatrix}\\
	&&+\frac{2\pi i}{4s^3+xs}\begin{pmatrix}
	0 & \sqrt{-\nu^2}\sin\sigma\\
	-\sqrt{-\nu^2}\sin\sigma & 0\\
	\end{pmatrix} + O\big(s^{-4}\big),
\end{eqnarray*}
to obtain
\begin{eqnarray*}
	I_3+I_4 &=&\frac{(-2\pi i)}{8s^2}\begin{pmatrix}
	u^2-v^2 & 2(u_x+uv)\\
	2(u_x+uv) & u^2-v^2\\
	\end{pmatrix}\\
	&&+ \frac{(-2\pi i)i\nu}{2s^3}\begin{pmatrix}
	u^2v+uu_x & u^3+vu_x+xu\\
	u^3+vu_x+xu & u^2v+uu_x\\
	\end{pmatrix}\\
	&&+\frac{2\pi i}{4s^3+xs}\begin{pmatrix}
	0 & \sqrt{-\nu^2}\sin\sigma\\
	\sqrt{-\nu^2}\sin\sigma & 0\\
	\end{pmatrix} + O\big(s^{-4}\big),\ \ s\rightarrow\infty.
\end{eqnarray*}
Adding together in \eqref{X2first}
\begin{eqnarray}\label{X2second}
	X_2 &=& 2\nu^2s^2I-i\nu s\begin{pmatrix}
	v & -u\\
	-u & v\\
	\end{pmatrix}-2\nu^2\begin{pmatrix}
	u^2 & u_x\\
	u_x & u^2\\
	\end{pmatrix} +\frac{1}{8}\begin{pmatrix}
	u^2-v^2 & 2(u_x+uv)\\
	2(u_x+uv) & u^2-v^2\\
	\end{pmatrix}\nonumber\\
	&&-\frac{2i\nu^3}{s}\begin{pmatrix}
	2uu_x & xu+2u^3\\
	xu+2u^3 & 2uu_x\\
	\end{pmatrix} +\frac{i\nu}{2s}\begin{pmatrix}
	u^2v+uu_x & u^3+vu_x+xu\\
	u^3+vu_x+xu & u^2v+uu_x\\
	\end{pmatrix}\nonumber\\
	&&+\frac{2i\nu s}{4s^2+x}\begin{pmatrix}
	\nu^2 & \sqrt{-\nu^2}\cos\sigma\\
	\sqrt{-\nu^2}\cos\sigma & \nu^2\\
	\end{pmatrix}-\frac{s}{4s^2+x}\begin{pmatrix}
	0 & \sqrt{-\nu^2}\sin\sigma\\
	\sqrt{-\nu^2}\sin\sigma & 0\\
	\end{pmatrix}\nonumber\\
	&&+O\big(s^{-2}\big),\ \ s\rightarrow\infty.
\end{eqnarray}
Finally the computation of $X_3$
\begin{eqnarray}\label{X3first}
	X_3 &=& -\frac{2\nu}{3}s^3(1+2\nu^2)\sigma_3 + \frac{i\nu^2s^3}{\pi}\big(I_1+I_2\big) - \frac{i\nu s^3}{\pi}\big(I_3+I_4\big)\sigma_3\nonumber\\ &&+\frac{is^3}{2\pi}\int\limits_{\Sigma_R}R_-(w)\big(G_R(w)-I\big)w^2dw.
\end{eqnarray}
Since
\begin{equation*}
	I_5 = \int\limits_{\Sigma_R}\big(R_-(w)-I\big)\big(G_R(w)-I\big)w^2dw=O\big(s^{-4}\big)
\end{equation*}
and
\begin{eqnarray*}
	I_6 &=& \int\limits_{\Sigma_R}R_-(w)\big(G_R(w)-I\big)w^2dw=\frac{(-2\pi i)i}{4s^3+xs}\begin{pmatrix}
	-\nu^2 & \sqrt{-\nu^2}\cos\sigma\\
	-\sqrt{-\nu^2}\cos\sigma & \nu^2\\
	\end{pmatrix}\\
	&&+\frac{(-2\pi i)i}{48s^3}\begin{pmatrix}
	-(v^3-3vu^2+2(xv-uu_x)) & -3(u(u^2+v^2)+2(vu_x+xu))\\
	3(u(u^2+v^2)+2(vu_x+xu))& v^3-3vu^2+2(xv-uu_x)\\
	\end{pmatrix}\\
	&&+O\big(s^{-4}\big)
\end{eqnarray*}
we add together and obtain in \eqref{X3first}, using the previous results
\begin{eqnarray*}
	&&X_3 =-\frac{2\nu}{3}s^3(1+2\nu^2)\sigma_3+i\nu^2 s^2\begin{pmatrix}
	v & u\\
	-u & -v\\
	\end{pmatrix}+2\nu^3 s\begin{pmatrix}
	u^2 & -u_x\\
	u_x & -u^2\\
	\end{pmatrix}\\
	&&-\frac{\nu s}{4}\begin{pmatrix}
	u^2-v^2 & -2(u_x+uv)\\
	2(u_x+uv) & -(u^2-v^2)\\
	\end{pmatrix}+2i\nu^4\begin{pmatrix}
	2uu_x & -u_{xx}\\
	u_{xx} & -2uu_x\\
	\end{pmatrix}\\
	&&+\frac{2i\nu^2s^2}{4s^2+x}\begin{pmatrix}
	-\nu^2 & \sqrt{-\nu^2}\cos\sigma\\
	-\sqrt{-\nu^2}\cos\sigma & \nu^2\\
	\end{pmatrix}+\frac{is^2}{4s^2+x}\begin{pmatrix}
	-\nu^2 & \sqrt{-\nu^2}\cos\sigma\\
	-\sqrt{-\nu^2}\cos\sigma & \nu^2\\
	\end{pmatrix}\\
	&&+\frac{2\nu s^2}{4s^2+x}\begin{pmatrix}
	0 & -\sqrt{-\nu^2}\sin\sigma\\
	\sqrt{-\nu^2}\sin\sigma & 0\\
	\end{pmatrix}-i\nu^2\begin{pmatrix}
	u^2v+uu_x & -(u^3+vu_x+xu)\\
	u^3+vu_x+xu & -(u^2v+uu_x)\\
	\end{pmatrix}\\
	&&+\frac{i}{48}\begin{pmatrix}
	-(v^3-3vu^2+2(xv-uu_x)) & -3(u(u^2+v^2)+2(vu_x+xu))\\
	3(u(u^2+v^2)+2(vu_x+xu))& v^3-3vu^2+2(xv-uu_x)\\
	\end{pmatrix}+O\big(s^{-1}\big).
\end{eqnarray*}
With the given information at hand, \eqref{bcddef} and \eqref{efdef} lead to 
\begin{equation}\label{bpara}
	b = 2X_1^{12} = iu-\frac{2\nu}{s}u_x-\frac{2i\nu^2u_{xx}}{s^2}+\frac{2i\sqrt{-\nu^2}}{4s^2+x}\cos\sigma+O\big(s^{-3}\big)
\end{equation}
as well as
\begin{equation}\label{cpara}
	c = 2X_1^{21} = -b+O\big(s^{-3}\big),\ \ s\rightarrow\infty
\end{equation}
and
\begin{eqnarray}\label{dpara}
	d &=& ix+8iX_1^{12}X_1^{21}\\
	&=& ix +2iu^2-\frac{8\nu}{s}uu_x-\frac{8i\nu^2}{s^2}\big((u_x)^2+xu^2+2u^4\big)+\frac{8iu\sqrt{-\nu^2}}{4s^2+x}\cos\sigma+O\big(s^{-3}\big).\nonumber
\end{eqnarray}
Furthermore
\begin{equation*}
	e = 8i\Big(X_1^{12}X_1^{22}-X_2^{12}\Big)=-2iu_x + \frac{4\nu}{s}(xu+2u^3)+\frac{8is\sqrt{-\nu^2}}{4s^2+x}\sin\sigma+O\big(s^{-2}\big)
\end{equation*}
and
\begin{equation}\label{fpara}
	f = -e +O\big(s^{-2}\big)
\end{equation}
where we made use of the following identities, see \eqref{X1second} and \eqref{X2second}
\begin{equation}\label{symmiden}
	X_1^{21} = -X_1^{12}+O\big(s^{-3}\big),\ \ X_1^{11} = -X_1^{22}+O\big(s^{-3}\big),\ \ X_2^{21} = X_2^{12}+O\big(s^{-2}\big),\ \ s\rightarrow\infty.
\end{equation}
We have now derived enough information to evaluate the first terms listed in Proposition \ref{prop3}. Put 
\begin{eqnarray*}
	\mathcal{P}_1(s,x,\gamma) &=& -8i\Big((X_3)_{\gamma}+(X_1)_{\gamma}(X_1^2-X_2)-(X_2)_{\gamma}X_1\Big)^{11}-2d\Big((X_1)_{\gamma}\Big)^{11}\\
	&&+4ib\Big((X_2)_{\gamma}-(X_1)_{\gamma}X_1\Big)^{21}-4ic\Big((X_2)_{\gamma}-(X_1)_{\gamma}X_1\Big)^{12}\\
	&&-e\Big((X_1)_{\gamma}\Big)^{21}-f\Big((X_1)_{\gamma}\Big)^{12}
\end{eqnarray*}
and notice that
\begin{equation*}
	\bigg(\frac{\partial}{\partial\gamma}X_k\bigg)^{ij} = \frac{\partial}{\partial\gamma}\Big(X_k^{ij}\Big),\ \ i,j,k=1,2.
\end{equation*}
From \eqref{fpara} and \eqref{symmiden} we obtain therefore
\begin{equation*}
	-e\Big((X_1)_{\gamma}\Big)^{21} = -f\Big((X_1)_{\gamma}\Big)^{12}+O\big(s^{-2}\big)
\end{equation*}
and via \eqref{bpara} and \eqref{symmiden}
\begin{equation*}
	4ib\Big((X_2)_{\gamma}\Big)^{21} = -4ic\Big((X_2)_{\gamma}\Big)^{12}+O\bigg(\frac{\ln s}{s^2}\bigg).
\end{equation*}
Since also
\begin{eqnarray*}
	-4ib\Big(\big(X_1\big)_{\gamma}X_1\Big)^{21} &=&  -4ib\Big(\big(X_1^{21}\big)_{\gamma}X_1^{11}+\big(X_1^{22}\big)_{\gamma}X_1^{21}\Big)\nonumber\\
	&=&4ic\Big(\big(X_1^{12}\big)_{\gamma}X_1^{22}+\big(X_1^{11}\big)_{\gamma}X_1^{12}\Big)+O\bigg(\frac{\ln s}{s^2}\bigg)\nonumber\\
	&=& 4ic\Big(\big(X_1\big)_{\gamma}X_1\Big)^{12}+O\bigg(\frac{\ln s}{s^2}\bigg)
\end{eqnarray*}
as $s\rightarrow\infty$ uniformly on any compact subset of the set \eqref{excset1}, we can simplify the expression for $\mathcal{P}_1(s,x,\gamma)$ asymptotically
\begin{eqnarray}\label{P1def}
	\mathcal{P}_1(s,x,\gamma)&=&-8i\Big(\big(X_3\big)_{\gamma}+\big(X_1\big)_{\gamma}(X_1^2-X_2)-\big(X_2\big)_{\gamma}X_1\Big)^{11}-2d\Big(\big(X_1\big)_{\gamma}\Big)^{11}\nonumber\\
	&&+8ib\Big(\big(X_2\big)_{\gamma}-\big(X_1\big)_{\gamma}X_1\Big)^{21}-2e\Big(\big(X_1\big)_{\gamma}\Big)^{21}+O\bigg(\frac{\ln s}{s^2}\bigg).
\end{eqnarray}
Next from \eqref{X1second}
\begin{eqnarray*}
	X_1^2 &=& \bigg[4\nu^2 s^2 -2i\nu vs-4\nu^2u^2+\frac{u^2-v^2}{4}-\frac{8i\nu^3}{s}uu_x+\frac{4i\nu^3 s}{4s^2+x}+\frac{i\nu}{s}(vu^2+uu_x)\bigg]I\\
	&&+O\big(s^{-2}\big)
\end{eqnarray*}
with $I$ denoting the $2\times 2$ identity matrix. Thus 
\begin{eqnarray*}
	&&\Big(\big(X_1\big)_{\gamma}X_1^2\Big)^{11}=\Big[-2\nu_{\gamma}s+\frac{i}{2}v_{\gamma}+\frac{(\nu u^2)_{\gamma}}{s}+\frac{i}{s^2}\big(2\nu^2uu_x\big)_{\gamma}-\frac{i}{4s^2+x}\big(\nu^2\big)_{\gamma}\\
	&&+O\bigg(\frac{\ln s}{s^3}\bigg)\Big]\Big[4\nu^2 s^2 -2i\nu vs-4\nu^2u^2+\frac{u^2-v^2}{4}-\frac{8i\nu^3}{s}uu_x+\frac{4i\nu^3 s}{4s^2+x}\\
	&&+\frac{i\nu}{s}(vu^2+uu_x)+O\big(s^{-2}\big)\Big]\\
	&=&-8\nu^2\nu_{\gamma}s^3+4i\nu\nu_{\gamma}vs^2+2i\nu^2s^2v_{\gamma}+8\nu^2\nu_{\gamma}u^2s-\nu_{\gamma}s\frac{u^2-v^2}{2}+\nu svv_{\gamma}\\
	&&+4\nu^2s\big(\nu u^2\big)_{\gamma}+16i\nu_{\gamma}\nu^3uu_x-\frac{8i\nu^3\nu_{\gamma}s^2}{4s^2+x}-2i\nu\nu_{\gamma}\big(vu^2+uu_x\big)-2i\nu^2u^2v_{\gamma}\\
	&&+\frac{i}{2}v_{\gamma}\frac{u^2-v^2}{4}-2i\nu v\big(\nu u^2\big)_{\gamma}+4i\nu^2\big(2\nu^2 uu_x\big)_{\gamma}-\frac{4i\nu^2 s^2}{4s^2+x}\big(\nu^2\big)_{\gamma}+O\bigg(\frac{\ln s}{s}\bigg).
\end{eqnarray*}
Moving on, we use
\begin{equation*}
	\Big(\big(X_1\big)_{\gamma}X_2 +\big(X_2\big)_{\gamma}X_1\Big)^{11} = \big(X_1^{11}X_2^{11}\big)_{\gamma} +\big(X_1^{12}\big)_{\gamma}X_2^{21}+\big(X_2^{12}\big)_{\gamma}X_1^{21}
\end{equation*}
and obtain
\begin{eqnarray*}
	&&\Big(\big(X_1\big)_{\gamma}X_2 +\big(X_2\big)_{\gamma}X_1\Big)^{11}=\bigg[-4\nu^3s^3+3i\nu^2s^2v +6\nu^3u^2s-3i\nu^2u^2v+12i\nu^4uu_x\\
	&&-i\nu^2uu_x-\frac{6i\nu^4s^2}{4s^2+x}-\nu s\frac{u^2-3v^2}{4}+\frac{iv}{2}\frac{u^2-v^2}{8}+O\big(s^{-1}\big)\bigg]_{\gamma}-\frac{\nu s}{2}uu_{\gamma}+\frac{s}{2}u\big(\nu u\big)_{\gamma}\\
	&&-i\nu^2u_xu_{\gamma}+\frac{i}{2}u_{\gamma}\frac{u_x+uv}{4}-i\nu u\big(\nu u_x\big)_{\gamma}+i\nu u_x\big(\nu u\big)_{\gamma}+iu\big(\nu^2 u_x\big)_{\gamma}-\frac{i}{2}u\frac{(u_x+uv)_{\gamma}}{4}\\
	&& + O\big(s^{-1}\big).
\end{eqnarray*} 
Furthermore
\begin{eqnarray*}
	\big(X_1\big)_{\gamma}X_1 &=& 4\nu\nu_{\gamma}s^2I-i\nu_{\gamma}s\begin{pmatrix}
	v & u\\
	u & v\\
	\end{pmatrix}-i\nu s\begin{pmatrix}
	v_{\gamma} & -u_{\gamma}\\
	-u_{\gamma} & v_{\gamma}\\
	\end{pmatrix}-2\nu\nu_{\gamma}\begin{pmatrix}
	u^2& -u_x\\
	-u_x & u^2\\
	\end{pmatrix}\\
	&&-\frac{1}{4}\begin{pmatrix}
	vv_{\gamma} - uu_{\gamma} & uv_{\gamma}-vu_{\gamma}\\
	uv_{\gamma} -vu_{\gamma} & vv_{\gamma}-uu_{\gamma}\\
	\end{pmatrix}-2\nu\begin{pmatrix}
	(\nu u^2)_{\gamma} & (\nu u_x)_{\gamma}\\
	(\nu u_x)_{\gamma} & (\nu u^2)_{\gamma}\\
	\end{pmatrix}+O\big(s^{-1}\big)
\end{eqnarray*}
and thus
\begin{eqnarray*}
&&\Big(\big(X_2\big)_{\gamma}-\big(X_1\big)_{\gamma}X_1\Big)^{21} = is(\nu u)_{\gamma}+i\nu_{\gamma}su-i\nu su_{\gamma}-2(\nu^2u_x)_{\gamma}+\frac{(u_x+uv)_{\gamma}}{4}\\
&&-2\nu\nu_{\gamma}u_x+\frac{1}{4}(uv_{\gamma}-vu_{\gamma})+2\nu(\nu u_x)_{\gamma}+O\big(s^{-1}\big),
\end{eqnarray*}
which implies
\begin{eqnarray*}
	&&b\Big(\big(X_2\big)_{\gamma}-\big(X_1\big)_{\gamma}X_1\Big)^{21} =-su(\nu u)_{\gamma}-su^2\nu_{\gamma}+s\nu uu_{\gamma}-2iu(\nu^2 u_x)_{\gamma}+iu\frac{(u_x+uv)_{\gamma}}{4}\\
	&&-2i\nu\nu_{\gamma}uu_x+iu\frac{uv_{\gamma}-vu_{\gamma}}{4}+2i\nu u(\nu u_x)_{\gamma}-2i\nu u_x(\nu u)_{\gamma}-2i\nu\nu_{\gamma}uu_x\\
	&&+2i\nu^2u_xu_{\gamma}+O\big(s^{-1}\big).
\end{eqnarray*}
Finally
\begin{equation*}
	d\Big(\big(X_1\big)_{\gamma}\Big)^{11} = -2i\nu_{\gamma}sx-4i\nu_{\gamma}u^2s-\frac{xv_{\gamma}}{2}-u^2v_{\gamma}+16\nu\nu_{\gamma}uu_x+O\big(s^{-1}\big)
\end{equation*}
and
\begin{equation*}
	e\Big(\big(X_1\big)_{\gamma}\Big)^{21} = -u_xu_{\gamma}+O\big(s^{-1}\big).
\end{equation*}
Let us write
\begin{equation*}
	\mathcal{P}_1(s,x,\gamma) = s^3\mathcal{P}_1^{(3)}(x,\gamma)+s^2\mathcal{P}_1^{(2)}(x,\gamma) +s\mathcal{P}_1^{(1)}(x,\gamma) +\mathcal{P}_1^{(0)}(x,\gamma)+O\bigg(\frac{\ln s}{s}\bigg)
\end{equation*}
where $\mathcal{P}_1^{(i)}(x,\gamma)$ are independent of $s$. Since
\begin{equation*}
	\nu\big|_{\gamma=0} = 0
\end{equation*}
we get from \eqref{P1def} and the previous computations 
\begin{eqnarray*}			
	\int\limits_0^{\gamma}\mathcal{P}_1^{(3)}(x,t)dt&=&\int\limits_0^{\gamma}\Bigg[-8i\bigg(-\frac{2\nu}{3}(1+2\nu^2)\bigg)_{t}
-8i\big(-8\nu^2\nu_{t}\big)+8i\big(-4\nu^3\big)_{t}\Bigg]dt\\
&=&\frac{16}{3}i\nu.
\end{eqnarray*}
Next
\begin{eqnarray*}
	\int\limits_0^{\gamma}\mathcal{P}_1^{(2)}(x,t)dt &=&\int\limits_0^{\gamma}\Bigg[-8i\big(i\nu^2 v\big)_{t}-8i\big(4i\nu\nu_{t}v+2i\nu^2v_{t}\big)+8i\big(3i\nu^2 v\big)_t\Bigg]dt\\
	&=&0.
\end{eqnarray*}
and
\begin{eqnarray*}
	&&\int\limits_0^{\gamma}\mathcal{P}_1^{(1)}(x,t)dt =\int\limits_0^{\gamma}\Bigg[-8i\bigg(2t^3u^2-\frac{t}{4}(u^2-v^2)\bigg)_t-8i\bigg(8\nu^2\nu_tu^2-\nu_t\frac{u^2-v^2}{2}\\
	&&+\nu vv_t+4\nu^2(\nu u^2)_t\bigg)+8i\bigg(6\nu^3u^2-\nu\frac{u^2-3v^2}{4}\bigg)_t +8i\bigg(-\frac{\nu}{2}uu_t+\frac{u}{2}(\nu u)_t\bigg)\\
	&&+8i\bigg(iu\big(i(\nu u)_t+i\nu_tu-i\nu u_t\big)\bigg)-2\bigg(-2i\nu_tx-4i\nu_tu^2\bigg)\Bigg]dt\\
	&=&4i\nu x.
\end{eqnarray*}
We conclude the computations for $\mathcal{P}_1(s,x,\gamma)$ by evaluating
\begin{eqnarray*}
	&&\int\limits_0^{\gamma}\mathcal{P}_1^{(0)}(x,t)dt =\int\limits_0^{\gamma}\Bigg[-8i\bigg(4i\nu^4 uu_x-\frac{2i\nu^4 s^2}{4s^2+x}-\frac{is^2\nu^2}{4s^2+x}-i\nu^2(u^2v+uu_x)\\
	&&-\frac{i}{48}(v^3-3vu^2+2(xv-uu_x))\bigg)_t-8i\bigg(16i\nu^3\nu_{t}uu_x-\frac{8i\nu^3\nu_{t}s^2}{4s^2+x}-2i\nu\nu_{t}(vu^2+uu_x)
\end{eqnarray*}
\begin{eqnarray*}
	&&-2i\nu^2u^2v_{t}+iv_{t}\frac{u^2-v^2}{8}-2i\nu v(\nu u^2)_{t}+8i\nu^2(\nu^2uu_x)_{t}-\frac{4i\nu^2s^2}{4s^2+x}(\nu^2)_{t}\bigg)\\
	&&+8i\bigg(-3i\nu^2u^2v+12i\nu^4uu_x-i\nu^2uu_x-\frac{6i\nu^4s^2}{4s^2+x}+iv\frac{u^2-v^2}{16}\bigg)_t\\
	&&+8i\bigg(-i\nu^2u_xu_{t}+iu_{t}\frac{u_x+uv}{8}-i\nu u(\nu u_x)_{t}+i\nu u_x(\nu u)_{t}+iu(\nu^2u_x)_{t}\\
	&&-iu\frac{(u_x+uv)_{t}}{8}\bigg)-2\bigg(-\frac{xv_{t}}{2}-u^2v_{t}+16\nu\nu_{t}uu_x\bigg)+8i\bigg(-2iu(\nu^2u_x)_{t}\\
	&&+iu\frac{(u_x+uv)_{t}}{4}-2i\nu\nu_{t}uu_x+iu\frac{uv_{t}-vu_{t}}{4}+2i\nu u(\nu u_x)_{t}-2i\nu u_x(\nu u)_{t}\\
	&&-2i\nu\nu_{t}uu_x+2i\nu^2u_xu_{t}\bigg)-2\bigg(-u_xu_{t}\bigg)\bigg]dt\\
	&=&-2\nu^2+\frac{2}{3}(xv-uu_x)+2\int_0^{\gamma}u_xu_{t}dt.
\end{eqnarray*}
The following observation is useful
\begin{prop}\label{prop6} Let $u=u(x,\gamma),\gamma<1$ denote the Ablowitz-Segur solution of the second Painlev\'e equation $u_{xx}=xu+2u^3$ as described in Theorem \ref{theo1}. Put
\begin{equation*}
	F(x,\gamma) = \frac{2}{3}\big(xv(x,\gamma)-u(x,\gamma)u_x(x,\gamma)\big)+2\int\limits_0^{\gamma}u_x(x,t)u_t(x,t)dt,
\end{equation*}
then
\begin{equation*}
	F(x,\gamma) = -\int\limits_x^{\infty}(y-x)u^2(y,\gamma)dy.
\end{equation*}
\end{prop}
\begin{proof}
Using the differential equation for $u$ and integration by parts, we obtain
\begin{equation*}
	\frac{\partial}{\partial x}F(x,\gamma) = v(x,\gamma) = \big(u_x(x,\gamma)\big)^2-xu^2(x,\gamma)-u^4(x,\gamma)
\end{equation*}
and hence after integration
\begin{equation}\label{prop6eq}
	F(x,\gamma) = -\int\limits_x^{\infty}(y-x)u^2(y,\gamma)dy +C(\gamma)
\end{equation}
with a constant $C$, only depending on $\gamma$. Since $u$ decays exponentially fast as $x\rightarrow+\infty$, the same limit on both sides of \eqref{prop6eq} gives us the stated identity.
\end{proof}
At this point we summarize the previous computations. As $s\rightarrow\infty$ uniformly on any compact subset of \eqref{excset1}
\begin{equation}\label{firstsummary}
	\int\limits_0^{\gamma} \mathcal{P}_1(s,x,t)dt = i\nu\bigg(\frac{16}{3}s^3+4sx\bigg)-\int\limits_x^{\infty}(y-x)u^2(y,\gamma)dy +2(i\nu)^2+O\bigg(\frac{\ln s}{s}\bigg)
\end{equation}
To move further ahead let us define
\begin{equation*}
		\mathcal{P}_2(s,x,\gamma) = \big(A_{11}-A_{22}\big)\hat{H}_{11}(s)+A_{12}\hat{H}_{21}(s)+A_{21}\hat{H}_{12}(s),\ \hat{H}(s) = \frac{\partial\check{X}}{\partial\gamma}(s)\big(\check{X}(s)\big)^{-1}
\end{equation*}
with (compare \eqref{ABdef})
\begin{equation*}
	A = \frac{\gamma}{2\pi i}\check{X}(s)\begin{pmatrix}
	-1 & 1\\
	-1 & 1\\
	\end{pmatrix}\big(\check{X}(s)\big)^{-1}.
\end{equation*}
Since 
\begin{equation*}
	A_{11} = -\frac{\gamma}{2\pi i}\big(\check{X}_{11}(s)+\check{X}_{12}(s)\big)\big(\check{X}_{21}(s)+\check{X}_{22}(s)\big)
\end{equation*}
we can now use the identities derived in section \ref{sec16} for $\check{X}_{ij}(\pm s)$. With
\begin{eqnarray}\label{Rplusminus1}
	R(\pm 1) &=& I+\frac{1}{2\pi i}\int\limits_{\Sigma_R}R_-(w)\big(G_R(w)-I\big)\frac{dw}{w\mp 1}\\
	&=&I+\frac{1}{2\pi i}\int\limits_{C_0}\big(G_R(w)-I\big)\frac{dw}{w\mp 1}+O\big(s^{-2}\big)
	=I\pm\frac{i}{2s}\begin{pmatrix}
	v & u\\
	-u & -v\\
	\end{pmatrix}+O\big(s^{-2}\big)\nonumber
\end{eqnarray}
and the classical identity
\begin{equation*}
	\Gamma(z)\Gamma(1-z) = \frac{\pi}{\sin\pi z},\ \ z\in\mathbb{C}\backslash\mathbb{Z}
\end{equation*}
one concludes
\begin{equation*}
	A_{11} = \nu+O\big(s^{-1}\big),\ \ A_{22}=-\nu+O\big(s^{-1}\big)
\end{equation*}
uniformly on any compact subset of the set \eqref{excset1}. Also
\begin{equation}\label{A12}
	A_{12} = \nu(16s^3+4xs)^{-2\nu}e^{2s^3\vartheta(1)}\frac{\Gamma(\nu)}{\Gamma(-\nu)} +O\big(s^{-1}\big)
\end{equation}
and
\begin{equation}\label{A21}
	A_{21} = -\nu(16s^3+4xs)^{2\nu}e^{-2s^3\vartheta(1)}\frac{\Gamma(-\nu)}{\Gamma(\nu)} +O\big(s^{-1}\big).
\end{equation}
Next we use \eqref{Rplusminus1} to simplify the identities for $\check{X}_{ij}(s)$ obtained in section \ref{sec16}
\begin{eqnarray*}
	\check{X}_{11}(s) &=& (16s^3+4xs)^{-\nu}e^{i\frac{\pi}{2}\nu}e^{s^3\vartheta(1)}\Big(c_0(-\nu)+c_1(-\nu)\bigg(\ln(16s^3+4xs)-i\frac{\pi}{2}\bigg)\Big)\\
	&&+O\bigg(\frac{\ln s}{s}\bigg)\\
	\check{X}_{12}(s)&=&-(16s^3+4xs)^{-\nu}e^{i\frac{\pi}{2}\nu}e^{s^3\vartheta(1)}\Big(c_0(1+\nu)\frac{\Gamma(1+\nu)}{\Gamma(-\nu)}+c_1(-\nu)\bigg(\ln(16s^3+4xs)\\
	&&+i\frac{\pi}{2}\bigg)\Big)+O\bigg(\frac{\ln s}{s}\bigg)\\
	\check{X}_{21}(s)&=&-(16s^3+4xs)^{\nu}e^{i\frac{\pi}{2}\nu}e^{-s^3\vartheta(1)}\Big(c_0(1-\nu)\frac{\Gamma(1-\nu)}{\Gamma(\nu)}+c_1(\nu)\bigg(\ln(16s^3+4xs)\\
	&&-i\frac{\pi}{2}\bigg)\Big)+O\bigg(\frac{\ln s}{s}\bigg)
\end{eqnarray*}
\begin{eqnarray*}
	\check{X}_{22}(s)&=&(16s^3+4xs)^{\nu}e^{i\frac{\pi}{2}\nu}e^{-s^3\vartheta(1)}\Big(c_0(\nu)+c_1(\nu)\bigg(\ln(16s^3+4xs)+i\frac{\pi}{2}\bigg)\Big)\\
	&&+O\bigg(\frac{\ln s}{s}\bigg).
\end{eqnarray*}
Combined with \eqref{A12} and \eqref{A21}, we deduce the following asymptotics for $\mathcal{P}_2(s,x,\gamma)$
\begin{eqnarray*}
	&&\mathcal{P}_2(s,x,\gamma)=2\nu e^{i\pi\nu}\Big(c_0(\nu)+c_1(\nu)\bigg(\ln(16s^3+4xs)+i\frac{\pi}{2}\bigg)\Big)\Bigg[\nu_{\gamma}\bigg(i\frac{\pi}{2}-\ln(16s^3+4xs)\bigg)\\
	&&\times\Big(c_0(-\nu)+c_1(-\nu)\bigg(\ln(16s^3+4xs)-i\frac{\pi}{2}\bigg)\Big)+\Big(c_0(-\nu)+c_1(-\nu)\bigg(\ln(16s^3+4xs)\\
	&&-i\frac{\pi}{2}\bigg)\Big)_{\gamma}\Bigg]-2\nu e^{i\pi\nu}\Big(c_0(1-\nu)\frac{\Gamma(1-\nu)}{\Gamma(\nu)}+c_1(\nu)\bigg(\ln(16s^3+4xs)-i\frac{\pi}{2}\bigg)\Big)\Bigg[\nu_{\gamma}\bigg(i\frac{\pi}{2}\\
	&&-\ln(16s^3+4xs)\bigg)\Big(c_0(1+\nu)\frac{\Gamma(1+\nu)}{\Gamma(-\nu)}+c_1(-\nu)\bigg(\ln(16s^3+4xs)+i\frac{\pi}{2}\bigg)\Big)\\
	&&+\Big(c_0(1+\nu)\frac{\Gamma(1+\nu)}{\Gamma(-\nu)}+c_1(-\nu)\bigg(\ln(16s^3+4xs)+i\frac{\pi}{2}\bigg)\Big)_{\gamma}\Bigg]-\nu e^{i\pi\nu}\frac{\Gamma(\nu)}{\Gamma(-\nu)}\\
	&&\times\Big(c_0(\nu)+c_1(\nu)\bigg(\ln(16s^3+4xs)+i\frac{\pi}{2}\bigg)\Big)\Big(c_0(1-\nu)\frac{\Gamma(1-\nu)}{\Gamma(\nu)}+c_1(\nu)\bigg(\ln(16s^3+4xs)\\
	&&-i\frac{\pi}{2}\bigg)\Big)_{\gamma}+\nu e^{i\pi\nu}\frac{\Gamma(\nu)}{\Gamma(-\nu)}\Big(c_0(\nu)+c_1(\nu)\bigg(\ln(16s^3+4xs)+i\frac{\pi}{2}\bigg)\Big)_{\gamma}\\
	&&\times\Big(c_0(1-\nu)\frac{\Gamma(1-\nu)}{\Gamma(\nu)}+c_1(\nu)\bigg(\ln(16s^3+4xs)-i\frac{\pi}{2}\bigg)\Big)\\
	&&+\nu e^{i\pi\nu}\frac{\Gamma(-\nu)}{\Gamma(\nu)}\Big(c_0(-\nu)+c_1(-\nu)\bigg(\ln(16s^3+4xs)-i\frac{\pi}{2}\bigg)\Big)\Big(c_0(1+\nu)\frac{\Gamma(1+\nu)}{\Gamma(-\nu)}+c_1(-\nu)\\
	&&\times\bigg(\ln(16s^3+4xs)+i\frac{\pi}{2}\bigg)\Big)_{\gamma}-\nu e^{i\pi\nu}\frac{\Gamma(-\nu)}{\Gamma(\nu)}\Big(c_0(-\nu)+c_1(-\nu)\bigg(\ln(16s^3+4xs)-i\frac{\pi}{2}\bigg)\Big)_{\gamma}\\
	&&\times \Big(c_0(1+\nu)\frac{\Gamma(1+\nu)}{\Gamma(-\nu)}+c_1(-\nu)\bigg(\ln(16s^3+4xs)+i\frac{\pi}{2}\bigg)\Big)+O\bigg(\frac{(\ln s)^3}{s}\bigg),\ \ s\rightarrow\infty.
\end{eqnarray*}
What is left in the identitiy stated in Proposition \ref{prop3} is the term
\begin{equation*}
	\mathcal{P}_3(s,x,\gamma) = \big(B_{11}-B_{22}\big)\tilde{H}_{11}(-s)+B_{12}\tilde{H}_{21}(-s)+B_{21}\tilde{H}_{12}(-s),
\end{equation*}
with 
\begin{equation*}
	B_{11} = \nu+O\big(s^{-1}\big),\hspace{0.5cm} B_{12}=\nu(16s^3+4xs)^{2\nu}e^{-2s^3\vartheta(1)}\frac{\Gamma(-\nu)}{\Gamma(\nu)}+O\big(s^{-1}\big)
\end{equation*}
and
\begin{equation*}
	B_{21} = -\nu(16s^3+4xs)^{-2\nu}e^{2s^3\vartheta(1)}\frac{\Gamma(\nu)}{\Gamma(-\nu)}+O\big(s^{-1}\big),\ \ s\rightarrow\infty
\end{equation*}
which, also here, holds uniformly on any compact subset of \eqref{excset1}. Again \eqref{Rplusminus1} allows us to simplify the idenitites for $\check{X}_{ij}(-s)$ obtained in section \ref{sec16} and we are lead to the following asymptotics for $\mathcal{P}_3(s,x,\gamma)$
\begin{eqnarray*}
	&&\mathcal{P}_3(s,x,\gamma)=2\nu e^{i\pi\nu}\Big(c_0(-\nu)+c_1(-\nu)\bigg(\ln(16s^3+4xs)-i\frac{\pi}{2}\bigg)\Big)\Bigg[\nu_{\gamma}\bigg(i\frac{\pi}{2}\\
	&&+\ln(16s^3+4xs)\bigg)\Big(c_0(\nu)+c_1(\nu)\bigg(\ln(16s^3+4xs)+i\frac{\pi}{2}\bigg)\Big)+\Big(c_0(\nu)+c_1(\nu)\\
	&&\times\bigg(\ln(16s^3+4xs)+i\frac{\pi}{2}\bigg)\Big)_{\gamma}\Bigg]-2\nu e^{i\pi\nu}\Big(c_0(1+\nu)\frac{\Gamma(1+\nu)}{\Gamma(-\nu)}+c_1(-\nu)\\
	&&\times\bigg(\ln(16s^3+4xs)+i\frac{\pi}{2}\bigg)\Big)\Bigg[\nu_{\gamma}\bigg(i\frac{\pi}{2}+\ln(16s^3+4xs)\bigg)\Big(c_0(1-\nu)\frac{\Gamma(1-\nu)}{\Gamma(\nu)}\\
	&&+c_1(\nu)\bigg(\ln(16s^3+4xs)-i\frac{\pi}{2}\bigg)\Big)+\Big(c_0(1-\nu)\frac{\Gamma(1-\nu)}{\Gamma(\nu)}\\
	&&+c_1(\nu)\bigg(\ln(16s^3+4xs)-i\frac{\pi}{2}\bigg)\Big)_{\gamma}\Bigg]-\nu e^{i\pi\nu}\frac{\Gamma(-\nu)}{\Gamma(\nu)}\Big(c_0(-\nu)+c_1(-\nu)\\
	&&\times\bigg(\ln(16s^3+4xs)-i\frac{\pi}{2}\bigg)\Big)\Big(c_0(1+\nu)\frac{\Gamma(1+\nu)}{\Gamma(-\nu)}+c_1(-\nu)\bigg(\ln(16s^3+4xs)+i\frac{\pi}{2}\bigg)\Big)_{\gamma}\\
	&&+\nu e^{i\pi\nu}\frac{\Gamma(-\nu)}{\Gamma(\nu)}\Big(c_0(-\nu)+c_1(-\nu)\bigg(\ln(16s^3+4xs)-i\frac{\pi}{2}\bigg)\Big)_{\gamma}\Big(c_0(1+\nu)\frac{\Gamma(1+\nu)}{\Gamma(-\nu)}\\
	&&+c_1(-\nu)\bigg(\ln(16s^3+4xs)+i\frac{\pi}{2}\bigg)\Big)+\nu e^{i\pi\nu}\frac{\Gamma(\nu)}{\Gamma(-\nu)}\Big(c_0(\nu)+c_1(\nu)\bigg(\ln(16s^3+4xs)\\
	&&+i\frac{\pi}{2}\bigg)\Big)\Big(c_0(1-\nu)\frac{\Gamma(1-\nu)}{\Gamma(\nu)}+c_1(\nu)\bigg(\ln(16s^3+4xs)-i\frac{\pi}{2}\bigg)\Big)_{\gamma}\\
	&&-\nu e^{i\pi\nu}\frac{\Gamma(\nu)}{\Gamma(-\nu)}\Big(c_0(\nu)+c_1(\nu)\bigg(\ln(16s^3+4xs)+i\frac{\pi}{2}\bigg)\Big)_{\gamma}\\
	&&\times\Big(c_0(1-\nu)\frac{\Gamma(1-\nu)}{\Gamma(\nu)}+c_1(\nu)\bigg(\ln(16s^3+4xs)-i\frac{\pi}{2}\bigg)\Big)+O\bigg(\frac{(\ln s)^3}{s}\bigg),\ s\rightarrow\infty.
\end{eqnarray*}
The two expansions for $\mathcal{P}_2(s,x,\gamma)$ and $\mathcal{P}_3(s,x,\gamma)$ combined together allow us now to evaluate
\begin{equation}\label{P23diff}
	\int\limits_0^{\gamma}\big(\mathcal{P}_2(s,x,t)-\mathcal{P}_3(s,x,t)\big)dt.
\end{equation}
For this evaluation it is important to recall the definitions of $c_0(\nu)$ and $c_1(\nu)$ 
\begin{equation*}
	c_0(\nu)=-\frac{1}{\Gamma(\nu)}\big(\psi(\nu)+2\gamma_E\big),\ \ c_1(\nu)=-\frac{1}{\Gamma(\nu)}
\end{equation*}
as well as the functional equation of the Digamma function (see e.g. \cite{BE})
\begin{equation*}
	\psi(z) = \psi(z+1)-\frac{1}{z}=\psi(1-z)-\pi\cot\pi z,\ \ z\in\mathbb{C}\backslash\{0,-1,-2,\ldots\}.
\end{equation*}
It implies
\begin{equation*}
	c_0(\nu)c_1(-\nu)+c_1(\nu)c_0(-\nu)-c_0(1+\nu)\frac{\Gamma(1+\nu)}{\Gamma(-\nu)}c_1(\nu)-c_1(-\nu)c_0(1-\nu)\frac{\Gamma(1-\nu)}{\Gamma(\nu)} = 0
\end{equation*}
and shows therefore that all terms of $O\big((\ln s)^2\big)$ in \eqref{P23diff} vanish. The remaining terms of $O(\ln s)$ and $O(1)$ can be computed in a similar way, we obtain
\begin{eqnarray}\label{P23asy}
	&&\int_0^{\gamma}\Big(\mathcal{P}_2(s,x,t)-\mathcal{P}_3(s,x,t)\Big)dt=6(i\nu)^2\ln s+8(i\nu)^2\ln 2\\
	&& +2\int\limits_0^{\gamma}\nu(t)\big(\ln\Gamma(\nu(t))-\ln\Gamma(-\nu(t))\big)_tdt+O\bigg(\frac{(\ln s)^3}{s}\bigg)\nonumber
\end{eqnarray}
as $s\rightarrow\infty$ uniformly on any compact subset of the set \eqref{excset1}. The latter statement combined with \eqref{firstsummary} implies \eqref{theo1result} with an error term of 
\begin{equation*}
	O\bigg(\frac{(\ln s)^3}{s}\bigg).
\end{equation*}
In order to improve this error term, we use Proposition \ref{prop1} and in particular the identities for $R(\pm s,\pm s)$ derived in section \ref{sec16}. With \eqref{Rplusminus1} after simplification
\begin{equation*}
	R(s,s) = -i\nu(8s^2+2x)-\frac{3(i\nu)^2}{s}+O\big(s^{-2}\big)
\end{equation*}
and also
\begin{equation*}
	R(-s,-s) = -i\nu(8s^2+2x) -\frac{3(i\nu)^2}{s}+O\big(s^{-2}\big)
\end{equation*}
which implies via \eqref{sidentity}
\begin{equation*}
	\frac{\partial}{\partial s}\ln\det(I-\gamma K_{\textnormal{csin}}) = i\nu(16s^3+4x)+\frac{6(i\nu)^2}{s}+O\big(s^{-2}\big),\ \ s\rightarrow\infty
\end{equation*}
uniformly on any compact subset of the set \eqref{excset1}. Integrating the latter equation with respect to $s$ and comparing with \eqref{P23asy}, we complete the proof of Theorem \ref{theo1}.

.


\section{Asymptotics of $\ln\det(I-\gamma K_{\textnormal{cin}})$ - proof of theorem \ref{theo2}}\label{sec18}

The statement on the asymptotic distribution of the zeros of $\det(I-\gamma K_{\textnormal{csin}}),\gamma>1$ will follow from an asymptotic expansion for the latter quantity which we derive from Proposition \ref{prop1} and \ref{prop2}. To this end we trace back the transformations
\begin{equation*}
	X(\lambda)\mapsto T(z)\mapsto S(z)\mapsto R(z)\mapsto P(z) \mapsto Q(z)
\end{equation*}
and use the identities \eqref{RSSplus},\eqref{RSSminus}, which were derived independently of the choice of $\gamma\in\mathbb{R}$. First the residue conditions \eqref{res1} and \eqref{res2} show explicitly that $R(z)$ is analytic at $z=1$
\begin{equation}\label{R1sing}
	R(1) = \Big(Q(1)+(I+B)Q'(1)\Big)\begin{pmatrix}
	1 & 0\\
	p & 0\\
	\end{pmatrix} +(I+B)Q(1)\begin{pmatrix}
	0 & 0\\
	\nu_0 p\frac{3+\frac{x}{4s^2}}{1+\frac{x}{4s^2}} & \frac{1}{2}\\
	\end{pmatrix}
\end{equation}
where $\nu_0$ and $p$ were introduced earlier as
\begin{equation*}
	\nu_0 = \frac{1}{2\pi i}\ln|1-\gamma| = \frac{1}{2\pi i}\ln(\gamma-1),\hspace{0.5cm} p = i\frac{\Gamma(1-\nu)}{\Gamma(\nu)}e^{-2s^3\vartheta(1)}\hat{\beta}_r^2(1).
\end{equation*}
Furthermore
\begin{eqnarray}\label{Rprime1}
	R'(1) &=& \big(Q(1)+(I+B)Q'(1)\big)\begin{pmatrix}
	0 & 0 \\
	\nu_0p\frac{3+\frac{x}{4s^2}}{1+\frac{x}{4s^2}} & \frac{1}{2}\\
	\end{pmatrix} +\bigg(Q'(1)+(I+B)\frac{Q''(1)}{2}\bigg)\begin{pmatrix}
	1 & 0 \\
	p & 0\\
	\end{pmatrix}\nonumber\\
	&&+(I+B)Q(1)\begin{pmatrix}
	0 & 0 \\
	\nu_0p\kappa(s,x) & -\frac{1}{4}\\
	\end{pmatrix}
\end{eqnarray}
where
\begin{equation*}
	\kappa(s,x) = \frac{1}{2}\bigg(\frac{10}{3}\Big(1+\frac{x}{4s^2}\Big)+\Big(\nu_0-\frac{1}{2}\Big)\Big(3+\frac{x}{4s^2}\Big)^2\bigg)\Big(1+\frac{x}{4s^2}\Big)^{-2}.
\end{equation*}
Similarly
\begin{equation}\label{R-1sing}
	R(-1)=\Big(Q(-1)-(I-B)Q'(-1)\Big)\begin{pmatrix}
	0 & p\\
	0 & 1\\
	\end{pmatrix}-(I-B)Q(-1)\begin{pmatrix}
	-\frac{1}{2} & -\nu_0p\frac{3+\frac{x}{4s^2}}{1+\frac{x}{4s^2}}\\
	0 & 0 \\
	\end{pmatrix}
\end{equation}
and
\begin{eqnarray}\label{chidprime-1}
	&&R'(-1) =\Big(Q(-1)-(I-B)Q'(-1)\Big)\begin{pmatrix}
	-\frac{1}{2} & -\nu_0p\frac{3+\frac{x}{4s^2}}{1+\frac{x}{4s^2}}\\
	0 & 0 \\
	\end{pmatrix}\\
	&&+\bigg(Q'(-1)-(I-B)\frac{Q''(-1)}{2}\bigg)\begin{pmatrix}
	0 & p\\
	0 & 1\\
	\end{pmatrix}-(I-B)Q(-1)\begin{pmatrix}
	-\frac{1}{4} & \nu_0p\kappa(s,x)\\
	0 & 0\\
	\end{pmatrix}.\nonumber
\end{eqnarray}
To evaluate $Q(\pm 1)$, we iterate. First for any $z\in\Sigma_R$ from \eqref{esti7}, \eqref{esti8} and residue theorem
\begin{eqnarray*}
	Q_-(z)&=&I+\frac{1}{2\pi i}\int\limits_{\Sigma_R}Q_-(w)\big(G_Q(w)-I\big)\frac{dw}{w-z_-}\\
	&=& I+\frac{i}{2sz}\Bigg(\begin{pmatrix}
	v & -u\\
	u & -v\\
	\end{pmatrix}
	-\begin{pmatrix}
	\frac{1}{z-1} & 0\\
	0 & \frac{1}{z+1}\\
	\end{pmatrix}\big(B_0(z)\big)^{-1}\\
	&&\times\begin{pmatrix}
	v & ue^{-2\pi i\nu}\\
	-ue^{2\pi i\nu} & -v\\
	\end{pmatrix}B_0(z)\begin{pmatrix}
	z-1 & 0\\
	0 & z+1\\
	\end{pmatrix}\Bigg) +O\big(s^{-2}\big),
\end{eqnarray*}
where we used that
\begin{eqnarray*}
	&&\begin{pmatrix}
	\frac{1}{w-1} & 0 \\
	0 & \frac{1}{w+1}\\
	\end{pmatrix}\big(B_0(w)\big)^{-1}\begin{pmatrix}
	v & ue^{-2\pi i\nu}\\
	-ue^{2\pi i\nu} & -v\\
	\end{pmatrix}B_0(w)\begin{pmatrix}
	w-1 & 0\\
	0 & w+1\\
	\end{pmatrix}\\
	&=&\begin{pmatrix}
	v & -u\\
	u & -v\\
	\end{pmatrix} +4w\begin{pmatrix}
		0 & u\\
		u & 0\\
		\end{pmatrix}\nu_0 +O\big(w^2\big),\hspace{0.5cm}w\rightarrow 0.
\end{eqnarray*}
Thus
\begin{eqnarray*}
	&&Q(\pm 1) = I+\frac{1}{2\pi i}\int\limits_{\Sigma_R}\big(Q_-(w)-I\big)\big(G_Q(w)-I\big)\frac{dw}{w\mp 1} +\frac{1}{2\pi i}\int\limits_{\Sigma_R}\big(G_Q(w)-I\big)\frac{dw}{w\mp 1}\\
	&=&I\pm\frac{i}{2s}\begin{pmatrix}
	v & -u\\
	u & -v\\
	\end{pmatrix}\mp\frac{\nu_0}{s^2}\begin{pmatrix}
	-u^2 & -u_x\\
	u_x & u^2\\
	\end{pmatrix}
	+\frac{1}{8s^2}\begin{pmatrix}
	u^2-v^2 & -2(u_x+uv)\\
	-2(u_x+uv) & u^2-v^2\\
	\end{pmatrix}\\
	&&+O\big(s^{-3}\big),\hspace{0.5cm} s\rightarrow\infty,
\end{eqnarray*}
uniformly on any compact subset of the set \eqref{excset2}, and similarly 
\begin{eqnarray*}
	Q'(\pm 1) &=&-\frac{i}{2s}\begin{pmatrix}
	v & -u\\
	u & -v\\
	\end{pmatrix}+\frac{\nu_0}{s^2}\begin{pmatrix}
	-u^2 & -u_x\\
	u_x & u^2\\
	\end{pmatrix}\mp\frac{1}{4s^2}\begin{pmatrix}
	u^2-v^2 & -2(u_x+uv)\\
	-2(u_x+uv) & u^2-v^2\\
	\end{pmatrix}\\
	&&+O\big(s^{-3}\big).
\end{eqnarray*}
Since
\begin{equation*}
	p = i\frac{\Gamma(1-\nu)}{\Gamma(\nu)}e^{-2s^3\vartheta(1)}\hat{\beta}_r^2(1) = ie^{-i\sigma},\hspace{0.5cm} \bar{p}=p^{-1}
\end{equation*}
with $\sigma =\sigma(s,x,\gamma)$ as in \eqref{sigmadef}, we obtain for the matrix
\begin{equation*}
	\Theta = \Bigg(Q(1)\binom{1}{p},Q(-1)\binom{p}{1}\Bigg),
\end{equation*}
which appears in \eqref{Bmatrix}, that
\begin{eqnarray}\label{bdeterminant}	
\det \Theta&=&-2ip\Big(\cos\sigma-\frac{v}{s}\sin\sigma+\frac{u}{s}\nonumber\\
&&\hspace{1.5cm}+\frac{2i\nu_0}{s^2}\big(u_x+u^2\sin\sigma\big)+\frac{u^2-v^2}{2s^2}\cos\sigma+O\big(s^{-3}\big)\Big),
\end{eqnarray}
Since we agreeded that $s$ stays away from the small neigbhorhood of the points $\{s_n\}$ defined by
\begin{equation}\label{exceptset}
	\cos\sigma(s_n,x,\gamma)=0
\end{equation}
we see that for all sufficiently large $s$ lying outside of the zero set of the latter transcendental equation, the stated determinant is non zero. Back to \eqref{Bmatrix}, this implies
\begin{eqnarray*}
	B&=&-\frac{2p}{\det \Theta}\Bigg[\begin{pmatrix}
	\sin\sigma & -1\\
	1 & -\sin\sigma\\
	\end{pmatrix}+\frac{\cos\sigma}{s}\begin{pmatrix}
	v & -u\\
	u & -v\\
	\end{pmatrix}\\
&&+\frac{1}{s^2}\bigg\{2i\nu_0\cos\sigma\begin{pmatrix}
	-u^2 & -u_x\\
	u_x & u^2\\
	\end{pmatrix}+\frac{1}{2}\begin{pmatrix}
	-v^2\sin\sigma &-u^2\\
	u^2 & v^2\sin\sigma\\
	\end{pmatrix}+\frac{u_x}{2}\begin{pmatrix}
	-1 & \sin\sigma\\
	-\sin\sigma & 1\\
	\end{pmatrix}\\
	&&+uv\sin\sigma\begin{pmatrix}
	0 & 1\\
	-1 & 0 \\
	\end{pmatrix}\bigg\}+O\big(s^{-3}\big)\Bigg],\hspace{0.5cm}s\rightarrow\infty
\end{eqnarray*}
and with
\begin{eqnarray*}
	-\frac{2p}{\det \Theta}&=&-\frac{i}{\cos\sigma}\Bigg[1+\frac{1}{s}\bigg(v\tan\sigma -\frac{u}{\cos\sigma}\bigg)+\frac{1}{s^2}\bigg(v^2\tan^2\sigma-\frac{2uv\sin\sigma}{\cos^2\sigma}+\frac{u^2}{\cos^2\sigma}\\
	&&-\frac{u^2-v^2}{2}-2i\nu_0u^2\tan\sigma-\frac{2i\nu_0 u_x}{\cos\sigma}\bigg)+O\big(s^{-3}\big)\Bigg]
\end{eqnarray*}
we obtain in turn
\begin{eqnarray*}
	B&=&-i\Bigg[\frac{1}{\cos\sigma}\begin{pmatrix}
	\sin\sigma & -1\\
	1 & -\sin\sigma\\
	\end{pmatrix}+\frac{1}{s}\bigg\{\frac{v\sin\sigma-u}{\cos^2\sigma}\begin{pmatrix}
	\sin\sigma & -1\\
	1 & -\sin\sigma\\
	\end{pmatrix}+\begin{pmatrix}
	v & -u\\
	u & -v\\
	\end{pmatrix}\bigg\}\\
	&&+\frac{1}{s^2}\bigg\{\frac{1}{2\cos\sigma}\begin{pmatrix}
	2v^2\sin\sigma-u^2\sin\sigma-u_x& -v^2+u_x\sin\sigma\\
	v^2-u_x\sin\sigma & -2v^2\sin\sigma+u^2\sin\sigma+u_x\\
	\end{pmatrix}\\
	&&-\frac{2i\nu_0(u^2\sin\sigma+u_x)}{\cos^2\sigma}\begin{pmatrix}
	\sin\sigma & -1\\
	1 & -\sin\sigma\\
	\end{pmatrix}+\frac{(v\sin\sigma-u)^2}{\cos^3\sigma}\begin{pmatrix}
	\sin\sigma & -1\\
	1 & -\sin\sigma\\
	\end{pmatrix}\\
	&&+\begin{pmatrix}
	-uv\cos\sigma -2i\nu_0u^2 & u^2\cos\sigma-2i\nu_0u_x\\
	-u^2\cos\sigma +2i\nu_0u_x & uv\cos\sigma+2i\nu_0u^2\\
	\end{pmatrix}\bigg\}+O\big(s^{-3}\big)\Bigg],\hspace{0.5cm}s\rightarrow\infty
\end{eqnarray*}
and all expansions are uniformly on any compact subset of the set \eqref{excset2}. Let us go back to \eqref{RSSplus} and \eqref{RSSminus}. Since $\nu = \nu_0+\frac{1}{2}$, we notice
\begin{eqnarray*}
	R(s,s)&=& -i\nu_0(8s^2+2x) -i(4s^2+x)\\
	&&+ip(16s^2+4x)\Big[R_{12}'(1)R_{22}(1)-R_{22}'(1)R_{12}(1)\Big]+O\big(s^{-1}\big)
\end{eqnarray*}
and similarly
\begin{eqnarray*}
	R(-s,-s)&=&-i\nu_0(8s^2+2x)-i(4s^2+x)+ip(16s^2+4x)\\
	&&\times\Big[R_{11}'(-1)R_{21}(-1)-R_{21}'(-1)R_{11}(-1)\Big]+O\big(s^{-1}\big)
\end{eqnarray*}
Next
\begin{eqnarray*}
	R_{12}'(1) &=& \frac{1}{2}\big(Q(1)+(I+B)Q'(1)\big)_{12}-\frac{1}{4}\big((I+B)Q(1)\big)_{12}\\
	R_{22}(1) &=&\frac{1}{2}\big((I+B)Q(1)\big)_{22}\\
	R_{22}'(1)&=& \frac{1}{2}\big(Q(1)+(I+B)Q'(1)\big)_{22}-\frac{1}{4}\big((I+B)Q(1)\big)_{22}\\
	R_{12}(1) &=&\frac{1}{2}\big((I+B)Q(1)\big)_{12}
\end{eqnarray*}
and therefore
\begin{eqnarray*}
	&&R'_{12}(1)R_{22}(1)-R_{22}'(1)R_{12}(1) = \frac{1}{4}\Big[\big(Q(1)+(I+B)Q'(1)\big)_{12}\big((I+B)Q(1)\big)_{22}\\
	&&\hspace{0.5cm}-\big(Q(1)+(I+B)Q'(1)\big)_{22}\big((I+B)Q(1)\big)_{12}\Big].
\end{eqnarray*}
Now combine the previously derived information on $Q(1),Q'(1)$ and $B$ to derive
\begin{equation*}
	ip(16s^2+4x)\Big[R'_{12}(1)R_{22}(1)-R_{22}'(1)R_{12}(1)\Big]\\
	=(4s^2+x)\big(i+\tan\sigma\big)+\alpha_++O\big(s^{-1}\big)
\end{equation*}
with a function $\alpha_+=\alpha_+(s,x,\gamma)$ such that
\begin{equation*}
	\int\alpha_+(s,x,\gamma)ds = O(\ln s),\ \ s\rightarrow\infty.
\end{equation*}
Following the same computations for $R(-s,-s)$
\begin{equation}\label{r-s-sfinal}
	R(-s,-s)=-i\nu_0(8s^2+2x)+(4s^2+x)\tan\sigma +\alpha_-+O\big(s^{-1}\big).
\end{equation}
where
\begin{equation*}
	\int\alpha_-(s,x,\gamma)ds = O(\ln s)
\end{equation*}
and together from Proposition \ref{prop1}
\begin{equation}\label{Sderivsing}
\frac{\partial}{\partial s}\ln\det(I-\gamma K_{\textnormal{csin}}) = i\nu_0(16s^2+4x)-(8s^2+2x)\tan \sigma -\alpha_+-\alpha_- +O\big(s^{-1}\big).
\end{equation}
Opposed to the latter equation we now recall Proposition \ref{prop2} and evaluate the logarithmic $x$-derivative. For $\gamma>1$
\begin{eqnarray*}
	X_1&=&\lim_{\lambda\rightarrow\infty}\Big(\lambda\big(X(\lambda)e^{-i(\frac{4}{3}\lambda^3+x\lambda)\sigma_3}-I\big)\Big)\\
		&=&-2\nu s\sigma_3+s(\sigma_3+B)+\frac{is}{2\pi}\int\limits_{\Sigma_R}Q_-(w)(w)\big(G_Q(w)-I\big)dw
\end{eqnarray*}
where the expansion for $B$ has already been computed. From this and residue theorem
\begin{eqnarray*}
	X_1&=& -2\nu s\sigma_3+s\sigma_3-\frac{is}{\cos\sigma}\begin{pmatrix}
	\sin\sigma & -1\\
	1 & -\sin\sigma\\
	\end{pmatrix}-\frac{i}{2}\begin{pmatrix}
	v & -u\\
	u & -v\\
	\end{pmatrix}\\
	&&-\frac{i(v\sin\sigma-u)}{\cos^2\sigma}\begin{pmatrix}
	\sin\sigma & -1\\
	1 & -\sin\sigma\\
	\end{pmatrix}+O\big(s^{-1}\big),
\end{eqnarray*}
hence
\begin{equation}\label{xsingfinal}
	\frac{\partial}{\partial x}\ln\det(I-\gamma K_{\textnormal{csin}})=4i\nu_0s+v-2s\tan\sigma-\frac{2v}{\cos^2\sigma}+\frac{2u\sin\sigma}{\cos^2\sigma}+O\big(s^{-1}\big).
\end{equation}
Integrating both identities \eqref{Sderivsing}, \eqref{xsingfinal} and comparing the result, we conclude for $s\rightarrow\infty$ outside the zero set \eqref{exceptset}
\begin{eqnarray}\label{final}
	\ln\det(I-\gamma K_{\textnormal{csin}})&=&i\nu_0\bigg(\frac{16}{3}s^3+4sx\bigg)+\ln|\cos\sigma(s,x,\gamma)|+\chi_1\ln s\nonumber\\
	&&-\int\limits_x^{\infty}(y-x)u^2(y,\gamma)dy +\chi_2+O\big(s^{-1}\big),
\end{eqnarray}
with real-valued constants $\chi_i$, solely depending on $\gamma$ and the error term is uniform on any compact subset of the set \eqref{excset2}. The given expansion \eqref{final} verifies the claim on the asymptotic distribution of the zeros of the Fredholm determinant as given in Theorem \ref{theo2}.

\begin{remark} We want to emphasize that our proof of Theorem \ref{theo2} produces in fact an entire asymptotic series for $\det(I-\gamma K_{\textnormal{csin}})$ in case $\gamma>1$, similarly to what we obtained in Theorem \ref{theo1}. However due to the increased amount of computations within the Riemann-Hilbert analysis, we chose not to compute the constant $\chi_1(\gamma)$.
\end{remark}

\end{document}